\documentclass[11pt]{article}

\usepackage{epsfig}
\usepackage{amssymb}
\usepackage{natbib}
\usepackage{hyperref}
\usepackage{enumitem}
\usepackage{graphicx}
\usepackage{thumbpdf}
\usepackage{amsfonts,amstext,amsmath,amsthm}
\usepackage{mathtools}
\usepackage{dsfont}
\usepackage{accents}
\usepackage[dvipsnames]{xcolor}
\usepackage{rotating}
\usepackage[utf8]{inputenc}
\usepackage{booktabs}
\usepackage{tikz}
\usepackage{multirow}
\usepackage{float}
\usepackage{lscape}
\usepackage{algorithm}
\usepackage{algpseudocode}
\usepackage{fullpage}
\usepackage[labelfont=bf, width=\textwidth, font=small]{caption}
\newcommand*{\doi}[1]{\href{http://dx.doi.org/#1}{doi:#1}}

\graphicspath{{Figures/}}

\newcommand{\gComment}[1]{\textcolor{gray}{\Comment{#1}}}

\renewcommand{\hat}{\widehat}

\newcommand{\dd}{\,\mathrm{d}}
\newcommand{\1}{\mathds{1}}

\theoremstyle{definition}
\newtheorem{proposition}{Proposition}
\newtheorem{assumption}{Assumption}
\newtheorem*{proposition*}{Proposition}
\newtheorem{definition}[proposition]{Definition}

\newtheorem{example}[proposition]{Example}
\newtheorem{theorem}[proposition]{Theorem}
\newtheorem{lemma}[proposition]{Lemma}

\newtheorem{setting}{Setting}

\DeclareMathOperator{\pa}{pa}
\DeclareMathOperator{\an}{an}
\DeclareMathOperator{\ch}{ch}

\newcommand{\trafos}[1]{\calH_{\calY,{#1}}}
\newcommand{\trafoall}{\trafos{\calX}^{*}}
\newcommand{\trafosub}{\trafos{\calX}}
\newcommand{\trafosubS}{\trafos{\calX^S}}
\newcommand{\trafosubPA}{\trafos{\calX^{S_*}}}
\newcommand{\trafoadd}{\trafos{\calX}^{\text{shift}}}
\newcommand{\trafolin}{\trafos{\calX}^{\text{linear}}}

\newcommand{\trafowald}{\trafos{\calX\times\calE}^{\text{Wald}}}
\newcommand{\trafoallS}{\trafos{\calX^S}^{*}}

\newcommand{\trafowaldS}{\trafos{\calX^S\times\calE}^{\text{Wald}}}
\newcommand{\bltrafo}{\calH_\calY}

\usepackage{accents}

\DeclareMathOperator{\supp}{supp}

\newcommand{\norm}[1]{\left\lVert#1\right\rVert}

\newcommand{\rZ}{Z}
\newcommand{\rY}{Y}
\newcommand{\rX}{\mX}
\newcommand{\rS}{\mS}
\newcommand{\rL}{\mL}
\newcommand{\rT}{\mT}

\newcommand{\rz}{z}
\newcommand{\ry}{y}
\newcommand{\rx}{\xvec}

\newcommand{\pZ}{F_\rZ}

\newcommand{\pSL}{F_{\SL}}
\newcommand{\pMEV}{F_{\text{minEV}}}

\newcommand{\pYx}{F_{\rY \mid \rX = \rx}}

\newcommand{\dZ}{f_\rZ}

\newcommand{\dYx}{f_{\rY | \rX = \rx}}

\newcommand{\h}{h}

\newcommand{\hY}{h_\rY}

\newcommand{\basisy}{\avec}
\newcommand{\bern}[1]{\avec_{\text{Bs},#1}}

\newcommand{\basisx}{\bvec}

\newcommand{\parm}{\varthetavec}
\newcommand{\eparm}{\vartheta}

\newcommand{\shiftparm}{\betavec}

\newcommand{\ie}{{i.e.,}~}
\newcommand{\eg}{{e.g.,}~}
\newcommand{\cf}{{cf.}~}
\newcommand{\wrt}{{w.r.t.}~}
\newcommand{\st}{{s.t.}~}

\newcommand{\Prob}{\mathbb{P}}
\newcommand{\Ex}{\mathbb{E}}
\newcommand{\RR}{\mathbb{R}}

\usepackage{dsfont}

 \DeclareMathOperator{\logit}{logit}
 \DeclareMathOperator{\expit}{expit}

 \DeclareMathOperator{\Cov}{Cov}
 
 \DeclareMathOperator{\Var}{Var}

 \DeclareMathOperator{\rank}{rank}

 \DeclareMathOperator*{\argmax}{{arg\,max}}

 \DeclareMathOperator{\ND}{N}

 \DeclareMathOperator{\BD}{Bernoulli}

 \DeclareMathOperator{\SL}{SL}

\def \avec {\text{\boldmath$a$}}    
\def \bvec {\text{\boldmath$b$}}    \def \mB {\text{$\mathbf B$}}
    
    \def \mD {\text{$\mathbf D$}}
\def \evec {\text{\boldmath$e$}}

\def \tvec {\text{\boldmath$t$}}

\def \xvec {\text{\boldmath$x$}}    \def \mX {\text{$\mathbf X$}}
\def \yvec {\text{\boldmath$y$}}

\def \rE {\text{\boldmath$E$}}

\def \rI {\text{\boldmath$I$}}

\def \rL {\text{\boldmath$L$}}

\def \rS {\text{\boldmath$S$}}
\def \rT {\text{\boldmath$T$}}

\def \rV {\text{\boldmath$V$}}

\def \rX {\text{\boldmath$X$}}

 \def \calC {\mathcal C}
 \def \calD {\mathcal D}
 \def \calE {\mathcal E}
 \def \calF {\mathcal F}
 \def \calG {\mathcal G}
 \def \calH {\mathcal H}

 \def \calM {\mathcal M}

 \def \calP {\mathcal P}

 \def \calS {\mathcal S}

 \def \calW {\mathcal W}
 \def \calX {\mathcal X}
 \def \calY {\mathcal Y}
 \def \calZ {\mathcal Z}

\def \betavec         {\text{\boldmath$\beta$}}
\def \gammavec        {\text{\boldmath$\gamma$}}

\def \thetavec        {\text{\boldmath$\theta$}}
\def \varthetavec     {\text{\boldmath$\vartheta$}}

\def \muvec           {\text{\boldmath$\mu$}}

\def \xivec           {\text{\boldmath$\xi$}}

\newcommand{\ubar}[1]{\underaccent{\bar}{#1}}

\newcommand{\pkg}[1]{\textbf{#1}}

\newcommand{\proglang}[1]{\textsf{#1}}
\newcommand{\code}[1]{\texttt{#1}}

\DeclareMathOperator{\Id}{Id}

\usepackage{tikz}

\usetikzlibrary{positioning, calc, shapes.geometric, shapes.multipart,
  shapes, arrows.meta, arrows,
  decorations.markings, external, trees}
\tikzstyle{line} = [draw, -latex']
\tikzstyle{Arrow} = [
        thick,
        decoration={
                markings,
                mark=at position 1 with {
                        \arrow[thick]{latex}
} },
        shorten >= 3pt, preaction = {decorate}
        ]

\newcommand{\RN}[1]{%
(\textup{\uppercase\expandafter{\romannumeral#1}})%
}

\newcommand{\indep}{\perp\nolinebreak\hspace{-6pt}\perp}
\mathtoolsset{showonlyrefs=true}

\setcitestyle{authoryear,open={(},close={)}}

\newcommand{\tram}{\textsc{tram}}
\newcommand{\tramicp}{\textsc{tramicp}}

\newcommand{\bcd}{\boldsymbol\cdot}

\begin{document}

\title{\bf Model-based causal feature selection for\\general response types}
\author{Lucas Kook\textsuperscript{1}\thanks{Email: \code{lucasheinrich.kook@gmail.com} $\bcd$
Preprint $\bcd$ Version July 8, 2024 $\bcd$ Licensed under CC-BY.},
Sorawit Saengkyongam\textsuperscript{2},
Anton Rask Lundborg\textsuperscript{3},\\
Torsten Hothorn\textsuperscript{4},
Jonas Peters\textsuperscript{2}}

\date{%
\small
\textsuperscript{1}Institute for Statistics and Mathematics, Vienna University
of Economics and Business\linebreak
\textsuperscript{2}Seminar for Statistics, ETH Zurich\linebreak
\textsuperscript{3}Department of Mathematical Sciences, University of Copenhagen\linebreak
\textsuperscript{4}Epidemiology, Biostatistics \& Prevention Institute, University of Zurich
}

\maketitle

\begin{abstract}%
Discovering causal relationships from observational data is a fundamental yet
challenging task. Invariant causal prediction \citep[ICP,][]{peters2016causal}
is a method for causal feature selection which requires data from heterogeneous
settings and exploits that causal models are invariant. ICP has been extended to
general additive noise models and to nonparametric settings using conditional
independence tests. However, the latter often suffer from low power (or poor
type I error control) and additive noise models are not suitable for
applications in which the response is not measured on a continuous scale, but
reflects categories or counts. Here, we develop transformation-model (\tram)
based ICP, allowing for continuous, categorical, count-type, and uninformatively
censored responses (these model classes, generally, do not allow for
identifiability when there is no exogenous heterogeneity). As an invariance
test, we propose \tram-GCM based on the expected conditional covariance between
environments and score residuals with uniform asymptotic level guarantees. For
the special case of linear shift \tram{s}, we also consider \tram-Wald, which
tests invariance based on the Wald statistic. We provide an open-source
\proglang{R}~package \pkg{tramicp} and evaluate our approach on simulated data
and in a case study investigating causal features of survival in critically ill
patients.
\end{abstract}

\section{Introduction} \label{sec:intro}

\subsection{Motivation}\label{sec:motivation}

Establishing causal relationships from observational data is a common goal in
several scientific disciplines. However, systems are often too complex to allow
for recovery of the full causal structure underlying the data-generating
process. In this work, we consider the easier task of uncovering the causal
drivers of a particular response variable of interest. We present methods,
theoretical results and user-friendly software for model-based causal feature
selection, where the response may represent a binary, ordered, count, or
continuous outcome and may additionally be uninformatively censored. We propose
\tramicp{} for causal feature selection, which is based on invariant causal
prediction \citep[ICP,][]{peters2016causal} and a flexible class of regression
models, called transformation models \citep[\tram{s},][]{hothorn2018most}.
\tramicp{} relies on data from heterogeneous environments and the assumption,
that the causal mechanism of the response given its direct causes (direct \wrt
the considered sets of covariates) is correctly specified by a \tram{} and does
not change across those environments
\citep{haavelmo1943statistical,frisch1948,aldrich1989autonomy,pearl2009causality,schoelkopf2012anti}.
The causal \tram{} will then produce score residuals (residuals defined
specifically for \tram{s} and potentially censored observations) that are
invariant across the environments.
If this assumption is violated (for instance, if the environment, which is not
included as a covariate, directly impacts the response) but faithfulness
\citep[p.~56]{spirtes2000causation} holds, \tramicp{} is conservative and will
produce an uninformative output.
We propose an invariance test based on the expected conditional covariance
between the score residuals and the environments given a subset $S$ of the
covariates, called \tram-GCM. With this invariance test, \tramicp{} recovers a
subset of the direct causes with high probability, by fitting a \tram{} for all
subsets of covariates, computing score residuals, testing whether those score
residuals are uncorrelated with the residualized environments and lastly,
intersecting all subsets for which the null hypothesis of invariance was not
rejected. For the special case of additive linear models, we propose another
invariance test, \tram-Wald, based on the Wald statistic for testing whether
main and interaction effects involving the environments are zero.

We illustrate the core ideas of \tramicp{} in the following example with a
binary response and the logistic regression model
\citep{mccullagh2019generalized}, which is a \tram. We defer all details on how
\tram{s} and score residuals are defined to Section~\ref{sec:bg} and describe
the \tram-GCM and \tram-Wald invariance tests in Section~\ref{sec:methods} and
Appendix~\ref{sec:wald}, respectively.

\begin{example}[Invariance in binary generalized linear models]\label{ex:intro}
Consider the following structural causal model \citep{pearl2009causality} over
$(Y, X^1, X^2, E)$:\\[0pt]
\begin{minipage}{0.49\textwidth}
\begin{align}
    E &\coloneqq N_E\\
    X^1 &\coloneqq -E + N_1\\
    Y &\coloneqq \1(0.5 X^1 > N_Y)\\
    X^2 &\coloneqq Y + 0.8 E + N_2,
\end{align}
\end{minipage}
\begin{minipage}{0.49\textwidth}
\centering
\begin{tikzpicture}[node distance=1.0cm, <-> /.tip = Latex, -> /.tip = Latex,
    thick, roundnode/.style={circle, draw, inner sep=1pt,minimum size=7mm},
    squarenode/.style={rectangle, draw, inner sep=1pt, minimum size=7mm}]
    \node[roundnode] (X) {$X^1$};
    \node[right=of X, roundnode] (Y) {$Y$};
    \node[below=of Y, roundnode] (E) {$E$};
    \node[right=of Y, roundnode] (X2) {$X^2$};
    \draw[->] (E) edge[bend left=00] (X);
    \draw[->] (E) edge[bend left=00] (X2);
    \draw[->] (X) edge[bend left=00] (Y);
    \draw[->] (Y) -- (X2);
\end{tikzpicture}
\end{minipage}\\[12pt]
where $N_E \sim \BD(0.5)$, $N_1 \sim \ND(0, 1)$, $N_2 \sim \ND(0, 1)$, $N_Y$ are
jointly independent noise variables and $N_Y$ follows a standard logistic
distribution. Here, $E$ encodes two environments in which the distribution of
$X^1$ and $X^2$ differ, but the causal mechanism of $Y$ given its direct causes
$X^1$ does not change.

Let us assume that both the above structural causal model and its implied
structure are unknown and that we observe an i.i.d.\ sample
$\{(e_i, x^1_i, x^2_i, y_i)\}_{i=1}^n$ from the joint distribution of $(E, X^1,
X^2, Y)$. We further know that $Y$ given its direct causes is correctly
specified by a logistic regression. All remaining conditionals do not need to
satisfy any model assumptions. Our task is now to infer (a subset of) the direct
causes of $Y$.

To do so, for each subset of the covariates $\rX^S$, $S \subseteq \{1, 2\}$ (\ie
for $\emptyset$, $\{1\}$, $\{2\}$ and $\{1, 2\}$), we now (i) fit a binary
logistic regression model, (ii) compute the score residuals $y_i - \hat{\Prob}(Y
= 1 \mid \rX^S = \rx^S_i)$ (from the logistic regression) and residualized
environments $e_i - \hat{\Prob}(E = 1 \mid \rX^S = \rx^S_i)$ (via a random
forest), and (iii) test whether the two residuals are correlated.
Figure~\ref{fig:intro} shows the residuals obtained in step (iii) for each
non-empty subset of the covariates.
\begin{figure}[!t]
    \centering
    \includegraphics[width=0.8\textwidth]{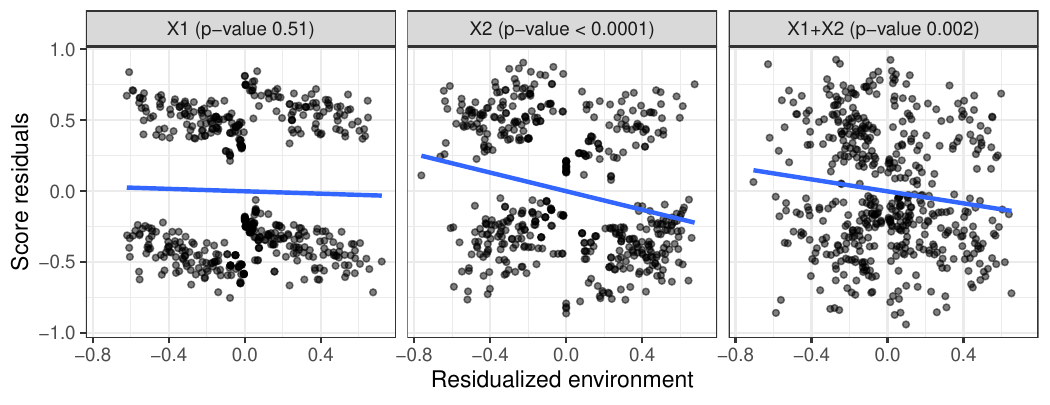}
    \caption{%
    Invariance in binary generalized linear models. By the data
    generating mechanism in Example~\ref{ex:intro}, we know that the
    conditional distribution of $Y$ given its direct cause $X^1$ does not
    change across the two environments $E = 0$ and $E = 1$. When predicting
    both $Y$ and $E$ from the three sets of covariates $\{1\}$, $\{2\}$
    and $\{1, 2\}$, the resulting residuals are uncorrelated only when
    conditioning on the invariant set $\{1\}$. The $p$-values of the
    invariance test we introduce in Section~\ref{sec:scoretest} are
    shown in the panel strips for the corresponding subset of covariates
    (we have also added linear model fits, see blue lines).
    The empty set is omitted, since the
    score residuals and residualized environments only take two values.
    }\label{fig:intro}
\end{figure}

In this example, even though the model using $\{X^1, X^2\}$ achieves higher
predictive accuracy than the model using the causal parent $\{X^1\}$, only
the model $Y \mid X^1$ is stable across the environments. If more than
one set is invariant, one can take the intersection of the invariant sets
to obtain a subset of the direct causes of $Y$
\citep{peters2016causal}.

With our openly available \proglang{R}~package \pkg{tramicp}
(\url{https://CRAN.R-project.org/package=tramicp}), the analysis in this example
can be reproduced with the following code, where \code{df} is a data frame with
500 independent observations from the structural causal
model above.
\begin{verbatim}
R> library("tramicp")
R> icp <- glmICP(Y ~ X1 + X2, data = df, env = ~ E, family = "binomial")
R> pvalues(icp, which = "set")
   Empty       X1       X2    X1+X2
1.82e-02 5.10e-01 4.54e-09 2.22e-03
\end{verbatim}
\end{example}

\subsection{Related work}\label{sec:related}

Several algorithms exist to tackle the problem of causal discovery, \ie
learning the causal graph from data, including constraint-based and score-based
methods
\citep{spirtes2000causation,chickering2002optimal,pearl2009causality,glymour2019review}.
Assuming
faithfulness, one
can hope to recover the causal graph up to the Markov equivalence class
\citep{verma1990causal,andersson1997markov,tian2013causal},
for which several algorithms have been proposed based on observational data,
interventional data, or a combination of both
\citep{spirtes2000causation,chickering2002optimal,castelo2003inclusion,he2008active,hauser2015jointly}.
However, in many real-world applications learning the full causal graph may be too
ambitious or unnecessary for tackling the problem at hand.
As opposed to causal discovery, causal feature selection aims to identify the
direct causes
of a given variable of interest (the response) from potentially many measured
covariates, instead of the full graph \citep{guyon2007causal}.

Invariant causal prediction (ICP) is an approach to causal feature selection which
exploits invariance of the conditional distribution of a response
given its direct causes under perturbations of the covariates
\citep[ICP,][]{peters2016causal}.
ICP can be formulated from a structural causal modeling, as well as potential
outcome perspective \citep{hernan2010causal}. In contrast to constraint- and
score-based algorithms, ICP requires a specific response variable and data
from heterogeneous environments.

ICP builds on the concept of invariance and can generally be formulated as
conditional independence between the response and the environments given a
candidate set \citep{heinze2018invariant}. Thus, nonparametric conditional
independence tests
\citep{fukumizu2007kernel,zhang2012kernel,candes2018,strobl2019approximate,berrett2019}
can, in principle, always be applied.
However, if one of the conditioning variables is continuous, conditional
independence testing is not feasible without further assumptions in the sense
that there is no test that simultaneously is level and has non-trivial power
\citep{shah2020hardness}. This holds also if the environments are discrete
\citep[Remark~4]{shah2020hardness}.

As an alternative to conditional independence testing, model-based
formulations of ICP have been formulated for linear \citep{peters2016causal}
and non-linear additive noise models \citep[``invariant residual distribution
test'' proposed in][]{heinze2018invariant}. \citet{diaz2022identifying} use
an ``invariant target prediction'' test from \citet{heinze2018invariant}
for testing invariance with a binary response by nonparametrically comparing
out-of-sample area under the receiver operating characteristic (ROC) curve (AUC).
Under correct  model specification, model-based ICP can have considerably higher
power than its nonparametric alternative. Model-based ICP has been extended to
generalized linear models \citep[GLMs, see dicussion in][]{peters2016causal}
and sequential data \citep{pfister2019invariant}. ICP for GLMs and additive
and multiplicative hazard models has been investigated in \citet{laksafoss2020msc}.
For real-world applications of ICP with exogenous environments, see,
for example,
\citet{meinshausen2016gene,heinze2018invariant,christiansen2020switching,migliavacca2021three}.

Many applications feature complex response types, such as ordinal
scales, survival times, or counts and the data-generating mechanism can
seldomly be assumed to be additive in the noise. This is reflected in the most
common model choices for these responses, namely proportional odds logistic
\citep{mccullagh1980regression,tutz2011regression},
Cox proportional hazards \citep{cox1972regression}, and generalized linear
models \citep{mccullagh2019generalized}, which do not assume
additive noise in general. Together, non-continuous responses and non-additive
noise render many causal feature selection algorithms inapplicable.
Moreover, proposed extensions to GLMs and hazard-based models rely on
case-specific definitions of invariance and thus a unified view on linear,
generalized linear, hazards, and general distributional regression is yet to
be established.

In practice, a model-based approach can be desirable, because it leads
to interpretable effect estimates, such as odds or hazard ratios. However, there
is a trade-off between model intelligibility and misspecification. Many
commonly applied regression models are not closed under marginalization or
the inclusion or exclusion of covariates that are associated with the response
\citep[collapsibility,][see also Appendix~\ref{app:collaps}]{greenland1996absence,greenland1999nc,didelez2022collapsibility}.

\subsection{Summary}\label{sec:summary}

Formally, we are interested in discovering the direct causes of a response $\rY
\in \calY \subseteq \RR$ among a potentially large number of covariates $\rX
\in \calX_{1} \times \dots\times \calX_d \subseteq \RR^d$. Consider a
set $S_* \subseteq \{1, \dots, d\}$ (the
reader may think about the ``direct causes'' of $\rY$) and assume that $Y
\mid \rX^{S_*}$ is correctly specified by a \tram{} while all other
conditionals remain unspecified. In Section~\ref{sec:sctram}, we define
structural causal \tram{s} and there, $S_*$ will be the set of causal
parents of $\rY$.
Thus, from now on, we refer to $S_*$ as the causal parents of $\rY$.
\tram{s} characterize the relationship between features and
response via the conditional cumulative distribution function (CDF)
$F_{Y \mid \rX^{S_*} = \rx^{S_*}}(\ry)
\coloneqq \Prob(Y \leq y \mid \rX^{S_*} = \rx^{S_*})$ on the
quantile-scale of a user-specified CDF $\pZ$.
More specifically, when using
\tram{s}, one  models the increasing function $\h(\bcd
\mid\rx^{S_*}) \coloneqq \pZ^{-1} \circ F_{\rY \mid \rX^{S_*} = \rx^{S_*}}(\bcd)$,
called a transformation function. The name stems from the fact that for all
$\rx^{S_*}$ its (generalized) inverse transforms samples of $\rZ~\sim \pZ$
to samples from
the conditional distribution $\rY \mid \rX^{S_*} = \rx^{S_*}$.
Specific choices of $\pZ$ and further modeling assumptions on the functional
form of $\h$ give rise to many well-known models (examples below). Throughout
the paper we illustrate \tramicp{} with a binary response (Ex.~\ref{ex:binary})
and give additional examples with a count and survival response in
Appendix~\ref{app:addex}.
None of the examples
can be phrased as additive noise models of the form $Y = f(X) + \varepsilon$
with $X\indep\varepsilon$. Together with the hardness of conditional independence
testing \citep[see also above]{shah2020hardness}, this motivates the need for
causal feature selection algorithms in more flexible non-additive noise models.

\begin{example}[Binary logistic regression] \label{ex:binary}
The binary logistic regression model (binomial GLM) with $\calY
\coloneqq \{0, 1\}$ can be phrased in terms of the conditional distribution
$F_{\rY \mid \rX^{S_*} = \rx^{S_*}}(0) = \expit(\eparm -
(\rx^{S_*})^\top\shiftparm)$, where
$\expit(\bcd) = \logit^{-1}(\bcd) = (1 + \exp(-\bcd))^{-1}$ denotes the standard
logistic CDF, and $\eparm$ denotes the baseline ($\rx^{S_*} = 0$) log-odds for
belonging to class $0$ rather than $1$. Here, $\shiftparm$ is interpretable as a
vector of log odds-ratios. The model can informally be written as
$F_{\rY \mid \rX^{S_*} = \rx^{S_*}}(\ry) = \expit(\hY(\ry) +
(\rx^{S_*})^\top\shiftparm)$,
where $\hY(0) \coloneqq \eparm$ and $\hY(1) \coloneqq + \infty$.
The latter way of writing the model extends to ordered responses
with more than two levels $\calY \coloneqq \{\ry_1, \ry_2, \dots, \ry_K\}$ with
$\ry_1 < \ry_2 < \dots < \ry_K$, $\hY(\ry_k) \coloneqq \eparm_k$,
for all $k$ and for $k = 2, \dots, K$, $\eparm_k > \eparm_{k-1}$,
using the convention $\eparm_K = +\infty$
\citep[see][proportional odds logistic regression]{mccullagh1980regression}.
\end{example}

In Example~\ref{ex:binary},
we have assumed that the response given its causal parents is correctly specified
by an linear shift \tram{} (see Definition~\ref{def:tramadditivelinear} for more
details). If conditioning on a set that is not $S^*$ always yielded a model
misspecification, one could attempt to identify the set of causal parents
by testing, for different sets $\rX^{S}$ of covariates, whether the model for
$Y$ given $\rX^{S}$ is correctly specified.
However, in Proposition~\ref{prop:id} below, we prove that, in
general, such a procedure does not work. More precisely, there exists a pair
of structural causal models such that both induce the same observational
distribution, and in both, the response given its causal parents
is correctly specified by an (linear shift) \tram{} but the parental sets
differ.

In this work, following a line of work in causal discovery
\citep{peters2016causal,meinshausen2016gene,heinze2018invariant,christiansen2020switching},
we instead assume to have access to data from heterogeneous environments.
Given such data, we define invariance in \tram{s} and propose invariance
tests based on the expected conditional covariance between the environments
and score residuals (\tram-GCM)
and an invariance test based on the Wald statistic for linear shift
\tram{s} in particular (\tram-Wald).
We prove that the \tram-GCM test
is uniformly asymptotically level $\alpha$ for any $\alpha \in (0,1)$
(Theorem~\ref{thm:icp}) and demonstrate empirically that it has
power comparable to or higher than nonparametric
conditional independence testing.
In the context of the result on the hardness of assumption-free
conditional independence testing assumptions for continuous distributions
\citep{shah2020hardness}, our theoretical results show that, under mild
assumptions on the relationship between $\rE$ and $\rX$, the model class
of \tram{}s can be sufficiently restrictive to allow for useful conditional
independence tests.

The rest of this paper is structured as follows. Section~\ref{sec:tram} gives
a technical introduction to transformation models which can be skipped at first
reading. We introduce structural causal \tram{s} in Section~\ref{sec:sctram} and
show that in this class, the set of causal parents is, in general, not identified
(Section~\ref{sec:id}). In Section~\ref{sec:methods}, we present the proposed
\tram-GCM invariance test and its theoretical guarantees.
We apply \tramicp{} to discover causal features of survival in
critically ill hospitalized patients in Section~\ref{sec:casestudy}.

\section{Using transformation models for causal inference} \label{sec:bg}

Transformation models, as introduced by \citet{box1964analysis} in their
earliest form, are models for the conditional cumulative distribution function
of a response given covariates
\citep{doksum1974,bickel1981,cheng1995analysis,hothorn2014conditional}.
\tram{s} transform the response conditional on covariates such that the
transformed response can be modelled on a fixed, continuous latent scale.
Given data
and a finite parameterization,
the transformation
can be estimated via maximum likelihood
\citep{hothorn2018most}. We formally define \tram{s} as a class of non-linear
non-additive noise models depending on the sample space of both response and
covariates.
Our treatment of \tram{s} may appear overly mathematical; however, the formalism
is needed to formulate and prove the identification result (see
Proposition~\ref{prop:id} in Section~\ref{sec:id}) and the uniform
asymptotic level guarantee for the \tram-GCM invariance test (Theorem~\ref{thm:icp}).
A more intuitive introduction to \tram{s} can be found in \citet{hothorn2018most},
for example. We then embed \tram{s} into a causal modeling framework, using
structural causal models \citep[SCMs,][]{pearl2009causality,bongers2016structural}.
We also adapt standard results from parametric \citep{hothorn2018most} and
semi-parametric \citep{mclain2013efficient} maximum likelihood estimation
enabling us to obtain results on consistency and asymptotic normality,
which are exploited by the proposed invariance tests.

\subsection{Transformation models} \label{sec:tram}

Let $\overline\RR \coloneqq \RR \cup \{-\infty, +\infty\}$ denote the extended
real line. Throughout the paper, let $\calZ$ denote the set of
functions $\pZ : \overline\RR \to [0, 1]$ that are (i)
strictly increasing
with $\lim_{x \rightarrow -\infty} \pZ(x) = 0$, $\lim_{x \rightarrow \infty}
\pZ(x) = 1$, (ii) three-times differentiable and have a log-concave derivative
$\dZ = \pZ'$ when restricted to $\RR$, and (iii) satisfy $\pZ(-\infty) = 0$
and $\pZ(+\infty) = 1$. We call $\calZ$ the set of \emph{extended differentiable
cumulative distribution functions}. Given that a CDF $F:\mathbb{R} \rightarrow
\mathbb{R}$
satisfies (i) and (ii), we may add (iii) and refer to the resulting function
as an \emph{extended CDF}. For instance, the extended standard logistic CDF is
given by $\pSL(\rz) = (1 + \exp(-\rz))^{-1}$ for all $\rz\in\RR$ and $\pSL(-\infty)
= 0$ and $\pSL(+\infty) = 1$. Besides $\pSL$, in our applications, we consider the
extended versions of the standard normal CDF $\Phi$, and the standard minimum
extreme value CDF $\pMEV: \rz\mapsto 1 - \exp(-\exp(\rz))$.
By slight abuse of notation, we
use the same letters $\Phi, \pSL, \pMEV$, for the extended CDFs. In general,
specification of a transformation model requires choosing a particular $\pZ \in \calZ$.
Further, for a symmetric positive semi-definite matrix $A$, let $\lambda_{\min}(A)$
denote its smallest eigenvalue and $\norm{A}_{\operatorname{op}}$ denote its operator
norm. For all $n \in \mathbb{N}$, we write $[n]$ as shorthand for $\{1, \dots, n\}$.

We call a function $h : \RR \to \overline{\RR}$ \emph{extended
right-continuous and increasing} (ERCI) on $\calY \subseteq \RR$ if (i) it is
right-continuous and strictly increasing
on $\calY$ and
fulfills $\h(\min\calY) > -\infty$
(if $\min\calY$ exists),
(ii) for all $y < \inf\calY$, we have $h(y) = - \infty$, (iii) for all
$y >\sup\calY$, we have $h(y) = + \infty$, (iv) for all
$t \in (\inf\calY, \sup\calY) \setminus \calY$, we have $\h(t) = \h(\ubar{t})$,
where $\ubar{t} \coloneqq \sup\{\upsilon \in \calY : \upsilon < t\}$
and (v) $\lim_{\upsilon\to-\infty} h(\upsilon) = -\infty$ and
$\lim_{\upsilon\to\infty} h(\upsilon) = \infty$. Condition
(iv) is needed to ensure that $\h$ is piece-wise constant outside of $\calY$.
Finally, for a function $f : \overline{\RR} \to \RR$, we denote the derivative
$f' : \RR \to \RR$ s.t. for all $x \in \RR$, $f'(x) = \frac{\dd}{\dd u}
f(u)\rvert_{u=x}$.
We are now ready to define the class of transformation models.

\begin{definition}[Transformation model]\label{def:tram} \sloppy
Let $\calY \subseteq \RR$
and $\calX \coloneqq \calX_1 \times \ldots \times
\calX_d \subseteq \RR^d$, where for all $i$, $\calX_i \subseteq \RR$.
The set of all \emph{transformation functions} on $\calY \times \calX$
is defined as
\begin{align*}
    \trafoall \coloneqq \bigg\{
    \h: \RR \times \calX \to \overline\RR
    \,\big\vert\, \forall \rx \in \calX,
    \ \h(\bcd \mid \rx) \mbox{ is ERCI on } \calY
    \bigg\}.
\end{align*}
Then, for a fixed \emph{error distribution} $\pZ \in \calZ$ and a set of
transformation functions $\trafosub \subseteq \trafoall$,
the \emph{family of \tram{}s} $\calM(\pZ, \calY, \calX, \trafosub)$
is defined as the following set of conditional cumulative distribution
functions\footnote{%
In Proposition~\ref{prop:cdf} in Appendix~\ref{app:lemmata}, we show that $\calM$
indeed only contains CDFs.}
\citep[see also Definition~2 in][]{hothorn2018most}:
\begin{align*}
    \calM(\pZ, \calY, \calX, \trafosub) \coloneqq\
    &\big\{F_{\rY\mid\rX = \bcd} : \RR\times\calX\to[0,1] \,\big|\,\\ &\qquad \exists \h
    \in \trafosub : \forall \rx \in \calX \
    \forall \ry \in \RR,\ F_{\rY\mid\rX=\rx}(\ry) = \pZ(\h(\ry \mid \rx))
    \big\}.
\end{align*}
As such, a single \tram{} is fully specified by $(\pZ, \h)$, $\pZ \in \calZ,
\h \in \trafosub$. The condition that for all $\rx\in\calX$,
$\h(\bcd\mid\rx)$ is ERCI on $\calY$ ensures that the support of the
induced conditional distribution specified by $\pYx$ is $\calY$.
Further, for all $\rx \in \calX$ and $\rz \in \overline\RR$, we write
$\h^{-1}(\rz \mid \rx) \coloneqq \inf \{\ry \in \calY : \rz \leq \h(\ry
\mid \rx)\}$ for the inverse transformation function.
\end{definition}

The inverse transformation function $\h^{-1}(\bcd \mid \rx)$ at a given
$\rx$ can be interpreted analogously to a quantile function:
Given some $\rX = \rx$, we can obtain an observation from $F_{\rY \mid \rX = \rx}$ by
sampling an observation from $\pZ$ and passing it through $\h^{-1}(\bcd \mid \rx)$.

In statistical modelling, it is common to additionally assume additivity of the
effects of $\rX$ on a specific scale. For instance, in linear regression the
covariates enter as a linear predictor on the scale of the conditional mean.
In this work, we restrict ourselves to the class of shift \tram{s} in which
additivity is assumed on the scale of the transformation function.
%
\begin{definition}[Shift \tram{s}]\label{def:tramadditive}
Let $\calY$, $\calX$ and $\pZ\in\calZ$ be as in Definition~\ref{def:tram}.
Further, let $\calF \coloneqq \{f: \calX \to \RR \mid
f \text{ measurable}\}$
and $\bltrafo \coloneqq
\{\hY : \RR \to \overline\RR \mid \hY \mbox{ is ERCI on } \calY\}$.
Let the set of \emph{shift transformation functions} be defined as
\begin{equation*}
    \trafoadd \coloneqq
    \left\{\h \in \trafoall \mid
    \exists \hY \in \bltrafo, \ f \in \calF :
    \forall \rx \in \calX, \ \h(\bcd \mid \rx) = \hY(\bcd) -
    f(\rx)\right\}.
\end{equation*}
Then, $\calM(\pZ, \calY, \calX, \trafoadd)$ denotes
the family of \emph{shift} \tram{s} and a \tram{} $\pZ
\circ \h$ is called \emph{shift \tram{}} iff $\h \in
\trafoadd$. Further, any $\hY \in \bltrafo$ is referred to as a \emph{baseline
transformation}.
\end{definition}
We next introduce the subset of linear shift \tram{s} in which
the covariates enter as a linear predictor.
%
\begin{definition}[Linear shift \tram{s}]\label{def:tramadditivelinear}
Consider shift \tram{s} specified by $\pZ, \calY, \calX, \calF,
\trafoadd$, as in Definition~\ref{def:tramadditive}.
Let $\basisx : \calX \to \RR^b$ be a finite collection of basis transformations
and define $\calF_{\basisx} \coloneqq \{f \in \calF \mid \exists
\shiftparm \in \RR^b \mbox{ \st} f(\bcd) = \basisx(\bcd)^\top\shiftparm\}$.
The set of \emph{linear shift transformation functions \wrt $\basisx$} is
defined as
\begin{align}
    \trafolin(\basisx) \coloneqq
    \left\{\h \in \trafoadd \, \big\vert \,
    \exists \hY \in \bltrafo,\ f \in \calF_{\basisx}
    : \forall \rx \in \calX : \h(\bcd \mid\rx) = \hY(\bcd) -
    f(\rx) \right\}.
\end{align}
Then, $\calM(\pZ, \calY, \calX, \trafolin(\basisx))$
denotes the family of \emph{linear shift \tram{s} \wrt $\basisx$}.
Further, a \tram{} $\pZ\circ\h$ is called \emph{linear shift \tram{}
\wrt $\basisx$} iff $\h \in \trafolin(\basisx)$.
For the special case of $\basisx : \rx \mapsto \rx$, we write
$\trafolin$ and refer to the class and its members
as \emph{linear shift \tram{s}}.
\end{definition}
%

Estimation and inference in \tram{s} can be based on the log-likelihood
function---if it exists. The following assumption ensures that this is the case.
\begin{assumption}\label{asmp:densities}
We have $\trafosub\subseteq\trafoadd$.
Furthermore, if $\calY$ is uncountable,
$\pZ, \calX, \trafosub$
are such that
for all $\rx\in\calX$ and $\h \in\trafosub$,
\begin{equation}\label{eq:cancondpdf}
\dYx(\bcd; \h) \coloneqq \pZ'(\h(\bcd\mid\rx))\h'(\bcd\mid\rx),
\end{equation}
where $\h'(\ry\mid\rx) \coloneqq \frac{\dd}{\dd\upsilon} \h(\upsilon
\mid\rx) \rvert_{\upsilon = \ry}$,
is well-defined and a density (\wrt Lebesgue measure)
of the conditional CDF induced by the \tram{}.
\end{assumption}
Assumption~\ref{asmp:densities}
allows us to define (strictly positive) canonical conditional densities
with respect to a fixed measure that we denote by $\mu$:
If $\calY$ is countable, we let $\mu$ denote the counting measure on $\calY$ and
for all $\ry \in \calY$,
define the canonical conditional density
by $f_{\rY\mid\rX=\rx}(
\ry
; \h) \coloneqq
\pZ(\h(
\ry
\mid\rx)) - \pZ(\h(\ubar\ry\mid\rx))$, where
$\ubar\ry \coloneqq \sup\{\upsilon \in \calY : \upsilon < y\}$\footnote{%
We adopt the convention that the supremum of the empty set is $-\infty$.}.
If $\calY$ is uncountable, we let $\mu$ denote the Lebesgue measure restricted to $\calY$ and
the canonical conditional density is then defined by~\eqref{eq:cancondpdf}.
In either case, $\trafosub\subseteq\trafoadd$ ensures that for all $\rx$ and $\ry
\in \calY$, $f_{\rY\mid\rX=\rx}(\ry; \h)>0$.
Thus, for $(\pZ, \calY, \calX, \trafosub)$ satisfying
Assumption~\ref{asmp:densities}, we can define the \tram{} log-likelihood as
$\ell : \trafosub \times \calY \times \calX \to \RR$
with $\ell(\h; \ry, \rx) \coloneqq \log \dYx(\ry; \h)$.

When applying ICP to linear additive noise models, invariance can be formulated
as uncorrelatedness between residuals and environments. In \tram{}s, however,
the response can be categorical, reducing the usefulness of classical residuals.
Instead, score residuals \citep{lagakos1981residuals,korepanova2020score,kook2021danchor} are a
natural choice for testing invariance of \tram{s}. Score residuals were first
introduced by \citet{lagakos1981residuals} for multiplicative hazard models
\citep[see also][for non-multiplicative hazard models]{korepanova2020score}
and extended to linear shift \tram{s} by
\citet[][Definition~2]{kook2021danchor}.
Score residuals coincide with scaled least-squares residuals in linear regression
with normal errors and martingale residuals in the Cox proportional hazards model
\citep{barlow1988residuals} and directly extend to censored responses
\citep{lagakos1981residuals,farrington2000residuals}.
In this work, score residuals play a major role in formulating invariance tests
(Section~\ref{sec:methods}) and have been used for causal regularization in a
distributional version of anchor regression
\citep{rothenhaeusler2018anchor,kook2021danchor}.
For defining score residuals, we require the following assumption
(which, by definition, is satisfied for $\trafoadd$
and $\trafolin$).
\begin{assumption}\label{asmp:closure}
$\trafosub$ is closed under scalar addition, that
is, for all $h \in \trafosub$ and $\alpha \in \RR$, we
have\footnote{We adopt the convention that for all $\alpha\in\RR$,
$-\infty + \alpha = -\infty$ and $\infty + \alpha = \infty$.}
$h + \alpha \in \trafosub$.
\end{assumption}
%
\begin{definition}[Score residuals, \citeauthor{lagakos1981residuals},
\citeyear{lagakos1981residuals}; \citeauthor{kook2021danchor},
\citeyear{kook2021danchor}]\label{def:sresids}
Let $\calY$, $\calX$, $\pZ\in\calZ$ and $\trafosub\subseteq
\trafoall$ be as in Definition~\ref{def:tram}. Impose
Assumptions~\ref{asmp:densities} and~\ref{asmp:closure}.
Then, considering $\alpha \in \RR$, the \emph{score residual} $R:
\trafosub \times \calY \times \calX \to \RR$ is defined as
\begin{align*}
    R: (\h; \ry, \rx) \mapsto
    \frac{\partial}{\partial\alpha}
    \ell(\h + \alpha; \ry, \rx) \big\rvert_{\alpha = 0}.
\end{align*}
\end{definition}

\begin{example}[Binary logistic regression, cont'd]\label{ex:bin:inv}
The family of binary linear shift logistic regression models is
given by $\calM(\pSL, \{0,1\}, \calX, \trafolin)$. We can thus
write for all $\rx \in \calX$, $\h(\bcd \mid \rx) \coloneqq \hY(\bcd) -
\rx^\top\shiftparm$ with $\hY(0) \coloneqq \eparm$ and, by convention, $\hY(1)
\coloneqq +\infty$. The likelihood contribution for a given observation $(\ry,
\rx)$ is $\pSL(\h(0 \mid \rx))^{1-\ry}(1-\pSL(\h(0 \mid \rx))^\ry$. The
score residual is given by $R(\h; \ry, \rx) = 1 - \ry - \pSL(\h(0 \mid
\rx))$. Further, the inverse transformation function
is given by $\h^{-1}: (\rz, \rx) \mapsto \1(\rz\geq\eparm -
\rx^\top\shiftparm)$.
\end{example}

\subsection{Structural causal transformation models}\label{sec:sctram}

Next, we cast \tram{s} into a structural causal modelling framework
\citep{pearl2009causality} and return to our examples from Section~\ref{sec:intro}.
For all subsets $S \subseteq [d]$, define $\calX^S$ to be the
projection of $\calX$ onto the ordered coordinates in $S$.
For the rest of this paper, we restrict ourselves to shift \tram{s}.
In this case, any ``global'' model class $\trafosub$ naturally induces
submodel classes $\trafosubS \subseteq \trafoallS$ for all
$S \subseteq [d]$ by the following construction: $\trafosubS \coloneqq \{h
\in \trafosubS^* \,|\, \exists  h^{\text{global}} \in \trafosub \text{ s.t. }
\forall \rx \in \calX,\,h^{\text{global}}(\bcd \mid \rx) = h(\bcd \mid
\rx^S)\}$. If $(F_Z, \calY, \calX, \trafosub)$ satisfies
Assumption~\ref{asmp:densities}, then $(F_Z, \calY, \calX^S, \trafosubS)$
does too. We are now ready to define structural causal \tram{s}.

\begin{definition}[Structural causal \tram{}]\label{def:tramscm}
Let $\calY$, $\calX$, $\pZ\in\calZ$ be as in Definition~\ref{def:tram}.
Let $\trafosub \subseteq \trafoall$ be a class of transformation functions
such that Assumption~\ref{asmp:densities} holds. Let $(\rZ, N_\rX)$ be jointly
independent with $\rZ \sim \pZ$. Then, a \emph{structural causal \tram{} $C$
over $(\rY, \rX)$} is defined as
\begin{equation}\label{eq:tram_scm}
C \coloneqq
\begin{cases}
    X^j \coloneqq g_{j}(\rX, \rY, N_{X^j}), \quad \forall j \in [d] \\
    \rY \coloneqq \h^{-1}(\rZ \mid \rX^{S_*}),
\end{cases}
\end{equation}
where $S_* \subseteq [d]$, $\h \in \trafosubPA$ is the \emph{causal
transformation function}
and $\pa_C(\rY) \coloneqq S_*$ denotes the set of causal parents of $\rY$
in $C$ and
$g_j$,
$j \in [d]$,
are
arbitrary measurable functions.
By $\Prob^C_{(\rY,\rX)}$ we denote the observational distribution induced
by $C$. We assume that the
induced graph (obtained by drawing directed edges from
the observed
variables on the
right-hand side to variables on the left-hand side) is acyclic.
(This, in particular, implies that the function $g_j$ does not
depend on $X^j$.)
We denote by $\calC(\pZ,\calY,\calX,\trafosub)$ the collection of all
structural causal \tram{s} with error distribution $\pZ$ and causal
transformation function $\h \in \trafosub$.
\end{definition}

\subsection{Non-identifiability of the causal parents in transformation models}
\label{sec:id}

We now show that performing causal feature selection in
structural causal transformation models
requires further assumptions. We consider a response variable $Y$ and a set
of covariates $\rX$ and assume that $(Y, \rX)$ are generated from an
(unknown) structural causal \tram{} (defined in \eqref{eq:tram_scm}) with
(known) $\trafosub \subsetneq \trafoall$. In our work,
the problem of causal feature selection concerns learning the causal parents
$\pa(Y)$ given a sample of $(Y, \rX)$ and knowledge of
$\pZ$, $\calY$, $\calX$, $\trafosub$ (which specifies the model
class $\calM(\pZ, \calY, \calX, \trafosub)$).

In this work, we specify the model class for the conditional of the
response, given its causal parents, $Y\mid\rX^{\pa(\rY)}$ by a \tram{};
the remaining
conditionals are unconstrained. Identifiability of causal structure has been
studied for several model classes that constrain the joint distribution
$(Y, \rX)$. When considering the class of linear Gaussian SCMs, for example,
the causal parents are in general not identifiable from the observational
distribution (as there are linear Gaussian SCMs with a different structure
inducing the same distribution). This is different for other model classes:
When considering linear Gaussian SCMs with equal noise variances
\citep{petersbuehlmann2013}, linear non-Gaussian SCMs \citep{shimizu2014lingam}
or nonlinear Gaussian SCMs \citep{hoyer2008neurips,peters2014jmlr}, for
example, the graph structure (and thus the set of causal parents of $Y$) is
identifiable under weak assumptions (identification then becomes possible by
using goodness-of-fit procedures). To the best of our knowledge, identifiability
in such model classes
(i.e.,\ recovering the causal parents of $Y$, not the entire graph or
equivalence classes)
has not been studied when constraining only the
conditional distribution of $Y$ given $\rX^{\pa(\rY)}$.

\tram{s} are generally not closed under marginalization (see
Appendix~\ref{app:collaps} for a detailed discussion on non-collapsability)
and one may hypothesize that this model class allows for identifiability
of the parents
(e.g., by considering different subsets of covariates and testing for goodness
of fit).
We now prove that this is not the case: In general, for
\tram{s} (and even for linear shift \tram{s}), the causal parents are not
identifiable from the observed distribution.
Instead, additional assumptions are needed to facilitate causal feature
selection in \tram{s}.

Definition~\ref{def:FZ_ident} formally introduces the notion of identifiability
of the causal parents and Proposition~\ref{prop:id}
provides the non-identifiability result.
%
\begin{definition}[Subset-identifiability of the causal parents]\label{def:FZ_ident}\sloppy
Let $\calC$ denote a collection of structural causal models. The set of causal
parents is said to be $\calC$-\emph{subset-identifiable} if for all pairs $C_1,
C_2 \in \calC$ it holds that
\begin{equation}
    \Prob^{C_1}_{(Y,\rX)} = \Prob^{C_2}_{(Y,\rX)} \implies
    \pa_{C_1}(Y) \subseteq \pa_{C_2}(Y)
    \; \lor \;  \pa_{C_2}(Y) \subseteq
    \pa_{C_1}(Y).
\end{equation}
\end{definition}

\begin{proposition}[Non-subset-identifiability]\label{prop:id}
For all $A \subseteq \RR$ that are either an interval or countable,
$\pZ\in\calZ$, $\calY\subseteq\RR$, there exists a class of transformation
functions $\trafos{A \times A} \subseteq \trafos{A\times A}^{\text{shift}}
\subsetneq \trafos{A \times A}^*$,
such that the
set of causal parents is not $\calC(\pZ, \calY, A\times A, \trafos{A \times
A})$-subset identifiable.
\end{proposition}
A proof is given in Appendix~\ref{proof:id}, where we construct a joint
distribution over three random variables $(Y, X^1, X^2)$, in which the two
conditionals $Y \mid X^1$ and $Y \mid X^2$ are
\tram{s}. This implies that there are two structural causal \tram{s}
that have identical observational distributions, while $Y$ has two different
(non-empty) sets of causal parents that do not overlap.
The proof in \ref{proof:id} characterizes how to construct
such a joint distribution for shift \tram{s}. For
illustrative purposes, we present a concrete
example in Appendix~\ref{proof:id}
in which $\calY = \calX^1 = \calX^2 = \{1, 2, 3\}$
and
$Y \mid X^1$ and $Y
\mid X^2$ are proportional odds logistic regression models.
We then sample from the induced distribution
We sample from the induced distributions of the two structural causal
\tram{s} constructed in the proof and apply the naive method described above
of performing goodness-of-fit tests to identify the parents. We see that this
method indeed fails to identify a non-empty subset of the
parents in this example.

Instead of subset-identifiability, one can also consider a stronger notion of
\emph{full identifiability}, which states that the set of causal parents can
be uniquely determined by the observed distribution (formally defined in
Appendix~\ref{app:id}). Proposition~\ref{prop:id} immediately implies that
the set of causal parents is not fully identifiable either.

\section{Transformation model invariant causal prediction} \label{sec:methods}

Even if the observational distribution is insufficient to identify causal
parents, identifiability can become possible if we have access to data from
multiple, heterogeneous environments.
Invariant causal prediction
\citep[ICP,][]{peters2016causal} exploits the invariance of causal mechanisms
\citep{haavelmo1943statistical,frisch1948,aldrich1989autonomy,pearl2009causality,schoelkopf2012anti}
under interventions on variables other than the response.
Depending on the response variable, multi-center clinical trials, data
collected from different countries or different points in time may fall into
this category. We then show that under Setting~\ref{set:env}, the set of causal
parents is subset-identifiable (Proposition~\ref{prop:parents}) and fully
identifiable if the environments are sufficiently heterogeneous
(Proposition~\ref{prop:identification}).

\begin{setting}[Data from multiple environments]\label{set:env}\sloppy
Let $\calY$, $\calX$, $\pZ \in \calZ$ be as in Definition~\ref{def:tram}
and let $\trafosub \subseteq \trafoall$ be a class of transformation functions
such that Assumptions~\ref{asmp:densities} and~\ref{asmp:closure} hold.
Let $C_*$ be a structural causal \tram{} (Definition~\ref{def:tramscm})
over $(\rY, \rX, \rE)$ such that\\
\begin{minipage}{0.49\textwidth}
\begin{align*}
C_* \coloneqq \begin{cases}
    E^k \coloneqq \ m_k(\rX, N_{E^k}), \quad \forall k \in [q]
\\
    X^j \coloneqq \ g_j(\rX, \rE, \rY, N_{X^j}), \quad  \forall j \in [d] \\
    Y \ \coloneqq \ \h^{-1}_*(Z \mid \rX^{S_*}),
\end{cases}
\end{align*}
\end{minipage}
\begin{minipage}{0.49\textwidth}
\centering
\begin{tikzpicture}[node distance=1.0cm, <-> /.tip = Latex, -> /.tip = Latex,
    thick, roundnode/.style={circle, draw, inner sep=1pt,minimum size=7mm},
    squarenode/.style={rectangle, draw, inner sep=1pt, minimum size=7mm}]
    \node[roundnode] (X) {$X^2$};
    \node[right=of X, roundnode] (X2) {$X^1$};
    \node[below=of X2, roundnode] (E) {$E$};
    \node[right=of X2, squarenode] (Y) {$Y$};
    \node[below=of Y, roundnode] (X3) {$X^3$};
    \node[right=of X3, roundnode] (X4) {$X^4$};
    \node[right=of Y, roundnode, fill=gray!20] (Z) {$\rZ$};
    \draw[->] (E) edge[bend left=00] (X2);
    \draw[->] (X) edge[bend left=45] (Y);
    \draw[->] (X2) -- (Y);
    \draw[->] (Z) -- (Y);
    \draw[->] (Y) -- (X3);
    \draw[->] (X4) -- (X3);
    \draw[->] (X) -- (X2);
    \draw[->] (E) -- (X3);
\end{tikzpicture}
\end{minipage}\\[12pt]
where $\h_* \in \trafosubPA$ with $S_* \subseteq [d]$ denoting the parents of
$\rY$ and $(Z, N_\rX, N_\rE)$ denoting the jointly independent noise variables.
By definition, the induced graph $\calG_*$ (containing the variables $E^1,
\ldots, E^q, X^1, \ldots, X^d, \rY$) is acyclic. In this setup, the random
vector $\rE$ encodes the environments and takes values in $\calE \subseteq
\RR^q$. We further assume that the parents of $\rE$ can only be
non-descendants\footnote{A node is called a non-descendant of $Y$ in $\calG_*$
if there is no directed path from $Y$ to that node in $\calG_*$.} of $Y$ in
$\calG_*$, which is satisfied, for example, if $\rE$ is exogeneous (that is,
each $E^k$ is a function of $N_{E^k}$ only); $\rE$ may be discrete or
continuous. An example of a DAG contained in this setup is depicted on the
right. By $\calD_n \coloneqq \{(y_i, \rx_i, \evec_i)\}_{i=1}^n$, we denote an
i.i.d.\ sample from $\Prob^{C_*}_{(\rY,\rX,\rE)}$.
\end{setting}

As for ICP, invariance plays a key role for \tramicp{}. We say a subset of
covariates is invariant if the corresponding transformation model correctly
describes the conditional distribution across the environments $\rE$. More
formally, we have the following definition.
%
\begin{definition}[$(\pZ, \trafosub)$-invariance]\label{def:traminv}
Assume Setting~\ref{set:env}. A subset of covariates $S \subseteq [d]$ is
$(\pZ,\trafosub)$\emph{-invariant} if there exists $\h^S\in \trafosubS$, such
that for $\Prob_{(\rX^S,\rE)}$-almost all $(\rx^S, \evec)$,
\begin{align}
    (\rY \mid \rX^S = \rx^S, \rE = \evec) \mbox{ and } (\rY \mid \rX^S = \rx^S)
    \mbox{ are identical with conditional CDF } \pZ(\h^S(\bcd \mid \rx^S)).
\end{align}
\end{definition}
%
If an \emph{invariant transformation function} $\h^S$ according to
Definition~\ref{def:traminv} exists, it is
$\Prob_{\rX^S}$-almost surely
unique (see Lemma~\ref{lem:unique} in Appendix~\ref{app:lemmata}).
Proposition~\ref{prop:parents} shows that the parental set fulfills
$(\pZ,\trafosub)$-invariance, which is sufficient to establish
coverage guarantees for invariant causal prediction in \tram{s}.
A proof is given in Appendix~\ref{proof:parents}.
%
\begin{proposition}\label{prop:parents}
Assuming Setting~\ref{set:env}, the set of causal parents $S_*$ is
$(\pZ,\trafosub)$-invariant.
\end{proposition}

The set of causal parents $S_*$ together with the causal transformation function
$\h_*$ in Setting~\ref{set:env} may not be the only set satisfying $(\pZ,
\trafosub)$-invariance. In this vein, we define the set of
\emph{identifiable causal predictors} as
\begin{equation}
    S_I \coloneqq \bigcap_{S \subseteq [d] : S \mbox{ is }
    (\pZ, \trafosub)\mbox{-invariant}} S.
\end{equation}
Since
$(\pZ, \trafosub)$-invariance
is a property of the observed distribution,
$S_I$ is identifiable from the observed distribution, too.
By Proposition~\ref{prop:parents}, $S_I \subseteq S_*$.
Thus,
the causal parents $S_*$ are subset-identifiable.
In a modified version of Setting~\ref{set:env} in which
$\rE$ is among the causal parents of $Y$ (see Setting~\ref{set:empty} in
Appendix~\ref{app:empty})
and the induced distribution is faithful w.r.t.\ the induced graph
(see \citealt{spirtes2000causation}, p.~56), the set of
identifiable causal predictors is empty (because there exists no
$(\pZ,\trafosub)$-invariant set; we prove this statement as
Proposition~\ref{prop:empty} in Appendix~\ref{app:empty}).

Furthermore,
if the environments induce a sufficient amount of heterogeneity in the data,
in the sense that $S_* \subseteq \ch(\rE)$, then $S_I = S^*$, so the causal
parents are fully identified
(this result assumes faithfulness\footnote{
Strictly speaking,
assuming faithfulness for the whole graph when proving
Proposition~\ref{prop:identification} is too strong.
As can be seen from the proof,
it suffices to assume that for all $S \subseteq [d]$,
we have $E$ is not $d$-separated from $Y$ given $\rX^S$
implies $E$ is not independent of
$Y$ given $\rX^S$.}

\begin{proposition}\label{prop:identification}
Assume Setting~\ref{set:env}. Let $\calG$ be the DAG induced by $C_*$ and assume
that $\Prob_{(\rY,\rX,\rE)}^{C_*}$ is faithful \wrt to $\calG$. If  $S_*
\subseteq \ch(\rE)$, where $\ch(\rE)$ denotes the children of $\rE$, we have
$S_I = S_*$.
\end{proposition}
A proof is given in Appendix~\ref{proof:identification}.
%
For simple model classes such as linear Gaussian SCMs, sufficient conditions
for faithfulness are known \citep{spirtes2000causation}. In our setting,
analyzing the faithfulness assumption is particularly challenging due to
non-collapsibility and non-closure under marginalization of \tram{s}
(see Appendix~\ref{app:collaps}). Nonetheless, we empirically show in our
simulations (see Appendix~\ref{sec:results}) that faithfulness is not
violated, for example, if the coefficients in linear shift \tram{s} are
sampled from a continuous distribution.

\subsection{Testing for invariance}\label{subsec:testing}

We now translate
$(\pZ,\trafosub)$-invariance into testable conditions which are
applicable to general \tram{s} and thus general response types. Here, we
propose an invariance condition based on score residuals
(Definition~\ref{def:sresids}). The following proposition shows that
the score residuals are uncorrelated with the environments (in
Setting~\ref{set:env}) when conditioning on an invariant set.
%
\begin{proposition}[Score-residual-invariance]\label{prop:residualinv}
Assume Setting~\ref{set:env} and that \eqref{eq:interch} in
Appendix~\ref{app:lemmata} holds. Then, we have the following implication:
\begin{align}\label{eq:residinv}
\begin{split}
S \mbox{ is } (\pZ,\trafosub)\mbox{-invariant} \implies
    &\Ex[R(\h^S; \rY, \rX^S) \mid \rX^S] = 0, \mbox{ and}\\
    &\Ex[\Cov[\rE, R(\h^S; \rY, \rX^S) \mid \rX^S]] = 0,
\end{split}
\end{align}
where $\Ex[\Cov[\rE, R(\h^S; \rY, \rX^s) \mid \rX^S]]
\coloneqq \Ex[\rE R(\h^S; \rY, \rX^s) \mid \rX^S] - \Ex[\rE\mid\rX^S]
\Ex[R(\h^S; \rY, \rX^S) \mid \rX^S]$ denotes the expected conditional
covariance between the residuals and environments.
\end{proposition}
A proof is given in Appendix~\ref{proof:residualinv}. In Appendix~\ref{app:cens},
we extend \tramicp{} (in particular, Proposition~\ref{prop:residualinv}) to
uninformatively censored observations, where $Y$ itself is unobserved.

We now turn to the problem of testing invariance from finite data.
Section~\ref{sec:scoretest} develops a test, based on similar ideas as
the Generalised Covariance Measure \citep[GCM,][]{shah2020hardness}, on how to
test the implication in \eqref{eq:residinv}. As a second, alternative, invariance
test, we also propose a Wald test for the existence of main and interaction terms
involving the environments in Appendix~\ref{sec:wald}; we
show in Proposition~\ref{prop:wald} that for linear shift
\tram{s}, such a test is closely related to the implication in
Proposition~\ref{prop:residualinv}.

For all $S \subseteq [d]$, and sample sizes $n$, let $p_{S,n} : (\RR \times
\calX^S \times \calE)^n \to [0, 1]$ be the $p$-value for the null hypothesis
that $S$ is $(\pZ,\trafosub)$-invariant. All proposed invariance tests are
embedded in a subset-search over the set of covariates, in which we return the
intersection of all non-rejected sets at a given level $\alpha \in (0, 1)$ (ICP;
Algorithm~\ref{alg:outer}).

\begin{algorithm}[!ht]
\caption{Invariant causal prediction \citep{peters2016causal}} \label{alg:outer}
\begin{algorithmic}[1]
\Require Data $\calD_n$ from Setting~\ref{set:env}, significance level $\alpha
\in (0, 1)$, and a family of invariance tests $(p_{S,n})_{S \subseteq \{1,
\dots, d\}}$ (outputting a $p$-value; see Algorithms~\ref{alg:tramgcm},
and~\ref{alg:tramicp} and the comparators in Section~\ref{sec:comparators})
\State For all $S \subseteq [d]$, compute $p_{S,n}(\calD_n)$
    \gComment{Compute $p$-value of invariance test}
\State \Return{$S_n \coloneqq \bigcap_{S : p_{S,n}(\calD_n) > \alpha} S$}
\gComment{Intersection over all non-rejected sets}
\end{algorithmic}
\end{algorithm}
It directly follows from Proposition~\ref{prop:parents} that if the tests are
level $\alpha$, then the output of Algorithm~\ref{alg:outer} is contained in
the causal parents with large probability
\citep[see][Theorem~1]{peters2016causal}, that is,
$\Prob({S}_n \subseteq \pa_{C_*}(\rY)) \geq 1 - \alpha$.\footnote{
The coverage guarantee holds by
$\Prob(S_n \subseteq \pa_{C_*}(\rY))
\geq \Prob(p_{S^*,n}(\calD_n) > \alpha) = 1 - \alpha$.}
This coverage guarantee does not require faithfulness or sufficiently
heterogeneous environments as assumed in
Proposition~\ref{prop:identification}.\footnote{If the test had
perfect power, then under the conditions assumed in
Proposition~\ref{prop:identification}, the procedure would output $S^*$.
In practice, even under the conditions assumed in
Proposition~\ref{prop:identification},
we may not correctly
reject all non-invariant sets, but the coverage guarantee
still holds. In this sense, the method adapts automatically to settings, in
which the heterogeneity is sufficiently strong.}
It only requires that
the environment is a measurable function of non-descendants of $Y$
and is not a causal parent of $Y$. Assuming
the induced distribution is faithful w.r.t.\ the induced graph
and oracle tests for
$(\pZ, \trafosub)$-invariance, the coverage guarantee
even
holds if $\rE$ is a causal parent of $Y$ (this, in particular, includes
cases of linear shift \tram{s} in which $\rE$ only interacts with $\rX^{S_*}$ to
cause $Y$, i.e.,~effect modification).
In this case, there exists no $S \in [d]$ that is $(\pZ,\trafosub)$-invariant
and \tramicp{} returns the empty set with high probability
if the test is sufficiently powerful
(see Proposition~\ref{prop:empty} in Appendix~\ref{app:empty}).

We refer to the combination of ICP (Algorithm~\ref{alg:outer}) with the
proposed \tram-GCM invariance test (Algorithm~\ref{alg:tramgcm}) as
\tramicp-GCM, with the proposed \tram-Wald invariance test
(Algorithm~\ref{alg:tramicp}) as \tramicp-Wald and using a nonparametric
conditional independence test (see Appendix~\ref{sec:comparators}) as
nonparametric ICP.

\subsubsection{Invariance tests based on score residuals} \label{sec:scoretest}

We can test the null hypothesis of $(\pZ,\trafosub)$-invariance by
testing the implication in~\eqref{eq:residinv}, \ie
uncorrelatedness between score residuals and residualized
environments in a GCM-type invariance test (Algorithm~\ref{alg:tramgcm}).
This requires that the maximum likelihood estimator exists and is unique.
\begin{assumption}\label{asmp:mle}
Under Setting~\ref{set:env} and for all $S \subseteq [d]$, the maximum likelihood
estimator,
given by $\argmax_{\h\in\trafosubS} \ell(\h; \calD_n)$, exists and is unique.
\end{assumption}
See also the regularity conditions in \citet[Assumptions~I--V]{mclain2013efficient}.
Theorem~\ref{thm:icp} shows that the proposed test
is uniformly asymptotically level $\alpha$ for any $\alpha \in (0,1)$.

\begin{algorithm}[!ht]
\caption{\tram-GCM invariance test
} \label{alg:tramgcm}
\begin{algorithmic}[1]
\Require Data $\calD_n$ from Setting~\ref{set:env}, $S \subseteq [d]$,
estimator $\hat\muvec$ for $\muvec(\rX^S) \coloneqq \Ex[\rE \mid \rX^S]$.
\State Fit the \tram{}:
$\hat\h \gets \argmax_{\h \in \trafosubS} \ell(\h; \calD_n)$
\State Obtain $\hat\muvec$ using data $\calD_n$
\State Compute residual product terms:
$\rL_i \gets R(\hat\h; \ry_i, \rx_i^S) \{\evec_i - \hat\muvec(\rx^S_i)\}, i = 1, \dots, n$
\State Compute residual covariance:
$
\hat{\Sigma} \gets n^{-1} \sum_{i=1}^n \rL_i \rL_i^\top
- \left(n^{-1} \sum_{i=1}^n \rL_i \right) \left(n^{-1} \sum_{i=1}^n \rL_i\right)^\top
$
\State Compute test statistic:
$\rT_n \gets \hat{\Sigma}^{-1/2} \left(n^{-1/2} \sum_{i=1}^n \rL_i \right)$
\State Compute $p$-value:
$p_{S,n}(\calD_n) \gets 1 -  F_{\chi^2_q}(\lVert \rT_n \rVert^2_2)$
\State \Return{$p_{S,n}(\calD_n)$}
\end{algorithmic}
\end{algorithm}

\begin{theorem}[Uniform asymptotic level of the invariance test in
Algorithm~\ref{alg:tramgcm}]\label{thm:icp}
Assume Setting~\ref{set:env}
and
Assumption~\ref{asmp:mle}
and for a fixed $S \subseteq [d]$
let $\calP \coloneqq \{\Prob_{(\rY,\rX^S,\rE)} \mid S \mbox{ is }
\trafosub\mbox{-invariant}\}$ denote the set of null distributions for the
hypothesis $H_0(S): S$ is $(\pZ, \trafosub)$-invariant
(Definition~\ref{def:traminv}). For all $P$ in $\calP$, we denote by $\h_P$ the
$h^S \in \trafosubS$ in the definition of $(F_Z, \trafosub)$-invariance and
$\muvec(\rX^S) \coloneqq \Ex_P[\rE \mid \rX^S]$.
Let $\xivec \coloneqq \rE - \muvec(\rX^S)$. Assume that
\begin{enumerate}[label=(\alph*)]
    \item $\inf_{P\in\calP} \lambda_{\min}(\Ex_P[R(\h_P;\rY,\rX^S)^2\xivec
        \xivec^\top]) > 0$, \label{gcm:c1}
    \item There exists $\delta > 0$, s.t.~$\sup_{P\in\calP} \Ex_P[\lVert R(\h_P;
        \rY,\rX^S) \xivec \rVert_2^{2+\delta}] < \infty$, \label{gcm:c2}
    \item $\sup_{P\in\calP} \max\left\{\Ex_P[\lVert\xivec \rVert_2^2 \mid \rX^S],
        \Ex_P[R(\h_P; \rY,\rX^S)^2 \mid \rX^S]\right\} < \infty$. \label{gcm:c3}
\end{enumerate}
Further, we require the following rate conditions on $M \coloneqq n^{-1}
\sum_{i=1}^n \lVert \hat\muvec(\rX^S_i) - \muvec(\rX^S_i)\rVert^2_2$ and $W \coloneqq n^{-1}
\sum_{i=1}^n \{R(\hat\h; \rY_{i}, \rX^S_{i}) - R(\h_P; \rY_{i},\rX^S_{i})\}^2$:
\begin{enumerate}[label=(\roman*)]
    \item $M = o_{\calP}(1)$,\label{gcm:c4}
    \item $W = o_{\calP}(1)$,\label{gcm:c5}
    \item $MW = o_{\calP}(n^{-1})$.\label{gcm:c6}
\end{enumerate}
Then $\rT_n$ converges to a standard $q$-variate normal distribution uniformly
over $\calP$. As a consequence, for all $\alpha \in (0, 1)$,
\begin{equation*}
    \lim_{n\to\infty} \sup_{P\in\calP} \Prob_P(
    p_{S,n}(\calD_n)
    \leq \alpha) = \alpha,
\end{equation*}
where
$p_{S,n}(\calD_n)$
is the $p$-value computed by Algorithm~\ref{alg:tramgcm}.
\end{theorem}
A proof is given in Appendix~\ref{proof:icp}.
%
Conditions (a)--(c) are mild regularity conditions on the distributions of
$(Y, \rX^S, \rE)$. Of the remaining conditions it is usually (iii) that is the
strictest. In the case of a parametric linear shift \tram{}, we
would expect $W=O_{\mathcal{P}}(n^{-1})$
and therefore would only need the regression of $\rE$ on $\rX^S$ to be
consistent. However, the \tram-GCM invariance test can still be correctly
calibrated even if the score residuals are learned at a slower-than-parametric
rate. Slower rates occur, for instance, in mixed-effects
\citep{tamasi2022tramme}, penalized linear shift
\citep{kook2020regularized}, or conditional
\tram{s} \citep{hothorn2014conditional}.

In Appendix~\ref{sec:results}, we demonstrate in a simulation
study that \tramicp-GCM and \tramicp-Wald are level at the nominal $\alpha$
and have non-trivial power (at least as high as nonparametric ICP) against
the considered alternatives in several model classes including binary logistic,
Weibull and Cox regression. However, the \tram-Wald invariance test hinges
critically on correct model specification. Despite its high power in the
simulation study (Appendix~\ref{sec:results}),
the \tram-Wald invariance test has size greater than its nominal level under
slight model misspecification (for instance, presence of a non-linear effect,
see Appendix~\ref{app:additive}).
The \tram-GCM test, however, directly extends to more flexible shift \tram{s}
which can incorporate the non-linearity, comes with theoretical guarantees,
and does not lead to anti-conservative behaviour under the null
when testing invariance. The robustness property of the
\tram-GCM invariance test does not necessarily hold without residualization
of the environments \citep{chern2017double,shah2020hardness}.
We illustrate empirically how also the naive correlation test may not be
level, in case of shift and penalized linear shift \tram{s}, in
Appendix~\ref{app:additive}.

\subsection{Practical aspects}\label{sec:practical}

\paragraph{Plausible causal predictors}
The procedure in Algorithm~\ref{alg:outer} can be used to compute $p$-values for
all $S \in \{1, \dots, d\}$. Based on \citet{peters2016causal} and as
implemented in \pkg{InvariantCausalPrediction} \citep{pkg:icp}, we can transform
the set-specific $p$-values into predictor-specific $p$-values: For all $j \in
[d]$, $\hat{p}_j \coloneqq 1$ if $\max_{S \subseteq [d]} p_{S,n}(\calD_n) <
\alpha$ and $\hat{p}_j \coloneqq \max_{S \subseteq [d] : j \not\in S}
p_{S,n}(\calD_n)$ otherwise. Now, for $j \in [d]$, $\hat{p}_j$ is a valid
$p$-value for the null hypothesis $H_0(j) : X^j \notin \pa(\rY)$ (assuming that
the true parents satisfy $(\pZ,\trafosub)$-invariance). We then refer to $X^j$
with $\hat{p}_j \leq \alpha$, $j \in [d]$ as \emph{plausible causal predictors}.

\paragraph{Unmeasured confounding}
In Setting~\ref{set:env}, we assume that all confounding variables between
covariates and response and all parents of the response have been measured. This
assumption can be weakened by instead assuming that there exists a subset of
observed ancestors $A \subseteq \an(\rY)$, such that $\rE \indep_{\calG^*} Y
\mid \rX^A$ (where $\indep_{\calG^*}$ denotes $d$-separation in $\calG^*$) and
the model for $\rY$ given $\rX^A$ is correctly specified by a \tram{}. Such
transformation models can be constructed in special cases
\citep{barbanti2019transformation,wienke2010frailty}, but a characterization of
this assumption is, to the best of our knowledge, an open problem. As in ICP in
the presence of hidden confounders \citep[][Proposition~5]{peters2016causal},
\tramicp{}, under this assumption, returns a subset of the ancestors of $Y$ with
large probability.

\paragraph{Nonparametric extension}
If the assumption that the response given its parents is correctly specified by
a TRAM is violated, we can still apply nonparametric approaches to estimate the
conditional CDF of $Y$ given $X^S$, $S \in [d]$. Appendix~\ref{app:nonp} shows
empirically that the \tram-GCM test based on score residuals obtained via
survival random forests \citep{ishwaran2008survforest} is level in a
data-generating process with right-censored responses where nonparametric ICP,
ignoring the censoring, is not. We leave a theoretical extension of our results
for shift \tram{s} to the nonparametric case for future work.

\section{Causal drivers of survival in critically ill adults}\label{sec:casestudy}

We apply \tramicp{} to the SUPPORT2 dataset \citep{knaus1995support} with
time-to-death in a population of critically ill hospitalized adults being the
response variable. SUPPORT2 contains data from 9105 patients of whom 68.1\% died
after a maximum follow-up of 5.55 years and the remaining 31.9\% of observations
were right-censored due to loss of follow-up. We consider the following
predictors measured at baseline (determined at most three days after hospital
admission): Sex (male/female), race (white, black, asian, hispanic, other),
number of comorbidities (0--9; \code{num.co}), coma score (0--100,
\code{scoma}), cancer (no cancer, cancer, metastatic cancer; \code{ca}), age
(years), diabetes (yes/no), dementia (yes/no), disease group (nine groups,
including colon and lung cancer; \code{dzgroup}).\footnote{\code{ca} is not a
deterministic function of \code{dzgroup}.} For our analysis, we treat
\code{num.co} (0, 1, $\dots$, 5, 6 or more) and \code{scoma} (11 levels) as
factors, square-root transform age and omit 43 patients with missing values in
any of the predictors listed above. We apply \tramicp{} using both \tram-GCM and
\tram-Wald. For \tram-Wald, we only test the presence of main effects of the
environments (without additional first-order interaction effects) due to
non-convergence when fitting the models with interaction effects.

\subsection{Choice of Environments}

When applying oracle tests and assuming faithfulness, \tramicp{} maintains the
coverage guarantee as long as the environment variables are non-descendant of
the response \citep[][Section~3.3]{peters2016causal}. In our study, all measured
predictors precede the response chronologically, so, if all model assumptions
are satisfied and faithfulness holds, all choices of environments come with the
correct coverage but may differ in power. We choose \code{num.co} as the
environment as, we believe, it is associated with several other predictors and
subsequently creates enough heterogeneity. In addition, because \code{num.co} is
constructed from the presence/absence of other (recorded and unrecorded)
comorbidities, it is a sink node in the corresponding graph. If an unrecorded
comorbidity or \code{num.co} were to directly cause (time to) death, the
population output of TRAMICP (assuming faithfulness) would be empty since the
path from \code{num.co} to (time to) death cannot be blocked without
conditioning on the presence/absence of this comorbidity itself. In
Appendix~\ref{app:casestudy:menv}, we apply \tramicp{} when additionally using
\code{race} as an environment. (For a single choice of a valid environment, no
multiple testing correction is needed; however, when applying \tramicp{} to
several choices of environments, in order to obtain a family-wise coverage
guarantee, one would need to apply a multiple testing correction, such as
Bonferroni with the number of choices of environments.)

\subsection{Results}\label{sec:casestudy:results}

\paragraph{The set of all predictors is not invariant}
In the model including all predictors the standard Wald test rejects the null
hypothesis of no effect for all predictors except \code{race}. A Wald test for
the main effect of \code{num.co} yields a $p$-value $< 0.0001$. This provides
strong evidence that the purely predictive model using all predictors is not
invariant across \code{num.co} and thus uses a set of features that is different
from the set of causal parents.

\paragraph{Evidence of age and cancer being direct causes of time-to-death}
We now apply \tramicp-GCM and \tramicp-Wald to the SUPPORT2 dataset specifying
the survival time as the response in a Cox proportional hazard model, using
\code{num.co} as the environment and including all other predictors. Both
algorithms output \code{ca} and \code{age} as plausible causal predictors (\ie
the intersection of all sets for which the invariance test was not rejected
equals $\{\code{ca}, \code{age}\}$). This can be seen in
Figure~\ref{fig:casestudy} in Appendix~\ref{app:casestudy:fig}, where all
non-rejected sets include both \code{ca} and \code{age}. The predictor-specific
$p$-values (see Section~\ref{sec:practical}) are given in
Table~\ref{tab:casestudy} (`Evidence of age and cancer being direct causes of
time-to-death'). In their original analysis of the SUPPORT2 dataset,
\citet{knaus1995support} have assumed that the censoring is uninformative. In a
sensitivity analysis in Appendix~\ref{app:casestudy:cens}, we show that while
\tramicp{} is somewhat robust when inducing (potentially) additional informative
censoring, it eventually returns the empty set.

\begin{table}[!ht]
\centering
\caption{%
\tramicp{} applied to the SUPPORT2 dataset in the different settings described
in Section~\ref{sec:casestudy}. Predictor-specific $p$-values (see
Section~\ref{sec:practical}) are reported for the \tram-GCM and \tram-Wald
invariant test, together with the environment variable used. $p$-values in bold
are significant at the 5\% level; in each row, the set of predictors with bold
numbers corresponds to the output of \tramicp{}.
}\label{tab:casestudy}
\resizebox{\textwidth}{!}{%
\begin{tabular}{@{}lrrrrrrrrc@{}}
\toprule
\textbf{Invariance test}
& \multicolumn{8}{c}{\textbf{Predictor-specific $p$-values}}&
\multicolumn{1}{l}{\textbf{Environment}} \\ \midrule
& \multicolumn{1}{l}{\code{scoma}} & \multicolumn{1}{l}{\code{dzgroup}} & \multicolumn{1}{l}{\code{ca}} & \multicolumn{1}{l}{\code{age}} & \multicolumn{1}{l}{\code{diabetes}} & \multicolumn{1}{l}{\code{dementia}} & \multicolumn{1}{l}{\code{sex}} & \multicolumn{1}{l}{\code{race}} & \multicolumn{1}{l}{} \\ \cmidrule(lr){2-9}
\multicolumn{10}{l}{\textit{Evidence of age and cancer being direct causes of time-to-death}}\\
\tram-GCM  & 0.239 & 0.239 & \textbf{0.000} & \textbf{0.003} & 0.157 & 0.176 & 0.162 & 0.220 & \multirow{2}{*}{\code{num.co}} \\
\tram-Wald & 0.127 & 0.127 & \textbf{0.000} & \textbf{0.001} & 0.080 & 0.077 & 0.089 & 0.127 &\\ \midrule
\multicolumn{10}{l}{\textit{Incorporating prior knowledge about direct causes}}\\
\tram-GCM  & 0.273 & 0.273 & \textbf{0.000} & - & - & - & 0.163 & 0.216 & \multirow{2}{*}{\code{num.co}} \\
\tram-Wald & 0.127 & 0.127 & \textbf{0.000} & - & - & - & 0.089 & 0.127 &\\
\bottomrule
\end{tabular}
}
\end{table}

\paragraph{Incorporating prior knowledge about direct causes}
If a set of predictors is known to cause the outcome, this set can always be
included in the conditioning set (which reduces computational complexity,
because fewer invariance tests have to be performed). We illustrate this by
including \code{age}, \code{dementia}, and \code{diabetes} as `mandatory'
covariates when running \tramicp{} (see Appendix~\ref{app:pkg}). In this case,
both \tramicp-GCM and \tramicp-Wald still output \code{ca} as a causal predictor
of survival. The predictor $p$-values are given in Table~\ref{tab:casestudy}
(`Incorporating prior knowledge about direct causes').

\section{Discussion} \label{sec:discussion}

In this paper, we generalize invariant causal prediction to transformation
models, which encompass many classical regression models and different types of
responses including categorical and discrete variables. We show that, despite
most of these models being neither closed under marginalization nor collapsible,
\tramicp{} retains the same theoretical guarantees in terms of identifying a
subset the causal parents of a response with high probability. We generalize the
notion of invariance to discrete and categorical responses by considering score
residuals which are uncorrelated with the environment under the null hypothesis.
Since score residuals remain well-defined for categorical responses, our
proposal is one way to phrase invariance in classification settings.

We have applied \tramicp{} to roughly ten real world data sets (which
technically would require a multiple testing correction), and have often
observed that, depending on the choice of environment, either no subset of
covariates is invariant (\ie all invariance tests are rejected) or all subsets
of covariates are invariant. In both cases, \tramicp{} outputs the empty
set---an output that is not incorrect but uninformative.

\section*{Acknowledgments}
We thank Niklas Pfister and Alexander Mangulad Christgau for insightful
discussions. We would also like to thank Juraj Bodik for helpful comments. The
research of LK was supported by the Swiss National Science Foundation (SNF;
grant no. 214457). LK carried out part of this work at the University of
Copenhagen and University of Zurich. During parts of this research project, SS
and JP worked at the University of Copenhagen and were supported by a research
grant (18968) from the VILLUM Foundation. ARL is supported by a research grant
(0069071) from Novo Nordisk Fonden.

\appendix
\renewcommand{\thesection}{\Alph{section}}
\renewcommand{\thesubsection}{\Alph{section}\arabic{subsection}}
\counterwithin{figure}{section}
\renewcommand\thefigure{\thesection\arabic{figure}}
\counterwithin{table}{section}
\renewcommand\thetable{\thesection\arabic{table}}
\counterwithin{algorithm}{section}
\renewcommand\thealgorithm{\thesection\arabic{algorithm}}

\section*{Appendix}

The appendix is structured as follows:
\begin{itemize}
\item \textbf{Appendix~\ref{app:tramicp}: \nameref{app:tramicp}}
    \begin{itemize}
        \item[\bf \ref{sec:wald}]: \nameref{sec:wald}
        \item[\bf \ref{app:collaps}]: \nameref{app:collaps}
        \item[\bf \ref{app:addex}]: \nameref{app:addex}
        \item[\bf \ref{app:id}]: \nameref{app:id}
        \item[\bf \ref{app:empty}]: \nameref{app:empty}
        \item[\bf \ref{app:cens}]: \nameref{app:cens}
    \end{itemize}
\item \textbf{Appendix~\ref{app:computationaldetails}: \nameref{app:computationaldetails}}
    \begin{itemize}
        \item[\bf \ref{sec:existingmethods}]: \nameref{sec:existingmethods}
        \item[\bf \ref{sec:comp}]: \nameref{sec:comp}
        \item[\bf \ref{sec:results}]: \nameref{sec:results}
        \item[\bf \ref{app:res}]: \nameref{app:res}
        \item[\bf \ref{app:faith}]: \nameref{app:faith}
        \item[\bf \ref{app:additive}]: \nameref{app:additive}
        \item[\bf \ref{app:nonp}]: \nameref{app:nonp}
    \end{itemize}
\item \textbf{Appendix~\ref{app:casestudy}: \nameref{app:casestudy}}
    \begin{itemize}
        \item[\bf \ref{app:casestudy:fig}]: \nameref{app:casestudy:fig}
        \item[\bf \ref{app:casestudy:menv}]: \nameref{app:casestudy:menv}
        \item[\bf \ref{app:casestudy:cens}]: \nameref{app:casestudy:cens}
    \end{itemize}
\item \textbf{Appendix~\ref{app:pkg}: \nameref{app:pkg}}
\item \textbf{Appendix~\ref{app:theory}: \nameref{app:theory}}
    \begin{itemize}
        \item[\bf \ref{app:proofs}]: \nameref{app:proofs}
        \item[\bf \ref{app:lemmata}]: \nameref{app:lemmata}
    \end{itemize}
\end{itemize}

\section{Additional information and results}\label{app:tramicp}

\subsection{Invariance tests based on the Wald statistic}\label{sec:wald}

For linear shift \tram{s}, $(\pZ,\trafolin)$-invariance
implies the absence of main and interaction effects involving the
environments if we include the environment variable as a covariate into the model
(see Proposition~\ref{prop:wald} below); this can be tested for using a Wald
test for both continuous and discrete responses (Algorithm~\ref{alg:tramicp}).

We now introduce notation for including main and interaction effects
involving the environments. Let $\otimes$ denote the Kronecker product and
define, for all $S \subseteq [d]$, $m^S : \calX^S \times \calE \to \RR^{q(1 +
\lvert S \rvert)}$, $(\rx^S, \evec)
\mapsto (1, (\rx^S)^\top)^\top \otimes \evec$, which sets up first order interaction
between environments $\evec$ and covariates $\rx^S$ together with environment main
effects.
Let $\trafos{\calX\times\calE}^*$ denote the class of all transformation
functions on $\calY \times (\calX \times \calE) \subseteq \RR \times \RR^d \times \RR^q$
(see Definition~\ref{def:tram}).
For a fixed vector of basis functions
$\basisy : \calY \to \RR^b$
(see Section~\ref{sec:comp} for typical choices of
bases and their correspondence to commonly used regression models), we
define
$\trafowald(\basisy) \coloneqq \{\h \in
\trafos{\calX\times\calE}^* \mid \exists
(\thetavec, \shiftparm, \gammavec) \in \Theta \times \RR^{d} \times
\RR^{q(1 + d)} : \forall \rx \in \calX, \evec \in \calE: \h (\bcd
\mid \rx, \evec) = \basisy(\bcd)^\top\thetavec + \rx^\top\shiftparm + m(\rx,
\evec)^\top \gammavec\}$
where $\Theta \subseteq \overline{\RR}^b$.
Thus, the transformation functions are parameterized
by $\parm \coloneqq (\thetavec, \shiftparm, \gammavec) \in \Theta
\times \RR^{d} \times \RR^{q(1 + d)}$.
For all $S \subseteq [d]$, this global model class induces subclasses
$\trafowaldS(\basisy)$ and $\trafos{\calX^S}^{\text{Wald}}(\basisy)$ as
described in Section~\ref{sec:sctram}. For $(\pZ, \calY, \calX\times\calE,
\trafowald(\basisy))$ satisfying Assumption~\ref{asmp:densities},
we define the log-likelihood function
$\ell^{\text{Wald}} : \trafowald(\basisy) \times \calY \times \calX \times \calE
\to \RR$ with $\ell : (\h, \ry, \rx, \evec) \mapsto \log f_{\rY\mid\rX=\rx,
\rE=\evec}(\ry; \h)$. For all subsets of covariates $S\subseteq[d]$, we then
estimate the \tram{} $\pZ \circ \h_{\parm^S}$, $\h_{\parm^S} \in
\trafowaldS(\basisy)$ and consider the hypothesis test $H_0(S): \gammavec^S = 0$.

\begin{proposition}\label{prop:wald}\sloppy
Assume Setting~\ref{set:env} with $\trafosub = \trafos{\calX}^{\text{Wald}}(
\basisy)$. Let $S \subseteq [d]$ be given and suppose that the canonical
conditional CDF of
$Y$ given $\rX^S$ and $\rE$ is an element of $\calM(\pZ,\calY,\calX^S\times\calE,
\trafowaldS)$; that is, there exist $(\thetavec, \shiftparm^S, \gammavec^S)
\in \Theta \times \RR^{\lvert S \rvert} \times \RR^{q(1 + \lvert S
\rvert)}$ \st the canonical conditional CDF equals $\pZ(\basisy(\bcd)^\top
\thetavec + (\rx^S)^\top\shiftparm^S + m^S(\rx^S, \evec)^\top\gammavec^S)$.
Then $S$ is $(\pZ, \trafos{\calX}^{\text{Wald}})$-invariant if and
only if $\gammavec^S = 0$.
\end{proposition}

A proof is given in Appendix~\ref{proof:wald}.

The Wald test uses the quadratic test statistic
$(\hat\gammavec^S)^\top\hat\Sigma_{\hat\gammavec^S}/n\hat\gammavec^S$
which converges in distribution to a $\chi^2$-distribution with
$\rank\hat\Sigma_{\hat\gammavec^S}$ degrees of freedom
\citep[under further regularity conditions, see][Theorems~1--3]{hothorn2018most}.
Here, $(\hat\Sigma_{\hat\gammavec^S})_{ij}
\coloneqq [\rI(\h_{\hat\parm^S}; \rX^S, \rE)]^{-1}_{ij}$, $i,j \in
\{\lvert S \rvert + 1, \dots, \lvert S\rvert + q(1 + \lvert S\rvert)\}$,
denotes the estimated variance-covariance matrix of the model restricted to
the main and interaction effects involving the environments, where
$\rI(\h_{\hat\parm^S};\rX^S,\rE)$ denotes an estimate of the Fisher
information,
which for all $S \subseteq [d]$ and
$\parm^S \in \Theta \times \RR^{\lvert
S\rvert} \times \RR^{q(1 + \lvert S\rvert)}$, is defined as $\rI(\h_{\parm^S};
\rX^S, \rE) \coloneqq \Ex\left[-\frac{\partial}{\partial \parm^S\partial
(\parm^S)^\top} \ell^{\text{Wald}}(\h_{\parm^S}; \rY, \rX^S, \rE)\mid\rX^S,
\rE\right]$.

\begin{algorithm}[!ht]
\caption{\tram-Wald invariance test}\label{alg:tramicp}
\begin{algorithmic}[1]
\Require Data $\calD_n$ from Setting~\ref{set:env}, $S \subseteq [d]$,
$\trafos{\calX^S\times\calE} \subseteq \trafowaldS(\basisy)$
\State Fit the \tram{}:
$\h_{\hat\parm^S_n} \gets \argmax_{\h_{\parm^S} \in
\trafos{\calX^S\times\calE}} \ell(\h_{\parm^S}; \calD_n)$
\State Compute the variance-covariance matrix:
$\hat\Sigma_{\hat\gammavec^S_n} \gets [\rI(\h_{\hat\parm^S_n};
\calD_n)]^{-1}_{\hat\gammavec^S}, \quad K_S \gets \rank \hat\Sigma_{\hat\gammavec^S}$
\State Compute Wald $p$-value:
$p_{S,n}(\calD_n) \gets 1 - F_{\chi^2(K_S)}\{(\hat\gammavec^S_n)^\top
\hat\Sigma_{\hat\gammavec^S_n}\hat\gammavec^S_n\}
$
\State \Return{$p_{S,n}(\calD_n)$}
\end{algorithmic}
\end{algorithm}

\subsection{Collapsibility and closure under marginalization}\label{app:collaps}

Linear shift \tram{}s, like the general additive noise model, are in general
not closed under marginalization. That is, when omitting variables from a correctly
specified \tram{}, the reduced model is generally no longer a \tram{} of
the same family and linear shift structure. This is closely related to
non-collapsibility. Non-collapsibility has been studied for contingency tables
\citep{whittemore1978collapsibility},
logistic-linear models \citep{guo1995collapsibility} and rate/hazard ratio models
\citep{greenland1996absence,martinussen2013collapsibility}. In non-collapsible
models, the marginal effect measure (obtained from regressing $\rY$ on $\rX$) is
not the same as in the conditional effect measure (the regression additionally
adjusting for $\rV$) even if $\rX$ and $\rV$ are independent. For instance, in
binary logistic regression, the conditional odds ratio given a binary exposure
across a conditioning variable with two strata is in general not the same as the
marginal odds ratio and may even flip sign
\citep[which is related to Simpson's paradox,][]{hernan2011simpson}.
Even if the marginal model is correctly specified, the conditional model
may not be correctly specified, and vice versa, which has been illustrated in
binary GLMs \citep{robinson1991logistic} and the Cox proportional hazards model
\citep{ford1995cox}.
Collapsibility depends on the chosen link function. In GLMs, the identity- and
log-link give rise to collapsible effect measures (difference in expectation, and
log rate-ratios, respectively). In transformation models, the most commonly used
link functions (logit, cloglog, loglog, probit) generally yield non-collapsible
effect measures \citep{gail1984biased,gail1986adjusting,greenland1999nc}.

It is possible to connect collapsibility with conditional independence
statements. For instance, the conditional odds ratio of a covariate $X$ for a
binary outcome $\rY$ is collapsible over another (discrete) covariate $V$ if
either $X \indep V \mid \rY$ or $\rY \indep V \mid X$
\citep{greenland1999nc}.
Recently, collapsibility has also been discussed in the context of causal
parameters \citep{didelez2022collapsibility}.

\subsection{Additional examples}\label{app:addex}

\begin{example}[Count regression]\label{ex:count}
Count random variables take values in $\calY\coloneqq\{0, 1, 2, \dots\}$. Without
restricting the distribution for $F_{\rY \mid \rX^{S_*} = \rx^{S_*}}$ to a
parametric
family, we can formulate models for the count response via $F_{\rY \mid
\rX^{S_*} = \rx^{S_*}}(\bcd) = \pZ(\hY(\bcd) - (\rx^{S_*})^\top \shiftparm)$,
where $\hY$ is
an increasing step function (with jumps at points in $\calY$)
and $\pZ$ a user-specified continuous cumulative
distribution function with log-concave density \citep[log-concavity ensures
uniqueness of the maximum likelihood estimator,][]{siegfried2020count}.

For count responses, we can define a family of linear shift \tram{s}
with support $\calY \coloneqq \{0, 1, 2, \dots\}$, given by $\calM(\pZ, \calY,
\calX, \trafolin)$. For all $\rx \in \calX$, the
transformation function $\h(\bcd \mid \rx)
\coloneqq \hY(\bcd) - \rx^\top\shiftparm$ is a right-continuous step function
with steps at the integers and linear shift effects.
The log-likelihood contribution for a single observation $(\ry, \rx)$ is given
by $\log\pZ(\h(0 \mid \rx))$ if $\ry = 0$ and $\log(\pZ(\h(\ry \mid \rx))
- \pZ(\h(\ry - 1 \mid \rx))$ for $\ry \geq 1$. The exact form of score
residual depends on the choice of $\pZ$. The (generalized) inverse
transformation function is given by $\h^{-1} : (\rz, \rx) \mapsto
\lfloor\hY^{-1}(\rz + \rx^\top\shiftparm)\rfloor$.
\end{example}

\begin{example}[Parametric survival regression]
\label{ex:surv}
Understanding the causal relationship between features and patient survival is
sought after in many biomedical applications. Let $\rY$ be a (strictly) positive
real-valued random variable, \ie $\calY \coloneqq \RR_{+}$. The Weibull
proportional hazards model is defined via $F_{Y\mid\rX^{S_*}=\rx^{S_*}}(\bcd) = 1 -
\exp(-\exp(\eparm_1 + \eparm_2 \log(\bcd) - (\rx^{S_*})^\top\shiftparm))$, with $\eparm_1
\in \RR, \eparm_2 > 0$. Here, $\shiftparm$ is
interpretable as a vector of log hazard ratios \citep{kleinbaum2012parametric}.
The model can be written as $F_{\rY\mid\rX^{S_*}=\rx^{S_*}}(\bcd) = F(\hY(\bcd) -
(\rx^{S_*})^\top\shiftparm)$, where $F(\bcd) \coloneqq 1 - \exp(-\exp(\bcd))$ denotes
the standard minimum extreme value distribution and $\hY(\bcd)\coloneqq \eparm_1 +
\eparm_2 \log(\bcd)$. The Cox proportional hazard model
\citep{cox1972regression} is obtained as an extension of the Weibull
model by allowing $\hY$ to be a step function (with jumps
at the observed event times) instead of a log-linear function.

For the Weibull proportional hazards model, we fix $\pZ(\rz) = 1 - \exp(-\exp(
\rz))$ and define a family of log-linear transformation functions
$\trafosub^{\text{log-lin}} \coloneqq \{\h \in
\trafolin \mid \exists (\eparm_1, \eparm_2) \in
\RR \times \RR_+ : \forall \ry \in \calY, \rx \in \calX \ \h(\ry \mid \rx)
= \eparm_1 + \eparm_2 \log(\ry) - \rx^\top\shiftparm\}$.
The log-likelihood contribution for an exact response is given by the
log-density and an uninformatively right-censored observation at time $t$
with covariates $\rx \in \calX$ contributes the log-survivor function evaluated
at $t$, \ie $\log(1 - \pZ(\h(t\mid \rx)))$, to the log-likelihood. The inverse
transformation function is given by $(\rz, \rx) \mapsto \hY^{-1}(\rz +
\rx^\top\shiftparm) \coloneqq \exp(\eparm_2^{-1}(\rz - \eparm_1 +
\rx^\top\shiftparm))$.
\end{example}

\subsection{Full identifiability of the causal parents}\label{app:id}

\begin{definition}[Full identifiability of the causal parents]\label{def:pointid}\sloppy
Let $\calC$ denote a class of structural causal models.
Then, the set of causal parents is said to be
$\calC$-\emph{fully identifiable}
if for all pairs $C_1,C_2 \in \calC$ it holds that
\begin{equation}
    \Prob^{C_1}_{(Y,\rX)} = \Prob^{C_2}_{(Y,\rX)} \implies
    \pa_{C_1}(Y) = \pa_{C_2}(Y).
\end{equation}
\end{definition}

\subsection[If an environment is a causal parent]{If an environment variable is
a causal parent}\label{app:empty}

In case $\rE$ is a causal parent of $Y$ but not included in the set of
predictors for causal feature selection, \tramicp{}, like nonparametric ICP and
other versions of ICP, outputs the empty set, assuming faithfulness and oracle
invariance tests.
Proposition~\ref{prop:empty} shows that, (i) assuming the induced distribution
is faithful w.r.t.\ the induced graph, the set of identifiable causal predictors
is empty; and (ii) additionally assuming oracle invariance tests, \tramicp{}
outputs the empty set with probability 1. Lastly, Example~\ref{ex:empty:boxcox}
shows that, if faithfulness is violated, nonparametric ICP and \tramicp{} can
produce non-empty and incorrect outputs in Box--Cox-type transformation models.
However, this counterexample is contrived in that it makes use of specifically
tuned coefficients.
This may not come as a surprise in that
faithfulness violations
occur with probability zero when choosing the coefficients of a linear Gaussian SCM randomly
according to a distribution that is absolutely continuous with respect to Lebesgue measure
\citep[Theorem~3.2]{spirtes2000causation}.
We hypothesize that similar statements hold for
linear shift \tram{s}, too.
Indeed, as discussed in Appendix~\ref{app:faith}, faithfulness violations
did not occur empirically when sampling the coefficients of linear shift
\tram{s} from a continuous distribution.

\begin{setting}[Data from multiple environments]\label{set:empty}\sloppy
Let $\calY$, $\calX$, $\pZ \in \calZ$ be as in Definition~\ref{def:tram} and let
$\trafos{\calX\times\calE} \subseteq \trafos{\calX\times\calE}^*$ be a class of
transformation functions such that Assumptions~\ref{asmp:densities}
and~\ref{asmp:closure} hold. Let $C_*$ be a structural causal \tram{}
(Definition~\ref{def:tramscm}) over $(\rY, \rX, \rE)$ such that
\begin{align*}
C_* \coloneqq \begin{cases}
    E^k \coloneqq \ m_k(\rX, N_{E^k}), \quad \forall k \in [q] \\
    X^j \coloneqq \ g_j(\rX, \rE, \rY, N_{X^j}), \quad  \forall j \in [d] \\
    Y \ \coloneqq \ \h^{-1}_*(Z \mid \rX^{S_*}, \rE),
\end{cases}
\end{align*}
where $\h_* \in \trafos{\calX^{S_*}\times\calE}$ and $(Z, N_\rX, N_\rE)$ denotes
the jointly independent noise variables. In this setup, the random vector $\rE$
encodes the environments and takes values in $\calE \subseteq \RR^q$. We further
assume that the parents of $\rE$ can only be non-descendants of $Y$ in
$\calG_*$; $\rE$ may be discrete or continuous.
\end{setting}

\begin{proposition}\label{prop:empty}
Assume Setting~\ref{set:empty} and that the induced distribution is faithful
w.r.t.\ the induced graph. Then, (i) $S_I = \emptyset$. Further assume that, for
all $S \subseteq [d]$, $p_{S,n}(\calD_n) = 0$ a.s.\ if $S$ is not
$(\pZ,\trafosub)$-invariant. Then, (ii) for all $\alpha \in (0, 1)$,
\[
\Prob(S_n^{\phi}\subseteq\pa_{C_*}(Y)) = 1.
\]
\end{proposition}
A proof is given in Appendix~\ref{proof:empty}.

\begin{example}\label{ex:empty:boxcox}
Consider a structural causal linear shift \tram{} with the following
assignments:
\begin{align*}
    E &\coloneqq N_E\\
    X &\coloneqq E + N_X\\
    Y &\coloneqq g(X + N_Y),
\end{align*}
where $N_E, N_Y, N_X$ are jointly independent and standard normally distributed
and, for all $z \in \RR$, $g(z) \coloneqq F_{\chi^2_5}^{-1}(\Phi(z))$. It
follows that $Y \indep E \mid X$ and $Y$ given $X$ is correctly specified by a
linear shift \tram{}. However, the following SCM (allowing for dependent errors)
induces the same observational distribution:
\begin{align*}
    \tilde E &\coloneqq N_{\tilde E}\\
    \tilde Y &\coloneqq \tilde g(\tfrac{1}{\sqrt{2}} \tilde E + N_{\tilde Y})\\
    \tilde X &\coloneqq (1, \tilde Y, \tilde E) \bvec + N_{\tilde X},
\end{align*}
where $N_{\tilde E} = N_E$, $\tilde N_Y \coloneqq \tfrac{1}{\sqrt{2}}(g^{-1}(Y)
- E) = \tfrac{1}{\sqrt{2}}(N_X + N_Y)$ are the population score residuals from
regressing $Y$ on $E$, for all $z\in\RR$, $\tilde g(z) \coloneqq g(\sqrt{2} z)$
is the baseline transformation, and $\tilde N_X \coloneqq X - (1, Y, E) \bvec$,
where $\bvec$ is the population OLS coefficient from regressing $X$ on $Y$ and
$E$. In this \tram-SCM, it holds that $\tilde E \indep \tilde Y \mid \tilde X$
(which violates faithfulness) and
oracle ICP outputs (incorrectly) that $\{\tilde X\}$ is a subset of the causal
parents: $\tilde Y$ given $\tilde X$ is correctly specified by a shift \tram{}
with $\pZ = \Phi$, and, by the conditional independence $\tilde E \indep \tilde
Y \mid \tilde X$, $\{\tilde X\}$ is $(\pZ, \trafolin)$-invariant and \tramicp{}
thus outputs $\{\tilde X\}$. However, the violation of faithfulness (which
allows for the incorrect result) relies on fine-tuned coefficients: If the
coefficient for $\tilde E$ in the structural equation of $\tilde Y$ changes
slightly, the output is more likely to be empty (see Table~\ref{tab:empty}).
\begin{table}[!ht]
\centering
\begin{tabular}{lr}
\toprule
\bf Coefficient $\tilde E \to \tilde Y$ & \bf Fraction output equals $\{X\}$\\
\midrule
$\tfrac{1}{\sqrt{2}} - 0.2$  & 0.00\\
$\tfrac{1}{\sqrt{2}} - 0.1$  & 0.00 \\
$\tfrac{1}{\sqrt{2}}$        & 1.00 \\
$\tfrac{1}{\sqrt{2}} + 0.1$  & 0.00 \\
$\tfrac{1}{\sqrt{2}} + 0.2$  & 0.00 \\
\bottomrule
\end{tabular}
\caption{%
Fraction \code{BoxCoxICP} outputs $\{\tilde X\}$ in the structural causal
\tram{} of Example~\ref{ex:empty:boxcox},
when changing the coefficient for $\tilde E$
in the structural equation of $\tilde Y$. The results are for a sample size of
3000 and 20 repetitions.
}\label{tab:empty}
\end{table}
\end{example}

\subsection{Uninformatively censored responses}\label{app:cens}

Censoring often occurs in studies with time-to-event outcomes
\citep{kleinbaum2012parametric,kalbfleisch2011statistical}.
Here, we extend \tramicp{} to uninformatively censored responses.

\begin{setting}[Censoring]\label{set:cens}\sloppy
Assume Setting~\ref{set:env} in which $Y$ is unobserved. Instead, we have
access to additional variables $L,U\in\calY$ with joint density $p_{(L,U)}$
\wrt $\mu \otimes \mu$ (see below Assumption~\ref{asmp:densities}
for the definition of $\mu$) and which fulfill $L < U$ with probability 1.
Further, we assume that
$(L, U) \indep \rY \mid \rX^{S_*}$,
\ie \emph{uninformative} censoring \citep{kalbfleisch2011statistical}. In
addition to $L$ and $U$, we observe the indicator random variables $\delta_L$
and $\delta_I$ which carry the information about whether an observation is
left-censored $Y \in (-\infty, L]$ ($\delta_L \coloneqq \1(\rY \leq L)$),
interval-censored $Y \in (L, U]$ ($\delta_I \coloneqq \1(L < \rY \leq U)$), or
right-censored $\rY \in (U, \infty]$ ($1 - \delta_I - \delta_L = \1(Y > U)$).
Let $\mu_{\{0,1\}}$ denote the counting measure on $\{0,1\}$. For
$\Prob_{\rX^{S_*}}$-almost all $\rx^{S_*}$ and all $l, u \in \calY$ s.t.~$l < u$
and $d_L, d_I \in \{0,1\}$, the joint density of $(\delta_L, \delta_I, L,U)$
(\wrt to the product measure $\mu_{\{0, 1\}} \otimes \mu_{\{0,1\}} \otimes \mu
\otimes \mu$) conditional on $\rX^{S_*}=\rx^{S_*}$, exists, and is denoted by
$p$.
\end{setting}

Under Setting~\ref{set:cens}, we have
\begin{align}
   &p(\delta_L, \delta_I, l, u \mid \rX^{S_*} = \rx^{S_*}) \\
   &= p_{(L,U) \mid \rX^{S_*}}(l, u \mid \rx^{S_*})
       p_{(\delta_L, \delta_I) \mid L, U, \rX^{S_*}}(\delta_L,
       \delta_I \mid l, u, \rx^{S_*}) \\
   &=
   \begin{cases}
      p_{(L,U) \mid \rX^{S_*}}(l,u \mid \rx^{S_*})
        (\pZ(\h_*(l\mid\rx^{S_*}))) & \text{if }\delta_L = 1, \delta_I = 0, \\
       p_{(L,U) \mid \rX^{S_*}}(l,u \mid \rx^{S_*})
        (\pZ(\h_*(u\mid\rx^{S_*}))- \pZ(\h_*(l\mid\rx^{Z_*}))) & \text{if }\delta_L = 0, \delta_I = 1, \\
       p_{(L,U) \mid \rX^{S_*}}(l,u \mid \rx^{S_*})
        (1 - \pZ(\h_*(u\mid\rx^{S_*}))) & \text{if }\delta_L = 0, \delta_I = 0, \\
       0 & \text{if }\delta_L = 1, \delta_I = 1,
   \end{cases}
\end{align}
which follows from
considering the four cases of the values for
$(\delta_L,\delta_I)$,
\eg $\Prob(\delta_L = 1, \delta_I = 0 \mid L = l, U = u, \rX^{S_*} = \rx^{S_*}) =
\Prob(\delta_L = 1 \mid L = l, \rX^{S_*} = \rx^{S_*}) =
\Prob(Y \leq l \mid \rX^{S_*}=\rx^{S_*}) = \pZ(\h_*(l\mid
\rx^{S_*}))$.
The log-likelihood is given by the log conditional
density (and is undefined for the case $\delta_L = \delta_I = 1$)
\begin{align}
\tilde\ell &:\trafosub\times\{(0,1), (1, 0), (0, 0)\}\times\calY^2\times \calX\to\RR,\\
\tilde\ell &:(\h; \delta_L,\delta_R, l, u, \rx) \mapsto \delta_L \log\pZ(\h(l\mid
\rx)) + \delta_I \log (\pZ(\h(u\mid\rx))-\pZ(\h(l\mid\rx)))\\
&\quad + (1 - \delta_L - \delta_I) \log (1 - \pZ(\h(u\mid\rx))) +
\log p_{(L,U) \mid \rX^{S_*}}(l,u \mid \rx^{S_*}).
\end{align}
The definition of $(\pZ,\trafosub)$-invariance and
Proposition~\ref{prop:parents} still hold when considering the unobserved $\rY$ in
Setting~\ref{set:cens}. However, when using \tramicp{}, since $\rY$ is unobserved,
we cannot test statements about $\rY$ directly---the censoring needs to
be taken into account.
For $(\pZ, \calY, \calX, \calH_{\calY,\calX})$ satsifying
Assumptions~\ref{asmp:densities} and~\ref{asmp:closure}, uninformatively censored
score residuals can be defined via $\tilde R : \trafosub \times \{(0,1), (1, 0), (0, 0)\}
\times\calY^2 \times\calX\to\RR$, with
$\tilde R : (\h; \delta_L, \delta_I, l, u,
\rx) \mapsto \frac{\partial}{\partial\alpha}
\tilde\ell(\h + \alpha; \delta_L, \delta_I, l, u, \rx)$. The following
proposition extends
invariance based on score residuals (Proposition~\ref{prop:residualinv}) to
uninformatively censored responses.
%
\begin{proposition}[Score-residual-invariance for uninformatively censored responses]
\label{prop:rinvcens}
Assume Setting~\ref{set:cens}. Then, we have the following implication:
\begin{align}
S \mbox{ is } (\pZ,\trafosub)\mbox{-invariant}
\implies
&\Ex[\tilde R(\h^S; \delta_L, \delta_I, L, U, \rX^S) \mid \rX^S] = 0,
\mbox{ and}\\
&\Ex[\Cov[\rE, \tilde R(\h^S; \delta_L, \delta_I, L, U, \rX^S)\mid\rX^S]] = 0.
\end{align}
\end{proposition}

A proof is given in Appendix~\ref{proof:rinvcens}. Establishing theoretical
results on invariance in case of informative censoring is left for future work.

\section{Simulation results and computational details}\label{app:computationaldetails}

We now evaluate the proposed algorithms on simulated data based on randomly
chosen graphs with restrictions on the number of possible descendants and
children of the response, as well as how the binary environment indicator
affects non-response nodes. In the simulations, the conditional distribution
of $Y \mid \pa(Y)$ is correctly specified by a \tram{}, all other structural
equations are linear additive with Gaussian noise.

\subsection{Existing methods}\label{sec:existingmethods}

\subsubsection{Nonparametric ICP via conditional independence testing}\label{sec:comparators}

Throughout, we report the oracle version on nonparametric ICP. Under
Setting~\ref{set:env}, let $\calG$ denote the DAG implied by $C_*$. Assuming
the Markov property \citep[see][p.~29]{spirtes2000causation})
and faithfulness of $\Prob_{(\rY,\rX,\rE)}^{C_*}$
\wrt $\calG$, the oracle is defined as the intersection of sets $S$ for which
$Y$ is conditionally independent of $\rE$ given $\rX^S$
\citep[Proposition~4.1]{mogensen2022invariant},
\begin{align}\label{eq:oicp}
S^{\operatorname{ICP}} = \pa(\rY) \cap (\ch(\rE) \cup \pa(\an(\rY) \cap
\ch(\rE))),
\end{align}
where $\pa(\bcd)$, $\ch(\bcd)$, $\an(\bcd)$ denote parents, children and
ancestors of a node in $\calG$, respectively.

A general-purpose algorithm for causal feature selection when having access
to data from heterogeneous environments is nonparametric conditional
independence testing \citep[][Algorithm~\ref{alg:outer} with $p_{S,n}$ being
a conditional independence test for the hypothesis $\rE \indep \rY \mid
\rX^S$]{zhang2012kernel,strobl2019approximate}.

\subsubsection{ROC-based test for binary responses}\label{sec:roc}

\citet{diaz2022identifying} use a nonparametric invariance test based on
comparing the area under the ROC curves obtained from regressing a binary
response $\rY$ on $\rX^S$ and $(\rX^S, \rE)$. In case $S$ is invariant,
$\rE$ contains no further information on $\rY$, hence the AUC should be
the same. Following their approach, we fit two logistic regression models:
(i) for $\rY$ given $\rX^S$ and (ii) for $\rY$ given $\rX^S$ and $\rE$ with
the interaction term between $\rX^S$ and $\rE$. Then, we apply the DeLong
test \citep{delong1988} to test for the equality of the resulting ROCs.
We apply the ROC invariance test only to the settings with binary logistic
regression.

\subsection{Simulation setup}\label{sec:comp}

We outline the simulation setup for our experiments,
 which includes the explicit parameterizations of the
transformation function for all models, details on the data generating process
and simulation scenarios. We compare \tramicp{} with \tram-GCM and \tram-Wald
invariance tests against nonparametric ICP and oracle ICP.

\paragraph{Models and parameterization}
We consider the the binary GLM (``Binary''), a discrete-odds count transformation
model (``Cotram''), and a parametric survival model (``Weibull'').
We also show results for the normal linear model and other
transformation models for ordered and continuous outcomes
(see Tables~\ref{tab:basis} and~\ref{tab:mods}) in Appendix~\ref{app:res:other}.
For our numerical experiments, we parameterize the transformation function $\h$
in terms of basis expansions depending on the type of response $\hY(\bcd)
\coloneqq \basisy(\bcd)^\top \thetavec$, with $\basisy : \calY \to \RR^b$, and
appropriate constraints on $\thetavec \in \Theta \subseteq \overline{\RR}^b$.
Table~\ref{tab:basis} contains a summary of the bases used for continuous,
bounded continuous, count, and binary/ordered responses.

\begin{table}[!ht]
\centering
\caption{Bases expansions used for the baseline transformation in our
experiments, together with the sample space on which they are defined and
the constraints needed to ensure monotonicity. Here,
$\bern{M}(\bcd) = (a_{1,M}, \ldots, a_{M+1,M} )(\bcd)$
denotes a Bernstein basis
of order $M$, \ie
for $\calY \coloneqq [L, U] \subsetneq \RR$ and $j \in \{0, \dots, M\}$,
the $j$th basis function is given by $a_{j,M} : y \mapsto \binom{M}{j} s(y)^j
(1 - s(y))^{M-j}$, where $s : \calY \to [0, 1]$, $s(y) \coloneqq
(y - L) / (U - L)$ maps the bounded continuous response to the unit interval.
Further, for countable $\calY$ with $\lvert\calY\rvert = K$,
$\basisy_{\text{dc},K}(\ry) \coloneqq (\1(\ry = \ry_1), \dots,
\1(\ry = \ry_K))^\top$ denotes the discrete basis.
}\label{tab:basis}
\resizebox{0.95\textwidth}{!}{%
\begin{tabular}{llll}
\toprule
\bf Basis & \bf Sample space & \bf Basis functions & \bf Constraints on $\thetavec$\\
\midrule
Linear & $\calY \subseteq \RR$ & $\basisy: \ry \mapsto (1, \ry)^\top$ &
    $\theta_2 > 0$ \\
Log-linear & $\calY \subseteq \RR_+$ & $\basisy: \ry \mapsto (1, \log\ry)^\top$ &
    $\theta_2 > 0$ \\
Bernstein of order $M$  &
    $\calY \subseteq \RR$ & $\basisy: \ry \mapsto \bern{M}(\ry)$ &
    $\theta_1 \leq \theta_2 \leq \dots \leq \theta_{M + 1}$ \\
Discrete & $\calY = \{\ry_1, \dots, \ry_K\}$ &
    $\basisy: \ry \mapsto \basisy_{\text{dc},K}(\ry)$
    &
    $\theta_1 < \theta_2 < \dots < \theta_{K - 1} < \theta_K
        \coloneqq + \infty$ \\
\bottomrule
\end{tabular}}
\end{table}

\begin{table}[!ht]
\centering
\caption{
Linear shift \tram{s} used in the experiments.
The error distributions are
extended cumulative distribution functions, namely, standard normal ($\Phi$),
standard logistic ($\pSL$), and standard minimum extreme value ($\pMEV$).
The choices of basis for $\hY$ with parameters $\thetavec$ and constraints,
are listed in Table~\ref{tab:basis}. We denote the discrete basis for
an outcome with $K$ classes by $\basisy_{\text{dc},K}$.
}\label{tab:mods}
\resizebox{0.95\textwidth}{!}{%
\begin{tabular}{llll}
\toprule
\bf TRAM
& \bf Basis & \bf Error distribution $\pZ$ & \textbf{Transformation function} $\h$ \\
\midrule
Lm & Linear & $\Phi$ & $(\ry \mid \rx) \mapsto\theta_1 + \theta_2 \ry - \rx^\top\shiftparm$ \\
Binary & Discrete & $\pSL$ & $(\ry \mid \rx) \mapsto\theta - \rx^\top\shiftparm$ \\
BoxCox & Bernstein of order $M$
& $\Phi$ & $(\ry \mid \rx) \mapsto\bern{M}(\ry)^\top\thetavec - \rx^\top\shiftparm$ \\
Polr & Discrete with $K$ classes & $\pSL$ & $(\ry \mid \rx) \mapsto \basisy_{\text{dc},K}(\ry)^\top\thetavec - \rx^\top\shiftparm$ \\
Colr & Bernstein of order $M$ & $\pSL$ & $(\ry \mid \rx) \mapsto\bern{M}(\ry)^\top\thetavec - \rx^\top\shiftparm$ \\
Cotram & Bernstein of order $M$ & $\pSL$ & $(\ry \mid \rx) \mapsto\bern{M}(\ry)^\top\thetavec - \rx^\top\shiftparm$ \\
Coxph & Bernstein of order $M$ & $\pMEV$ & $(\ry \mid \rx) \mapsto\bern{M}(\ry)^\top\thetavec - \rx^\top\shiftparm$ \\
Weibull & Log-linear & $\pMEV$ & $(\ry \mid \rx) \mapsto\theta_1 + \theta_2\log\ry - \rx^\top\shiftparm$ \\
\bottomrule
\end{tabular}}
\end{table}

\paragraph{Data-generating process} In each iteration, the data are drawn from a
random DAG with a pre-specified number of potential ancestors and descendants of
the response. In the DAG, the response given its causal parents is correctly
specified by one of the linear shift \tram{s} in Table~\ref{tab:mods}. All
other structural equations are linear with additive noise. The environments are
encoded by a single Bernoulli-distributed random variable.
Algorithm~\ref{alg:dgp} details the DGP. We generated data for 100 random DAGs
and 20 repetitions per DAG. Sample sizes 100, 300, 1000, and 3000 are
considered. The DAGs were generated with at most 3 ancestors and 2 descendants
of the response (excluding the environment).

\paragraph{Summary measures}
We summarize simulation results in terms of power, measured by Jaccard
similarity of the \tramicp{} output to the true parents,
$\operatorname{Jaccard}(S_n,S_*)$, and the family-wise
error rate (the proportion in which \tramicp{} outputs a non-parent),
$\hat\Prob(S_n \not\subseteq S_*)$.
Jaccard similarity between two sets $A$ and $B$ is defined as
$\operatorname{Jaccard}(A, B) \coloneqq \lvert A \cap B \rvert / \lvert A \cup
B \rvert$ and is 1 iff $A = B$ and 0 if $A$ and $B$ do not overlap.

\paragraph{Software}
Code for reproducing the results in this manuscript is openly available on
GitHub at \url{https://github.com/LucasKook/tramicp.git}. For setting up,
fitting, and sampling from \tram{s}, we rely on packages \pkg{tram}
\citep{pkg:tram} and \pkg{cotram} \citep{pkg:cotram}. We generate and
sample from random DAGs using package \pkg{pcalg} \citep{pkg:pcalg}. The
ROC-based test described in Section~\ref{sec:roc} relies on the \pkg{pROC}
package \citep{pkg:pROC} using \code{roc.test} with \code{method = "delong"}.
To apply \tramicp{}, we use the \pkg{tramicp} package described in
Appendix~\ref{app:pkg}. Conditional independence testing is done via the GCM
test implemented in \pkg{GeneralisedCovarianceMeasure}
\citep{pkg:GeneralisedCovarianceMeasure} using a random forest as implemented
in \pkg{ranger} \citep{pkg:ranger} for both regressions.

\subsubsection{Data generating process}

Algorithm~\ref{alg:dgp} summarizes the data generating process. First, a sample
of a random DAG of $d_1$ potential ancestors of the response is generated and
perturbed with environment effects $\mB^1_E$. Afterwards, the response is
sampled from the given \tram{} with baseline and shift coefficients
$\thetavec$ and $\betavec$, respectively (\cf Table~\ref{tab:mods}). Lastly,
the DAG of $d_2$ potential descendants is generated and likewise perturbed with
environment effects $\mB^2_E$. Based on the adjacency matrix of the resulting
full DAG, we can compute the oracle ICP output (\cf Section~\ref{sec:comparators}).
We generate random DAGs using \code{pcalg::randomDAG()} and sample from them with
\code{pcalg::rmvDAG()}. We simulate from \tram{s} using \code{tram::simulate()}.

\begin{algorithm}[!ht]
\caption{Data generation from random DAGs} \label{alg:dgp}
\begin{algorithmic}[1]
\Require Number of ancestors $d_1$ and descendants $d_2$, a \code{tram} with
parameters $\thetavec \in \RR^K$, $\betavec \in \RR^{d_1}$ with $\pa(\rY) =
\supp \betavec$, sample size $n$, edge probabilities $p_A$ (ancestral graph),
$p_D$ (descendental graph),
coefficient matrices $\mB_E^1$ (environment on ancestors), $\mB_E^2$
(environments on descendants), $\mB_Y$ (response on its children), $\mB_1$ and
ancestors on descendants, routine \code{randomDAG($d,p$)} which outputs a random
DAG with $d$ nodes with connection probability $p$, routine \code{rmvDAG(DAG,
$n$)} which draws an i.i.d.~sample from the distribution implied by \code{DAG},
a routine \code{simulate(tram, newdata)} which draws i.i.d.~observations from
the conditional distribution implied by \code{tram} given \code{newdata}.
\State $E \sim \BD(0.5)$, $\evec \gets n$ i.i.d.\ observations of $E$
\gComment{Simulate environment}
\State $\mD_1 \gets$ \code{rmvDAG(randomDAG(}$d_1, p_1$\code{),$n$)} +
$\mB_E^1\evec$ \gComment{Simulate ancestors}
\State $\yvec \gets$ \code{simulate(tram(}$\thetavec, \betavec$\code{),
newdata = }$\mD_1$\code{)}
\gComment{Simulate response}
\State $\mD_2 \gets$ \code{rmvDAG(randomDAG(}$d_2, p_2$\code{),$n$)} +
$\mB_E^2\evec$ + $\mB_1 \mD_1$ + $\mB_\rY\yvec$
\gComment{Simulate descendants}
\State \Return{$\left(\mD_1 \mid \mD_2 \mid \yvec \mid \evec\right)$}
\gComment{Output simulated data}
\end{algorithmic}
\end{algorithm}

\subsubsection{Simulation scenario}\label{app:sim}

We generate data from all models
listed in Table~\ref{tab:mods}, even the ones that were not discussed in
Section~\ref{sec:results}; we run \tramicp-(\{Wald, GCM\}) and nonparametric ICP
for 100 random DAGs and 20 samples per DAG. We considered sample sizes
100, 300, 1000, 3000, and 10000. The random DAGs contain at most 3 ancestors
and 2 descendants of the response. The probability of two vertices
being connected is fixed at 0.8, which includes also edges between environment
and covariates. Bernstein polynomial bases are fixed to be of order 6. Likewise,
for the Polr model, which features an ordinal response, the number of classes
is fixed at 6. Coefficients for the baseline transformation are fixed at
$\pZ^{-1}(p_0), \pZ^{-1}(p_1), \dots, \pZ^{-1}(p_{M})$, where
$p = (p_0, p_1, \dots, p_{M})^\top$, with $p_0 < p_1 < \dots < p_{M}$,
is a vector of equally spaced
probabilities with $p_0 = 2.5 \times 10^{-4}$ and $p_{M} = 1 - p_1$ fixed.
The coefficients for the log-linear basis are fixed at $(-22, 8)^\top$ to
ensure well-behaved Weibull conditional CDFs reflecting rare events. All
covariates are standardized to mean zero and unit variance. Nodes other than
the response and environment are generated with normal errors and linear
structural equations with additive noise. The coefficients for those structural
equations are drawn from a standard uniform distribution. The non-zero coefficients
for the parents and children of the response are drawn from a normal distribution
with mean zero and variance 0.9. The environments induce only mean shifts in the
covariates by entering the structural equation linearly (Algorithm~\ref{alg:dgp}).
The coefficients are drawn from a normal distribution with mean zero and variance
10.

\subsection{Simulation results}\label{sec:results}

Figure~\ref{fig:simres} summarizes the simulation results for the scenario
described in Section~\ref{sec:comp}. All ICP procedures are level at
the nominal 5\% for all considered models.
Like ICP in linear additive noise models, \tramicp{} is conservative in terms
of FWER. For all models and invariance tests, the power (measured in terms of
Jaccard similarity to the parents) increases with sample size. For all models,
\tramicp-Wald has more power than \tramicp-GCM, which in turn slightly
outperforms nonparametric ICP for the Cotram model. The ROC invariance test is
only applied to binary responses and has comparable power to \tram-GCM and GCM,
but lower power than the Wald test.

\begin{figure}[!t]
\centering
\includegraphics[width=0.8\textwidth]{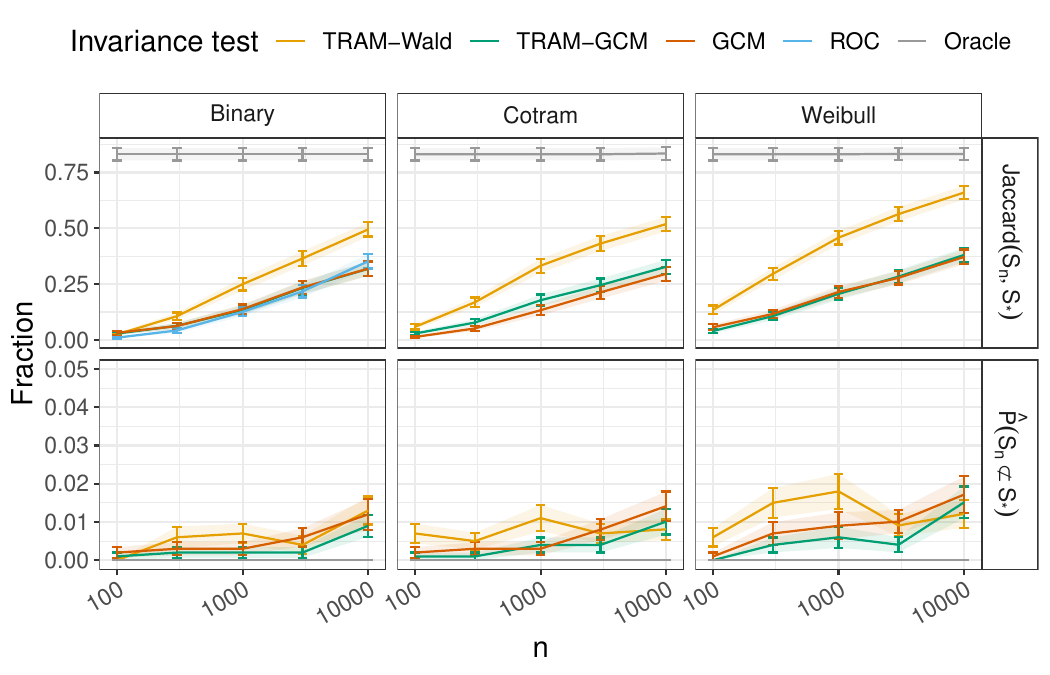}
\vspace{-0.5cm}
\caption{%
Comparison of \tramicp-GCM and \tramicp-Wald with existing methods in a binomial
GLM (``Binary''), a discrete odds count transformation model (``Cotram'') and a
Weibull model (``Weibull'') for different sample sizes. The two proposed
methods are compared against the GCM test (a nonparametric conditional
independence test), the ROC test (for binary logistic regression only) and the
oracle version of ICP, see \eqref{eq:oicp}, in terms of Jaccard similarity to
the set of parents and fraction of outputs containing non-parents (FWER). ICP
is level at the nominal 5\% with all proposed invariance tests, while \tramicp-Wald
is most powerful in this setup (\tram-Wald is not level at the nominal 5\% under
model misspecification, see Appendix~\ref{app:additive}). The ROC
invariance test is on par with \tram-GCM and GCM in terms of power.
The expression $S_n \not \subset S_*$ in the figure
should be understood as $S_n \cap S_*^c \neq \emptyset$.
}
\label{fig:simres}
\end{figure}

\subsection{Additional simulation results}\label{app:res}

\subsubsection{Other models and response types}\label{app:res:other}

Figure~\ref{fig:sim1app} contains
simulation results for the simulation scenario described in
Appendix~\ref{app:sim} for all models from Table~\ref{tab:mods} not
included in the main text, \ie BoxCox, Colr, Coxph, Lm, and Polr. The
conclusions drawn in the main text remain the same: ICP is level at the nominal
5\% for all considered invariance tests. The \tram-Wald invariance
test exerts the highest power (measured in terms of Jaccard-similarity to the
true parents), followed by the \tram-GCM invariance test and the
nonparametric conditional independence test (GCM), which perform on par.

\begin{figure}[!t]
    \centering
    \includegraphics[width=0.85\textwidth]{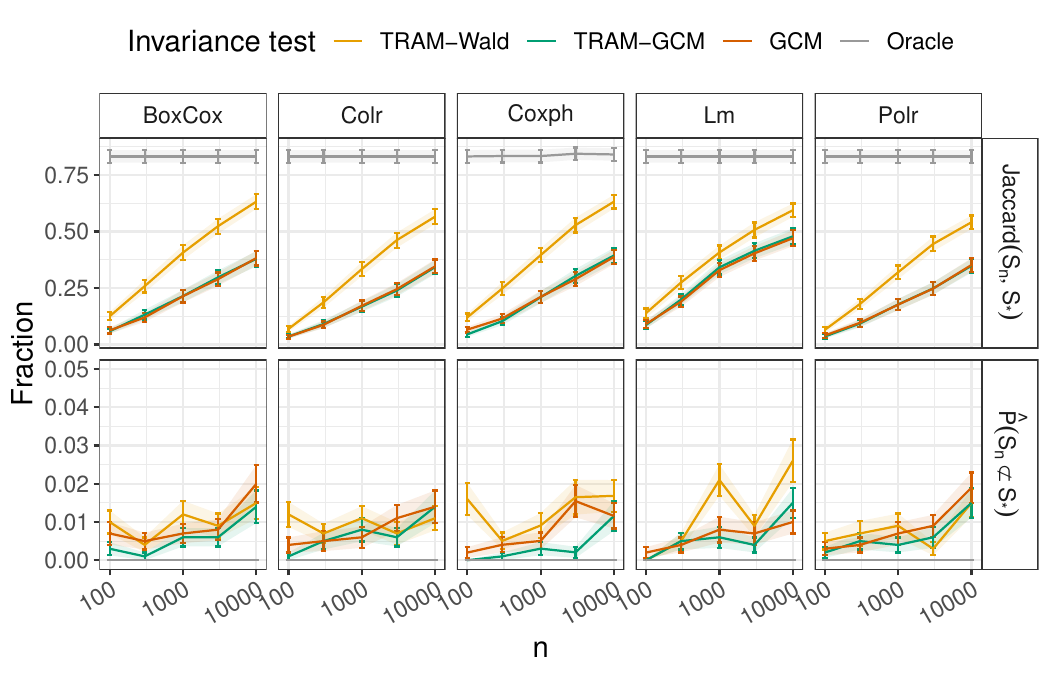}
    \vspace{-0.5cm}
    \caption{%
    Results for the experiment described in Appendix~\ref{app:res:other}. We
    use the simulation scenario described in Appendix~\ref{app:sim} and show
    the models that were not included in Figure~\ref{fig:simres}.
    }
    \label{fig:sim1app}
\end{figure}

\subsubsection{Hidden parents}\label{app:res:hidden}

In the presence of hidden variables, \tramicp{} can still output a subset of
the ancestors, for example, if there is a set of ancestors that is invariant
(see Section~\ref{sec:practical}). Suppose, we generate $Y$ as a \tram{} and
omit one of its observed parents.
If this parent is directly influenced by $E$, then we do not expect
any set of ancestors to be invariant. More precisely, there is no set of observed
covariates, such that $E$ is $d$-separated from $Y$ given this set, which,
assuming faithfulness, implies that there is no invariant set.
When a test has power equal to one to reject non-invariant sets, the output of ICP
is the empty set (and thus still a subset of the ancestors of $Y$, in fact,
even of the parents of $Y$).

Figure~\ref{fig:misppa} shows the results of the simulation
described in Appendix~\ref{app:sim} with two modifications. First, the edge
probabilities for $\rE \to \rX$ and $\rX \to \rY$ are fixed to one, thus $X^1,
X^2, X^3$ are fixed as a parents of $Y$ in all graphs; second, $X^1$ is omitted
(treated as unobserved) when running all ICP algorithms. As stated in
Section~\ref{sec:results}, this situation is adversarial in that there is
no invariant ancestral set as described in Section~\ref{sec:practical}
(the path $E \to X^1 \to Y$ is always unblocked because $X^1$ is
unobserved). There is evidence that, for the considered sample sizes,
\tramicp-GCM and GCM are not level at the nominal 5\%. For smaller sample
sizes, \tramicp-Wald is also anti-conservative.

\begin{figure}[!t]
    \centering
    \includegraphics[width=\textwidth]{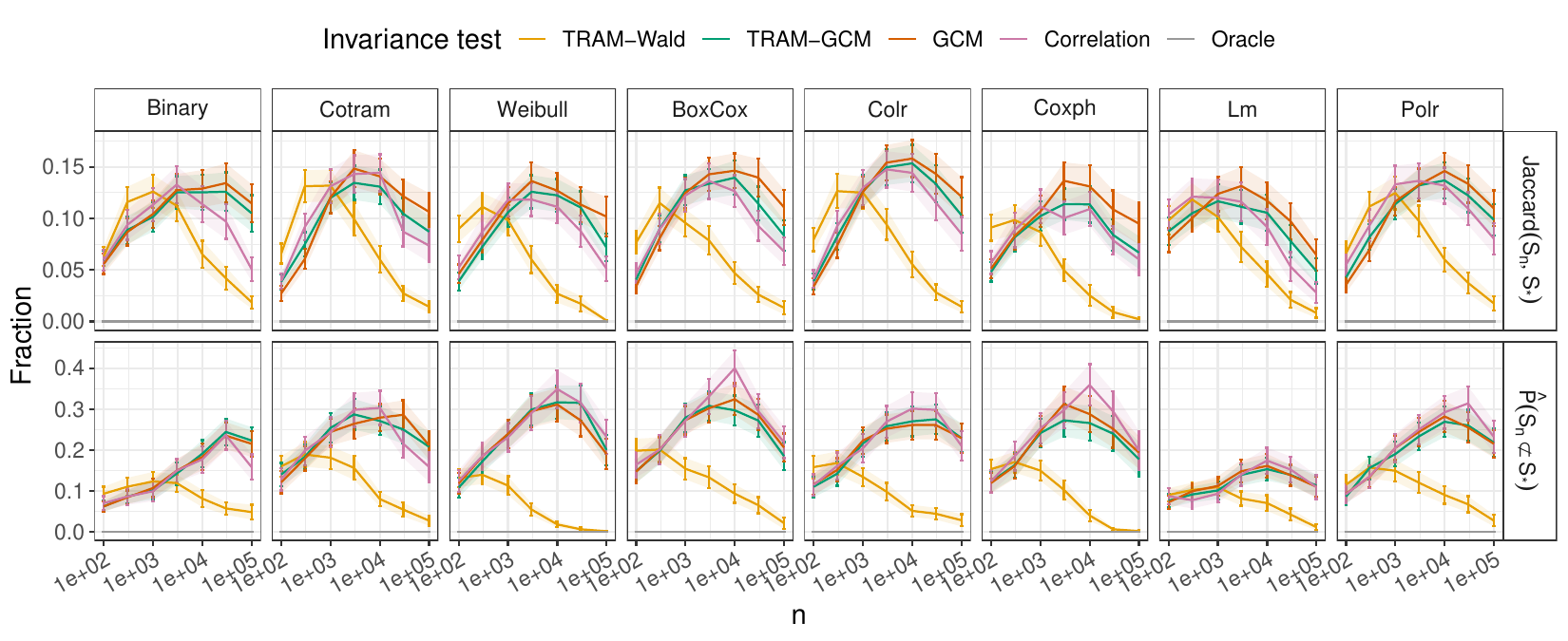}
    \vspace{-1cm}
    \caption{%
    Results for the experiment described in Section~\ref{app:res:hidden}.
    We use the same
    simulation scenario
    as in Appendix~\ref{app:sim}
    with $X^1$ as a fixed, but
    treated-as-hidden parent.
    Since $E \to X^1 \to Y$,
    there is no invariant set.
    There is evidence that ICP (with all invariance tests considered here) is
    not level at the nominal 5\% for small sample sizes and all models.
    }
    \label{fig:misppa}
\end{figure}

However, for large sample sizes the empirical probability of obtaining a set
that is not a subset of the parents of $Y$ decreases for all models. In addition,
Figure~\ref{fig:misppa:larger} shows increased power for all invariance test
when the coefficient for $X^1 \to Y$ is multiplied by $2.5$ (for the same DAGs
used in Figure~\ref{fig:misppa}).

\begin{figure}[!t]
    \centering
    \includegraphics[width=\textwidth]{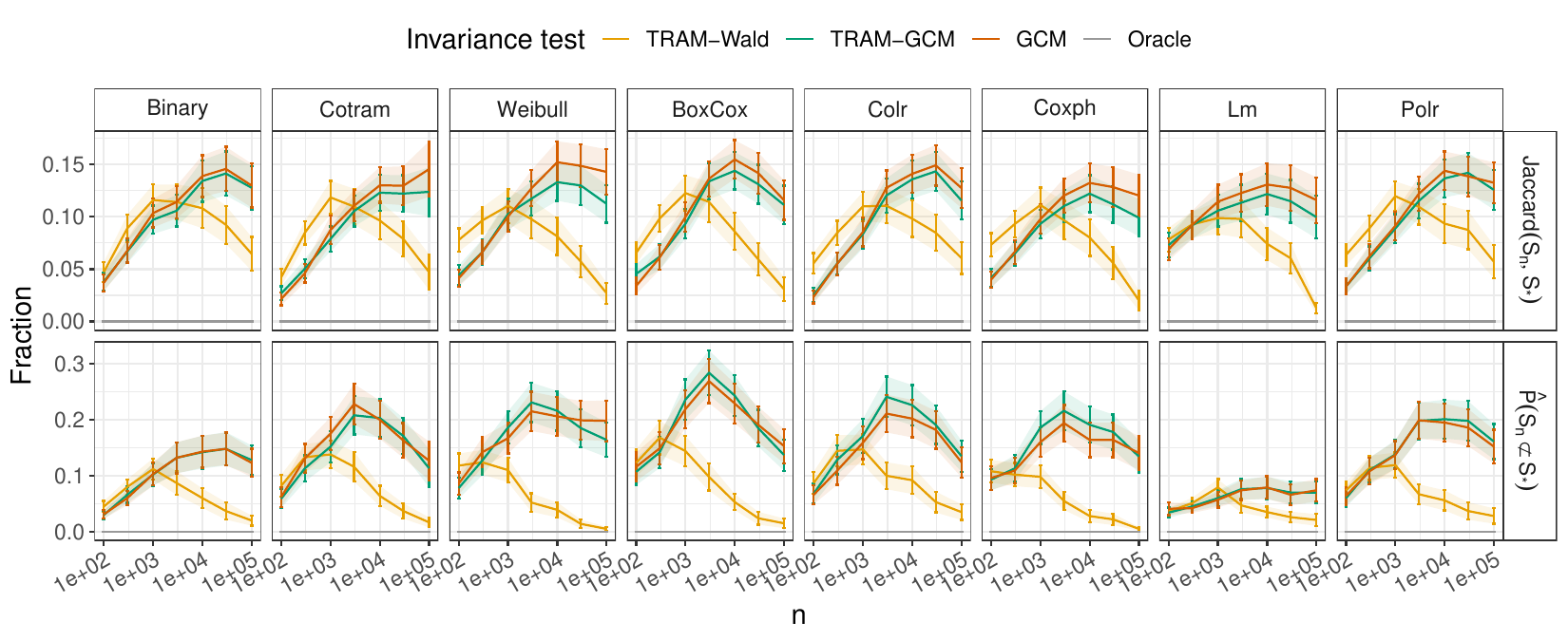}
    \vspace{-1cm}
    \caption{%
    Results for the experiment described in Section~\ref{app:res:hidden}.
    We use the same simulation scenario as in Appendix~\ref{app:sim}
    with $X^1$ as a fixed, but treated-as-hidden parent and with coefficients
    for $X^1 \to Y$ multiplied by $2.5$. The invariance tests show larger
    power to reject non-invariant sets. In order for the \tram{s} to be
    well-behaved, the data $\mD_1$ in Algorithm~\ref{alg:dgp} are standardized,
    so that each component has empirical variance $1/25$.
    }
    \label{fig:misppa:larger}
\end{figure}

\subsubsection{Misspecified error distribution}\label{app:res:error}

Choosing an error distribution $\pZ$ different from the one generating the
data is another way to misspecify the model.
Figure~\ref{fig:mispfz} contains
results for the simulation scenario in Appendix~\ref{app:sim} when the error
distribution when fitting the \tram{} is different from the one generating
the data. For the Cox and Colr model, there is mild evidence that \tramicp-Wald
is not level at the nominal 5\%. All \tramicp{} algorithms show slightly reduced
power under a misspecified error distribution. However, \tramicp-GCM empirically
shows some robustness and still performs at least on par with nonparametric ICP
(with the Colr model being an exception).

\begin{figure}[t!]
    \centering
    \includegraphics[width=\textwidth]{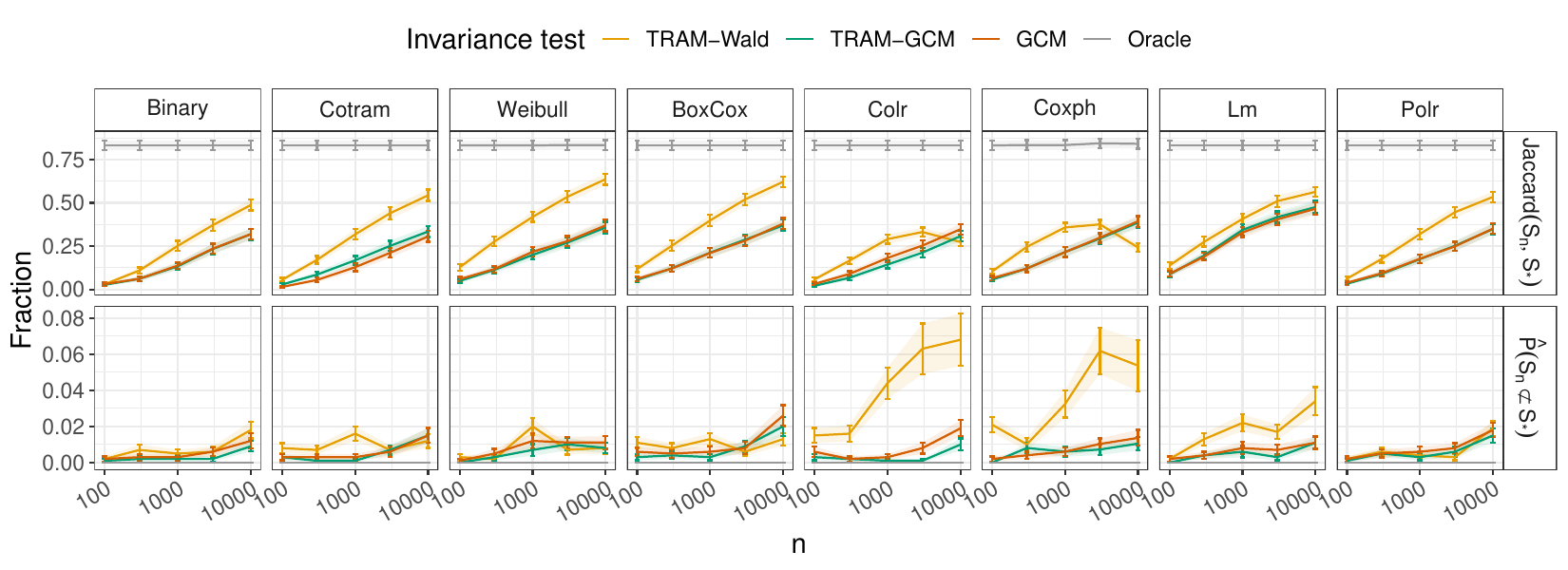}
    \vspace{-1cm}
    \caption{%
    Results for the experiment described in Appendix~\ref{app:res:error}.
    We use the same simulation scenario as in Appendix~\ref{app:sim} but with
    misspecified $\pZ$. The data is generated
    according to the models in Table~\ref{tab:mods}, but the error
    distribution used when running ICP is different from the one generating
    the data, as follows: Binary: Standard normal instead
    of standard logistic. BoxCox: Standard logistic instead of standard
    normal. Colr: Standard minimum extreme value instead of standard logistic.
    Cotram: Standard normal instead of standard logistic. Coxph:
    Standard logistic instead of standard minimum extreme value. Lm: Standard
    logistic instead of standard normal. Polr: Standard normal instead of
    standard logistic. Weibull: Standard log-logistic instead of standard
    minimum extreme value.
    }
    \label{fig:mispfz}
\end{figure}

\subsection{Faithfulness in linear shift transformation models}\label{app:faith}

Studying faithfulness violations theoretically in \tram{s} is again
hindered by the lack of collapsibility and marginalizability (discussed
in Appendix~\ref{app:collaps}).
We now give an example illustrating that there are structural causal
\tram{s} that induce distributions that are non-faithful \wrt the induced graph.
Consider the following structural causal \tram{},
\begin{align}
    E &\coloneqq N_E\\
    X^1 &\coloneqq E + N_1\\
    X^2 &\coloneqq E + N_2\\
    Y &\coloneqq g(Z - 0.5 X^1 + \beta X^2),
\end{align}
where $g \coloneqq F_{\chi^2(3)}^{-1}\circ\Phi$ denotes the inverse
baseline transformation and $N_E \sim \BD(0.5)$, $N_1 \sim \ND(0, 1)$,
$N_2 \sim \ND(0, 1)$, and $Z \sim \ND(0,1)$ are jointly independent
noise terms. The population version of \tramicp{} returns $\{1, 2\}$ as
an invariant set since $E \indep Y \mid X^1, X^2$. We let $\beta$ vary from
$0.2$ to $0.8$; for $\beta = 0.5$, we have the following cancellation
\begin{align}
    Y = g(Z - 0.5 (X^1 - X^2)) = g(Z - 0.5 (N_1 - N_2)).
\end{align}
For $\beta = 0.5$, due to the cancellation, $E \indep Y$, \ie the empty set is
invariant, too, and the output of a population version of \tramicp{} equals the
empty set. For each of the above values of $\beta$, we simulate a single sample
from the implied joint distribution with sample size $n = 10^4$ and
apply \code{BoxCoxICP} with the \tram-GCM invariance test.
Table~\ref{tab:faith} shows set $p$-values and the output of \tramicp{}.
As expected, for $\beta = 0.5$ (which implies $E \indep Y$), \tramicp{} indeed
outputs the empty set.

\begin{table}[!t]
\centering
\begin{tabular}{lrrrrr}
\toprule
$\beta$ & $\emptyset$ & $X^1$ & $X^2$ & $X^1 + X^2$ & Output\\
\midrule
0.2 & 0.000 & 0 & 0 & 0.280 & $\{1,2\}$\\
0.3 & 0.000 & 0 & 0 & 0.271 & $\{1,2\}$\\
0.4 & 0.000 & 0 & 0 & 0.266 & $\{1,2\}$\\
0.5 & 0.949 & 0 & 0 & 0.273 & $\emptyset$\\
0.6 & 0.000 & 0 & 0 & 0.281 & $\{1,2\}$\\
0.7 & 0.000 & 0 & 0 & 0.291 & $\{1,2\}$\\
0.8 & 0.000 & 0 & 0 & 0.302 & $\{1,2\}$\\
\bottomrule
\end{tabular}
\caption{%
Results for the experiment described in Appendix~\ref{app:faith}.
Set-specific $p$-values and output of \tramicp{} are shown for different
choices
for the coefficient $\beta$ of $X^2$ to illustrate faithfulness violations
in linear shift \tram{s}. The $p$-values stem from a single run
of \tramicp{} on data generated according to the structural causal
\tram{} in the main text with a sample size of $n = 10^4$.
The setting $\beta=0.5$ implies $E \indep Y$ and thus corresponds
to a faithfulness violation. The test indeed does not
reject this hypothesis and consequently, \tramicp{} outputs the empty set.
The seed for data generation is fixed to the same value for each $\beta$, such
that the results are directly comparable across rows.
}
\label{tab:faith}
\end{table}

This example illustrates how faithfulness violations can lead to uninformative
results when using \tramicp. However, such cancellations seem to be avoided
when sampling the coefficients from a continuous distribution, as done in our
experiments in Section~\ref{sec:results}.

\subsection{Wider applicability of the TRAM-GCM invariance test}\label{app:additive}

We have theoretical guarantees for \tramicp-Wald only under the class of linear
shift \tram{s} and, indeed, the presence of a nonlinear covariate
effect can lead to \tramicp-Wald being anti-conservative. The
\tram-GCM invariance test can still be used for the more general
shift \tram{s} (or when linear shift \tram{s} are estimated via penalized
maximum likelihood.

The \tram-Wald invariance test requires regularity conditions and is
restricted to linear shift \tram{s} (see Proposition~\ref{prop:wald}),
whereas the \tram-GCM invariance test
is not limited to linear shift \tram{s}, for example (Theorem~\ref{thm:icp}).
Here, we show empirically that if the relationship between the response
and covariates is governed by a nonlinear function or penalized maximum
likelihood estimation is used, the \tram-Wald invariance test and a naive
Pearson correlation test between score residuals and environments are not
level $\alpha$, whereas \tram-GCM remains level $\alpha$.
For this, we consider two examples, shift \tram{s}
(Example~\ref{ex:additive}) and linear shift \tram{s}
estimated via penalized maximum likelihood (Example~\ref{ex:penalized}).

\begin{example}[\tram-Wald may be anti-conservative]\label{ex:additive}
Consider a structural causal shift \tram{} with the following assignments:
\begin{align}
    E &\sim \BD(0.5)\\
    X^1 &\coloneqq E + N_1\\
    X^2 &\coloneqq N_2\\
    Y &\coloneqq Z + f_1(X^1) + 0.5 X^2\\
    X^3 &\coloneqq Y - E + N_3,
\end{align}
where $N_1, N_2, Z, N_3$ are i.i.d. standard normally distributed and
$f_1 : x \mapsto 0.5 x + \sin(\pi x) \exp(-x^2/2)$.
A scatter plot for $Y$ and $X^1$ is depicted in Figure~\ref{fig:additive}A.
In Figure~\ref{fig:additive}B, we depict the empirical cumulative distribution
functions (ECDFs) of the $p$-values obtained
from the \tram-GCM invariance test, \tram-Wald invariance test, and a naive
correlation test between the score residuals and environments for $S = \{1\}$.
The $p$-values of the \tram-GCM invariance test are close to uniform
(a binomial test for the null hypothesis
$H_0: \text{rej.\ probability at level } 0.05 \leq 0.05$ yields a $p$-value of $0.25$).
There is, however, strong evidence
that the \tram-Wald invariance is not level at the nominal 5\%
(the binomial test described above yields a $p$-value less than $0.0001$).
and the correlation test seems overly conservative (the binomial test described
above yields a $p$-value of $0.99$).
\begin{figure}[!t]
    \centering
    \includegraphics[width=0.9\textwidth]{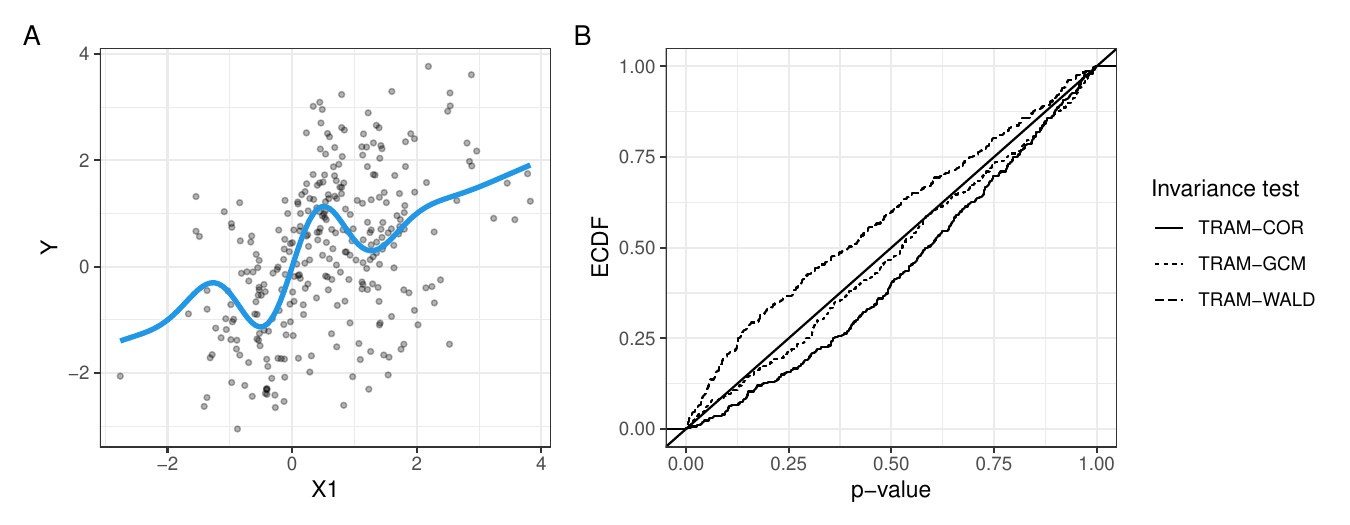}
    \caption{%
    Results for the experiment described Example~\ref{ex:additive} in in
    Appendix~\ref{app:additive} on the wider applicability of the \tram-GCM
    invariance test in shift \tram{s}.
    A: Relationship between $Y$ and $X^1$ with a sample size of $n = 300$.
    B: ECDFs of $p$-values of the \tram-GCM invariance test, a Pearson
    correlation test (\tram-COR), and the \tram-Wald invariance test for
    $n = 300$ and 300 repetitions.
    }
    \label{fig:additive}
\end{figure}
\end{example}

\begin{example}[The Pearson correlation invariance test may be
anti-conservative]\label{ex:penalized}\sloppy
Consider a structural causal linear shift \tram{} over $(Y, X^1, X^2, \dots,
X^{12}, E)$ with the same assignments as in Example~\ref{ex:additive} using
$f_1 : x \mapsto 0.5 x$ and additional covariates $X^4, \dots, X^{12}$
drawn i.i.d.\ from a standard normal distribution. We fit penalized
linear shift \tram{s} with linear transformation function and an $\ell_2$-penalty
using $\lambda = 10$. In Figure~\ref{fig:penalized}, we again display the ECDF
of $p$-values obtained from the \tram-GCM invariance test, the \tram-Wald
invariance test and a naive correlation test between score residuals and
environments for the set $\{1, 2, 4, 5, \dots, 12\}$. There
is no evidence that the \tram-GCM invariance test is not level
(the binomial test described above yields a $p$-value equal to $0.78$).
However, there is strong evidence that the \tram-Wald
invariance test
and correlation test
are not level at the nominal 5\% (in both cases the binomial test described
above yields a $p$-value less than $0.0001$).
\begin{figure}[H]
    \centering
    \includegraphics[width=0.6\textwidth]{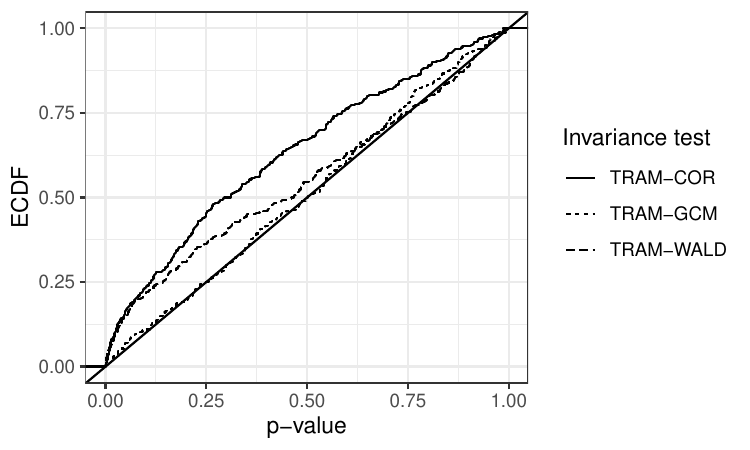}
    \caption{%
    Results for the experiment described in Example~\ref{ex:penalized}
    in Appendix~\ref{app:additive} on the wider applicability of the \tram-GCM
    invariance test in penalized \tram{s}.
    We show the ECDF of $p$-values of \tram-GCM invariance test versus the
    \tram-Wald invariance test and the Pearson correlation test (\tram-COR) for
    $n = 300$ and 300 repetitions in penalized \tram{s}.
    }
    \label{fig:penalized}
\end{figure}
\end{example}

\subsection[Nonparametric extensions of TRAMICP]{%
Nonparametric extensions of TRAMICP}\label{app:nonp}

Although our theoretical results (Theorem~\ref{thm:icp}) rely on the class of
shift \tram{s}, the \tram-GCM test can be
applied even when
using score residuals obtained via a nonparametric estimator
$\hat{F}_{\rY \mid \rX^S = \rx^S}$
for $F_{\rY \mid \rX^S = \rx^S}$, $S \subseteq [d]$, $\rx^S \in \calX^S$,
such as quantile \citep{meinshausen2006quantile} or survival
\citep{ishwaran2008survforest} random forests.
For a given $\pZ \in \calZ$ and $S \subseteq [d]$, we compute the \tram-GCM
invariance test with the score residuals using, for all $\rx^S \in \calX^S$,
$\hat\h(\bcd \mid \rx^S) \coloneqq \pZ^{-1}(\hat{F}_{Y \mid \rX^S = \rx^S}(\bcd))$.
This nonparametric extension
can help to
circumvent
the issue of misspecification and non-closure under marginalization of
\tram{s}. We demonstrate \tramicp{} based on survival forests in a
simulation with the following setup.

Consider a structural causal model with the following assignments:
\begin{align}
    X^1 &\coloneqq N_{X^1} \\
    E &\coloneqq 2 X^1 + N_E \\
    X^2 &\coloneqq 2 E + N_{X^2} \\
    Y^* &\coloneqq (1 + \expit(1.5 + 3 X^1 + 3 X^2 \1(X^2 < 0.2) +
    2 \1(X^2 > 0.2) \1(X^1 > 0) + N_Y))^3 \\
    Y &\coloneqq \min(Y^*, 4) \\
    X^3 &\coloneqq -6E + 2Y^* + N_{X^3}
\end{align}
where $N_{X^1}, N_{X^2}, N_{X^3}$ are jointly independent and
logistically distributed with scale 0.1

We fix the sample size to be $n = 80$ and repeat the simulation 300 times.
We apply the misspecified \code{coxphICP}, and nonparametric ICP using
random forests (while ignoring the censoring indicator), and the proposed
nonparametric extension of \tramicp{} using survival forests and computing
score residuals with $\pZ \coloneqq \expit$.

\begin{figure}[!ht]
\centering
\includegraphics[width=0.85\textwidth]{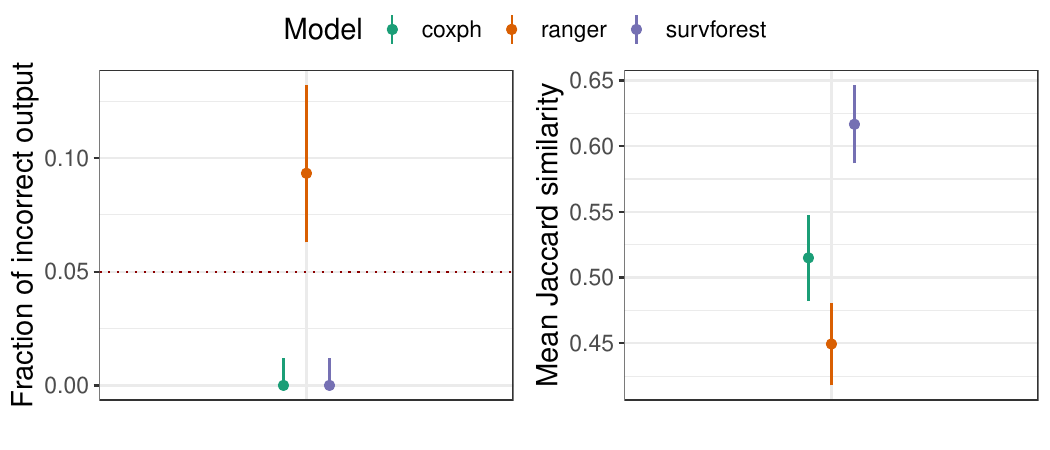}
\vspace{-0.5cm}
\caption{%
Fraction of incorrect output (left) and mean Jaccard similarity (between the
output and the true set of causal parents; right) for different variants of
ICP for the setup described in Appendix~\ref{app:nonp}.
Results are shown for the (misspecified) Cox model, nonparametric ICP based on
ranger (which ignores the censoring present in the data) and the proposed
nonparametric extension of \tramicp{} using survival forests.
}\label{fig:nonp}
\end{figure}

\paragraph{Results}
The fraction of incorrect output and power (in terms of Jaccard similarity
between the output and the true set of parents) is shown in
Figure~\ref{fig:nonp}. While there is evidence that
nonparametric ICP via GCM is not level, both \tramicp{} using the Cox model and
survival forests seem to be level. In terms of power, the proposed nonparametric
extension of \tramicp{} based on survival random forests outperforms both
nonparametric ICP and \tramicp{} based on the Cox model. Despite misspecification,
the Cox model is surprisingly robust, but has, at least in this scenario, less
power than the proposed extension.

If the risk of misspecification is deemed high, the proposed nonparametric
extension can be used to mitigate this risk and has been empirically shown
to be level even in situations where nonparametric ICP is not (potentially
due to the presence of censoring). We leave theoretical investigations of
the proposed extension for future work. The \proglang{R}~package
\pkg{tramicp} implements the extension in \code{qrfICP()} and
\code{survforestICP()} for quantile and survival random forests,
respectively; we have added a note to indicate that it does not come with
theoretical guarantees yet.

\section{Causal drivers of survival in criticially ill adults}\label{app:casestudy}

\subsection{Evidence of age and cancer being direct causes of survival}\label{app:casestudy:fig}

Figure~\ref{fig:casestudy} shows the set specific $p$-values of \tram-GCM
and \tram-Wald. All non-rejected sets contain \code{ca} and \code{age}.
Since all $p$-values fall below the identity, \tram-GCM is more conservative.

\begin{figure}[!ht]
\centering
\includegraphics[width=0.85\textwidth]{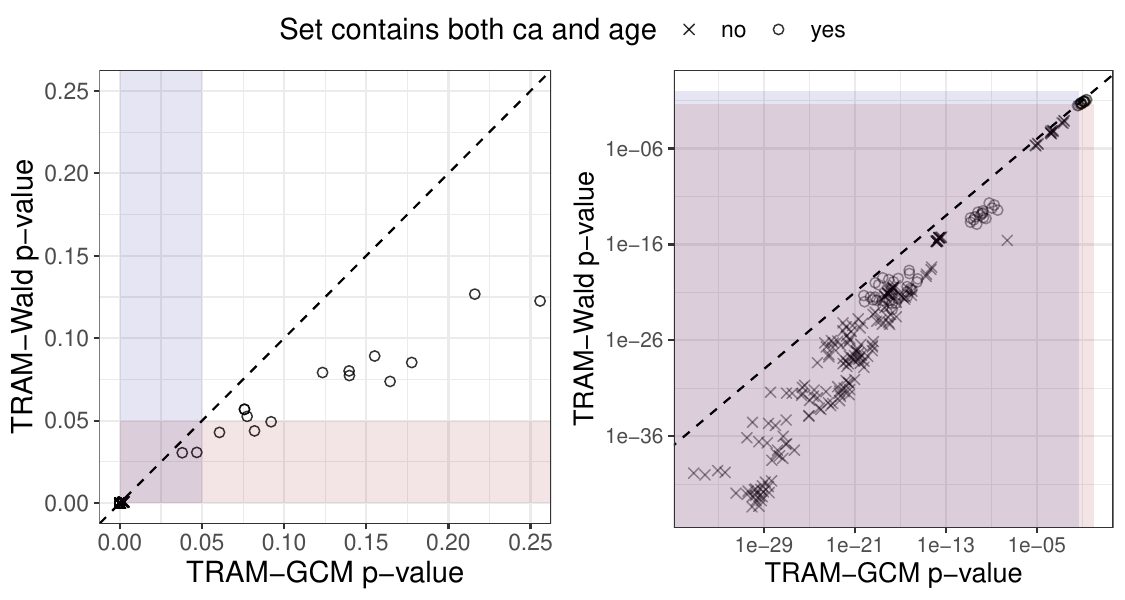}
\caption{%
    Set-specific $p$-values of the \tram-GCM and \tram-Wald invariance test in
    the \mbox{SUPPORT2} case study described in Section~\ref{sec:casestudy}.
    Sets containing both \code{ca} and \code{age} are depicted as circles,
    whereas sets containing neither are depicted as crosses. The blue and red
    shaded regions mark the rejection region for \tram-GCM and \tram-Wald,
    respectively, and the dashed line is the identity function. Left: $p$-values
    on a linear scale. All invariant sets contain cancer and age, which is
    therefore the output of both \tramicp-GCM and \tramicp-Wald (see
    Table~\ref{tab:casestudy}). Right: $p$-values on the $\log_{10}$-scale.
    \tram-GCM is more conservative than \tram-Wald, as all $p$-values fall below
    the identity (dashed line).
}\label{fig:casestudy}
\end{figure}
%

\subsection{Multiple environments}\label{app:casestudy:menv}

\tramicp{} allows several environments to be specified. When using both
\code{race} and \code{num.co} as environments in the SUPPORT2 dataset,
\tramicp-GCM outputs \code{ca}, \code{age}, \code{diabetes} as plausible
causal predictors. \tramicp-Wald additionally outputs \code{dementia} and
\code{sex}. Choosing more variables as  environments may yield more
heterogeneity but may at the same time decrease the power of statistical tests.
The predictor $p$-values are given in the Table~\ref{tab:casestudy:app}.

\begin{table}[!ht]
\centering
\caption{%
\tramicp{} applied to the SUPPORT2 dataset with \code{num.co} and \code{race}
as the environments. Predictor-specific $p$-values (see
Section~\ref{sec:practical}) are reported for the \tram-GCM and
\tram-Wald invariant test.
$p$-values in bold are significant at the 5\% level; in each row,
the set of predictors with bold numbers corresponds
to the output of \tramicp{}.
}\label{tab:casestudy:app}
\resizebox{\textwidth}{!}{%
\begin{tabular}{@{}lrrrrrrrrc@{}}
\toprule
\textbf{Invariance test}
& \multicolumn{8}{c}{\textbf{Predictor-specific $p$-values}}&
\multicolumn{1}{l}{\textbf{Environment}} \\ \midrule
& \multicolumn{1}{l}{\code{scoma}} & \multicolumn{1}{l}{\code{dzgroup}} & \multicolumn{1}{l}{\code{ca}} & \multicolumn{1}{l}{\code{age}} & \multicolumn{1}{l}{\code{diabetes}} & \multicolumn{1}{l}{\code{dementia}} & \multicolumn{1}{l}{\code{sex}} & \multicolumn{1}{l}{\code{race}} & \multicolumn{1}{l}{} \\ \cmidrule(lr){2-9}
\multicolumn{10}{l}{\textit{Multiple environments}}\\
\tram-GCM  & 0.089 & 0.089 & \textbf{0.000} & \textbf{0.003} & \textbf{0.040} & 0.056 & 0.085 & - & \multirow{2}{*}{\code{num.co + race}}           \\
\tram-Wald & 0.050 & 0.050 & \textbf{0.000} & \textbf{0.002} & \textbf{0.031} & \textbf{0.029} & \textbf{0.040} & - &\\
\bottomrule
\end{tabular}
}
\end{table}

\subsection{Sensitivity analysis: Informative censoring}\label{app:casestudy:cens}

When applying \tramicp{} in a survival context, the assumption of
uninformative censoring plays an important role. In their original
analysis of the SUPPORT dataset, \citet{knaus1995support} have assumed
that the censoring is uninformative. As a sensitivity analysis, we can
introduce (possibly additional) informative censoring by treating the
observed event times of a fraction of all patients who experienced the
event as right-censored.
For 10\% (20\%, ..., 90\%) of randomly chosen patients with
exact event times, we repeat the analysis ten times. For low fractions of
additionally informatively censored observations (10--20\%), the output of
\tramicp-Wald remains
stable (\code{ca}, \code{age}). For larger fractions (30--80\%), \code{age}
is contained in the output less often. For the largest considered fraction
(90\%), \tramicp-Wald outputs mostly the empty set (see Table~\ref{tab:cens}).

The sampling is repeated 10 times per fraction. Table~\ref{tab:cens} shows
and how often \tramicp-Wald outputs the empty set, \code{ca} or \code{ca+age}
over the ten repetitions. \tramicp-Wald is relatively robust for fractions
up to 40\%. For larger amounts of informative censoring, \code{age} is
contained in the output less often. At fraction 0.9, the output is empty for
most of the repetitions.

\begin{table}[!ht]
\centering
\caption{%
\tramicp{} applied to the \mbox{SUPPORT2} dataset with different fractions
of additionally informatively censored survival times, as described in
Appendix~\ref{app:casestudy}. The additionally censored patients are sampled
from all patients who experienced the event. The experiment is repeated ten
times and the table shows how often the output of \tramicp-Wald is the empty
set, \code{ca}, or \code{ca+age} out of the ten repetitions.
}\label{tab:cens}
\renewcommand{\arraystretch}{0.6}
\begin{tabular}{lrrr}
\toprule
\textbf{Fraction} &
\multicolumn{3}{c}{\textbf{Output (out of 10)}}\\
& \code{Empty} & \code{ca} & \code{ca+age}  \\
\cmidrule(lr){2-4}
0.1 &
0 & 0 & 10\\
0.2 &
0 & 0 & 9 \\
0.3 &
1 & 2 & 4 \\
0.4 &
2 & 2 & 6 \\
0.5 &
1 & 4 & 4 \\
0.6 &
3 & 5 & 2 \\
0.7 &
0 & 9 & 1 \\
0.8 &
3 & 3 & 3 \\
0.9 &
7 & 2 & 1 \\
\bottomrule
\end{tabular}
\end{table}

\section{\proglang{R} package \pkg{tramicp}}\label{app:pkg}

With \pkg{tramicp}, we provide a user-friendly implementation for applying
\tramicp{}, which we briefly outline in this section. For every model
implementation listed in Table~\ref{tab:mods} in Section~\ref{sec:comp}, there
is a corresponding alias in \pkg{tramicp} which appends \code{ICP} to the model
name (\eg \code{glmICP} in Example~\ref{ex:intro}).

As an example, in the below code snippet, we apply \tramicp-GCM to data generated
from a structural causal \tram{} with DAG $\calG$ (shown below) and a ``Cotram''
\citep[cf.\ count \tram,][]{siegfried2020count}
model for the count-valued
response with $\pZ = \pSL$ (see Table~\ref{tab:mods} and
Example~\ref{ex:count}).
\begin{verbatim}
R> cotramICP(Y ~ X1 + X2 + X3 + X4, data = dat, env = ~ E,  type = "residual",
+    test = "gcm.test", verbose = FALSE)
\end{verbatim}
The argument \code{type} specifies the type of invariance considered, \ie
\code{"wald"} for the Wald test or \code{"residual"} for any residual-based
test (the default is \code{test = "gcm.test"}).

The corresponding output is shown below.
\tramicp{} correctly returns $\{1, 2\}$ as the causal parent
of the response. The reported $p$-value for a predictor of interest is computed
as the maximum $p$-value over all tested sets not containing the predictor of
interest (in case all sets are rejected, the $p$-value is set to 1, see
Section~\ref{sec:practical}). An illustration of the \tram-GCM invariance
test can be found in Figure~\ref{fig:intro}.
\\[6pt]
\begin{minipage}{0.6\textwidth}
\begin{verbatim}
Model-based Invariant Causal Prediction
Discrete Odds Count Transformation Model

Call: cotramICP(formula = Y ~ X1 + X2 + X3 + X4,
    data = df, env = ~E, verbose = FALSE,
    type = "residual", test = "gcm.test")

 Invariance test: gcm.test

 Predictor p-values:
   X1    X2    X3    X4
0.001 0.000 0.699 0.699

 Set of plausible causal predictors: X1 X2
\end{verbatim}
\end{minipage}
\begin{minipage}{0.39\textwidth}
\includegraphics[width=0.95\textwidth]{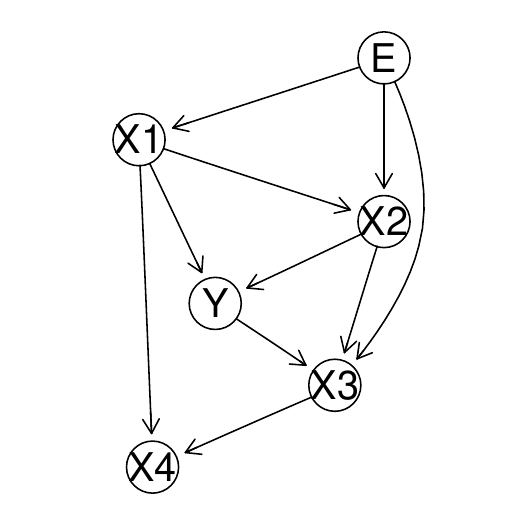}
\end{minipage}\\[6pt]

If prior knowledge of the form $S_m \subsetneq S_*$ is
available, only super-sets of $S_m$ need to be tested. Thus, $S_m$ can be
interpreted as a mandatory part of the conditioning set. In \pkg{tramicp},
such mandatory predictors can be supplied to the \code{mandatory} argument
as a formula. In our example above, we could, for instance, specify
\code{mandatory = $\sim$ X1} in the call to \code{cotramICP()}. In
Section~\ref{sec:casestudy}, we illustrate how such prior
knowledge can be used to reduce computation time by testing fewer sets.
If non-parents are falsely included as mandatory, the output of \tramicp{}
may be misleading. Also, the robustness guarantees discussed in
Section~\ref{sec:practical} `Unmeasured confounding' can break down.

The \pkg{tramicp} package implements ICP for the models and invariance tests
listed in Tables~\ref{tab:pkg} and~\ref{tab:tests}, respectively.
In its general form, \tramicp{} is implemented in the \code{dicp()} function.
The function takes the arguments \code{formula}, \code{data}, \code{env}
(a formula for the environment variables) and \code{modFUN} (the function for
fitting the model).
\code{modFUN} can be one of the implementations listed in
Table~\ref{tab:pkg}, or any other model which possesses an
\proglang{S}3~method corresponding to \code{stats::residuals()} (for
\code{type = "residual"} invariance) or \code{multcomp::glht()} (for
\code{type = "wald"} invariance) (below we give an example of how to
use shift \tram{s} implemented in \pkg{tramME}). Package
\pkg{tramicp} implements aliases for
\code{dicp(..., modFUN = <implementation>)}, which are listed in
Table~\ref{tab:pkg} and use some convenient defaults for arguments supplied
to \code{modFUN}. For instance, \code{MASS::polr()} has no corresponding
\code{residuals()} implementation, which has been added with \pkg{tramicp}.
Implemented tests can be supplied as character strings
(Table~\ref{tab:tests}). For \code{type = "residual"}, custom tests can be
supplied as \code{function(r, e, controls) \{...\}}). Importantly, these
tests should output a list containing an entry \code{"p.value"}. Additional
arguments can be supplied via the \code{controls} argument, which can be
obtained and altered via \code{dicp\_controls()}. Some tests (for instance,
\code{coin::independence\_test()}) can handle multivariable environments,
\ie formulae like \code{env = $\sim$ E1 + E2}.

\paragraph{Kernel-based independence test}
For continuous responses, $(\pZ,\trafosub)$-invariance can be tested via
independence tests, \ie by testing the null hypothesis $H_0(S) :
\rE \indep R(\h; \rY, \rX^S)$ for a given $S\subseteq[d]$ and $\h \in
\trafosubS$. In \pkg{tramicp}, we implement a test based on
the Hilbert-Schmidt Independence Criterion \citep[HSIC,][]{gretton2007kernel}.
\citet{gretton2007kernel} introduce HSIC as a kernel-based independence
measure, which is zero if and only if the two arguments are independent
(when using characteristic kernels). Here, we use dHSIC \citep{pfister2018hsic}.
By default, a Gaussian kernel for $R$ and a discrete (for discrete
environments) or Gaussian (for continuous environments) kernel for
$\rE$ is used \citep[for details see][]{pfister2018hsic}.

\begin{table}[!ht]
    \centering
    \caption{%
    Model classes directly available in \pkg{tramicp}. Cusom model
    functions can be supplied in \code{dicp()} and have to take arguments
    \code{formula} and \code{data}. For details,
    the reader is referred to the package documentation.
    }
    \label{tab:pkg}
    \begin{tabular}{lrr}
    \toprule
         \bf Model & \bf Function & \bf Implementation \\
         \midrule
         Box--Cox-type & \code{BoxcoxICP()} & \code{tram::Boxcox()} \\
         Continuous-outcome logistic regression & \code{ColrICP()} & \code{tram::Colr()} \\
         Count transformation model & \code{cotramICP()} & \code{cotram::cotram()} \\
         Cox proportional hazards & \code{CoxphICP()} & \code{tram::Coxph()} \\
         Generalized linear regression & \code{glmICP()} & \code{stats::glm()} \\
         Lehmann-type regression & \code{LehmannICP()} & \code{tram::Lehmann()} \\
         Normal linear regression & \code{LmICP()} & \code{tram::Lm()} \\
         Normal linear regression & \code{lmICP()} & \code{stats::lm()} \\
         Proportional odds logistic regression & \code{PolrICP()} & \code{tram::Polr()} \\
         Proportional odds logistic regression & \code{polrICP()} & \code{MASS::polr()} \\
         Weibull regression & \code{SurvregICP()} & \code{tram::Survreg()} \\
         Weibull regression & \code{survregICP()} & \code{survival::survreg()} \\
         \bottomrule
    \end{tabular}
\end{table}

\begin{table}[!ht]
    \centering
    \caption{%
    Invariance tests implemented in \pkg{tramicp}. Except for the
    \tram-Wald test (which requires \code{type = "wald"}), all tests have
    to be used with \code{type = "residual"} (the default). Custom
    residual-based tests can be supplied as functions of the form
    \code{function(r, e, controls) \{...\}}. Custom functions for computing
    the residuals can be used by passing the function name to
    \code{dicp\_controls()}.
    }
    \label{tab:tests}
    \begin{tabular}{lrr}
    \toprule
         \bf Test & \bf Name & \bf Implementation \\
         \midrule
         \tram-GCM & \code{"gcm.test"} & \code{GeneralisedCovarianceMeasure::gcm.test()} \\
         \tram-Wald & \code{"wald"} & \code{multcomp::Chisqtest()} \\
         \tram-COR & \code{"cor.test"} & \code{stats::cor.test()} \\
         $F$-test & \code{"var.test"} & \code{stats::var.test()} \\
         $t$-test & \code{"t.test"} & \code{stats::t.test()} \\
         Combined & \code{"combined"} & \code{stats::\{t,var\}.test()} \\
         HSIC & \code{"HSIC"} & \code{dHSIC::dhsic.test()} \\
         Independence & \code{"independence"} & \code{coin::independence\_test()} \\
         Spearman & \code{"spearman"} & \code{coin::spearman\_test()} \\
         \bottomrule
    \end{tabular}
\end{table}

\paragraph{ICP with shift transformation models}
Shift \tram{s} are implemented in \pkg{tramME} \citep{tamasi2022tramme} but
not directly imported into \pkg{tramicp} and thus do not have an alias.
Nonetheless, models implemented in \pkg{tramME} (such as \code{BoxCoxME()})
can still be used in \pkg{tramicp} via \code{dicp(..., modFUN = "BoxCoxME")}
after loading \pkg{tramME}.

\paragraph{Reducing computational complexity}
The computational complexity of ICP scales exponentially in the number of
predictors due to the need of testing all possible subsets.
\citet{peters2016causal}
proposed to reduce the high computational
burden by pre-screening the predictors using a variable selection
algorithm. Pre-screening of covariates in \tram{s} can in principle be
done via $L_1$-penalized likelihood estimation \citep{kook2020regularized}
or nonparametric feature selection methods that can handle discrete and
censored responses.
Given a variable selection procedure with output $S^{\text{VS}}_n$
which guarantees $\lim_{n\to\infty} \Prob(S^{\text{VS}}_n \supseteq
\pa_{C_*}(\rY)) \geq 1 - \alpha$ at level $\alpha \in (0, 1)$, we can
run \tramicp{} with the potentially reduced set of covariates
$S^{\text{VS}}_n$ at level $\alpha$ and maintain the coverage guarantee
of ICP at level $2 \alpha$ \citep[][Section 3.4]{peters2016causal}.

While pre-screening simplifies the application of \tramicp{} in
practice, it is difficult to ensure that the screening property holds.
Although $L_1$-penalized maximum likelihood is fast and asymptotically
guaranteed to return the Markov boundary for linear additive Gaussian noise models
\citep{meinshausen2006mb}, this no longer holds true for general linear
additive noise models \citep[][Example~1 in the supplement]{nandy2017mb}
or (linear shift) \tram{s}, since the parametric regression model of the
response on all covariates can be misspecified if a child is included
in the conditioning set.

\section{Proofs and auxiliary results}\label{app:theory}

\subsection{Proofs}\label{app:proofs}

\subsubsection{Proof of Proposition~\ref{prop:id}}\label{proof:id}

\begin{proof}
Let $\trafos{A \times A}$ be a class of shift transformation functions such that
Assumption~\ref{asmp:densities} is satisfied. For such a class,
there exists a function
$\check\h : \calY \times A \to \overline{\RR}$ and $\h_1, \h_2 \in \trafos{A
\times A}$ such that for all $(x^1, x^2) \in A \times A$, $\h_1(\bcd \mid
(x^1, x^2)) = \check\h(\bcd \mid x^1)$ and $\h_2(\bcd \mid (x^1, x^2)) =
\check\h(\bcd\mid x^2)$. Thus, $\check{h} \in \trafos{A}$.
We first show that there exists a joint distribution $\Prob_{(Y, X^1, X^2)}$
for the random variables $(Y, X^1, X^2) \in \calY \times A \times A$, \st the
conditional distributions $\Prob_{\rY \mid X^1}$ and $\Prob_{Y \mid X^2}$
are \tram{s} with transformation function $\check\h$. To that
end, let $\nu$ denote the Lebesgue measure restricted to
$A$ in the case that $A$ is an interval or counting measure on $A$ in the case
that $A$ is countable. Let further $f_X$ denote an arbitrary density with
respect to $\nu$ that satisfies for all $x \in A$ that $f_{X}(x) > 0$.
Let\footnote{We use notations $f_{Y \mid X}(y \mid x)$ and $f_{Y
\mid X=x}(y)$ interchangeably.}
$f_{Y \mid X}$ denote the strictly positive conditional density with respect
to
$\mu$ (see Assumption~\ref{asmp:densities} and the sentences thereafter)
corresponding to the \tram{} $\pZ\circ \check\h$.
By integrating over $A$, we obtain a strictly positive density with respect
to $\mu$;
\[
f_Y(\bcd) \coloneqq \int_A f_{Y \mid X}(\bcd \mid x) f_X(x) \dd\nu(x).
\]

By Bayes' theorem, the function $f_{X \mid Y} : \calX\times\calY\to\RR_+$,
which for all $x \in A$ and $y \in \calY$ is given by
\[
f_{X \mid Y}(x \mid y) \coloneqq \frac{f_{Y \mid X}(y \mid x)f_X(x)}{f_Y(y)}
\]
is another strictly positive conditional density with respect to $\nu$.

Now, let the joint density of $(Y, X^1, X^2)$ with respect to $\mu \otimes
\nu \otimes \nu$ for all $y \in \calY$,
$x^1, x^2 \in A$ be given by
\[
f(y, x^1, x^2) \coloneqq f_X(x^1) f_{Y \mid X}(y \mid x^1) f_{X \mid Y}(x^2
\mid  y) =  \frac{f_X(x^1) f_{Y \mid X}(y \mid x^1)f_{Y \mid X}(y
\mid x^2)f_X(x^2)}{f_Y(y)}.
\]
The marginal density $f_{(Y, X^2)}$ of $(Y, X^2)$ is obtained by integrating
out $X^1$;
\[
f_{(Y, X^2)}(y, x^2) = \frac{f_X(x^2) f_{Y \mid X}(y \mid x^2)}{f_Y(y)}
\int_A f_{Y \mid X}(y
\mid x^1)f_X(x^1) \dd\nu(x^1) = f_X(x^2) f_{Y \mid X}(y \mid x^2)\]
and we see from the factorization that the conditional density of $Y$ given
$X^2$ is $f_{Y \mid X}$ which is a \tram{} with transformation function
$\check h$ by construction. An analogous
argument (with the roles of $X^1$ and $X^2$ exchanged) shows that the same
is true for the distribution of $Y$ given $X^1$.

We now construct two structural causal \tram{s} $C_1$ and $C_2$ in
$\calC(\pZ,\calY,\calX,\trafos{A\times A})$ where the observational
distributions satisfy $\Prob^{C_1}_{(Y,X^1, X^2)} = \Prob^{C_2}_{(Y,X^1, X^2)}
= \Prob_{(Y,X^1, X^2)}$. The causal models are given by
\begin{equation}
C_1:
\begin{cases}
    X^1 \coloneqq N_{X^1}^1 \\
    Y \coloneqq \check\h^{-1}(Z^1 \mid X^1) \\
    X^2 \coloneqq g(Y, N_{X^2}^1),
\end{cases}
\qquad \qquad
C_2:
\begin{cases}
    X^2 \coloneqq N_{X^2}^2 \\
    Y \coloneqq \check\h^{-1}(Z^2 \mid X^2) \\
    X^1 \coloneqq g(Y, N_{X^1}^2),
\end{cases}
\end{equation}
where for all $j \in \{1, 2\}$, $(N_{X^1}^j, N_{X^2}^j, Z^j)$ are independent,
$N_{X^1}^1$ and $N_{X^2}^2$ have density $f_X$ with respect to
$\nu$, $Z^1$ and $Z^2$ have distribution function $\pZ$, $N_{X^2}^1$ and
$N_{X^1}^2$ are uniformly distributed on $(0, 1)$ and $g$ is the generalized
inverse of the conditional CDF corresponding to $f_{X \mid Y}$.

It is immediate from this construction that $\Prob^{C_1}_{(Y,X^1, X^2)} =
\Prob^{C_2}_{(Y,X^1, X^2)} = \Prob_{(Y,X^1, X^2)}$ but also that $\pa_{C_1}(Y) =
\{1\}$ and $\pa_{C_2}(Y) = \{2\}$. This proves the desired result.
\end{proof}

\paragraph{Explicit example for the construction in the proof of Proposition~\ref{prop:id}}
For illustrative purposes, we now
give an example with $\calY = A = \{1, 2, 3\}$ in which we construct a
structural causal \tram{} over three (ordered) categorical random
variables $(Y, X^1, X^2)$, in which the two pairwise marginal conditionals $Y
\mid X^1$ and $Y \mid X^2$ are proportional odds logistic regression models%
; we sample from this model and illustrate the non-identifiability empirically.

\begin{example} \label{ex:counterappendi}
Let $\calY = A = \{1, 2, 3\}$, and
$\basisy = \basisy_{\text{dc}, 3}$
(see Table~\ref{tab:basis})
and
$\calX \coloneqq A \times
A$. Let $\trafosub = \trafolin(\basisx)$ and fix
$\check\h \in \trafos{A} = \trafos{A}^{\text{linear}}(\basisx)$,
\st for all $x \in A$,
\begin{align}
    \check\h(\bcd\mid x) &= \basisy(\bcd)^\top\parm - \basisy(x)^\top\shiftparm,\\
    \parm &= (\log 0.5, \log 1.5, +\infty)^\top,\\
    \shiftparm &= (0, \log 1.4, \log 1.8)^\top.
\end{align}
Now, let $f_{Y \mid X^1}$ and $f_{Y \mid X^2}$ be conditional densities
corresponding to the proportional odds model $\pZ \circ \check\h$.
We then construct
the joint distribution by fixing $X^1$ (and $X^2$) to have a
uniform distribution, that is, $f_{X^1} : j \mapsto 1/3$ and follow
the
argument in the proof above.

We now sample from the joint distribution above and perform
goodness-of-fit tests to illustrate that feature selection based
on goodness-of-fit tests yields the empty set in case of non-subset
identifiability.
In Figure~\ref{fig:id}, we show the relative bias
of a parameter estimate $\hat\eparm$ with ground truth $\eparm$, \ie
$\hat\eparm/\eparm - 1$, when estimating the parameters of the
two proportional odds logistic regressions of $Y \mid X^1$ and $Y \mid
X^2$ for sample size $n = 10^4$ and over $10^3$ repetitions. For all $j \in
\{1, 2\}$, we
refer to the parameters for $Y \mid X^j$ as $\eparm_{j1}, \eparm_{j2}$,
$\beta_{j2}$, and $\beta_{j3}$. In addition, we perform a Pulkstenis--Robinson
goodness-of-fit $\chi^2$-test \citep{pulkstenis2004gof} implemented in
\proglang{R}~package \pkg{generalhoslem} \citep{pkg:generalhoslem} for both
models at a sample size of $10^4$ and 100 repetitions. For $Y \mid X^1$
and $Y \mid X^2$ the test rejects 4 times and 5 times at the 5\%~level,
respectively. Thus, the intersection over accepted sets is empty in most
cases.

\begin{figure}[!t]
\centering
\includegraphics[width=0.6\textwidth]{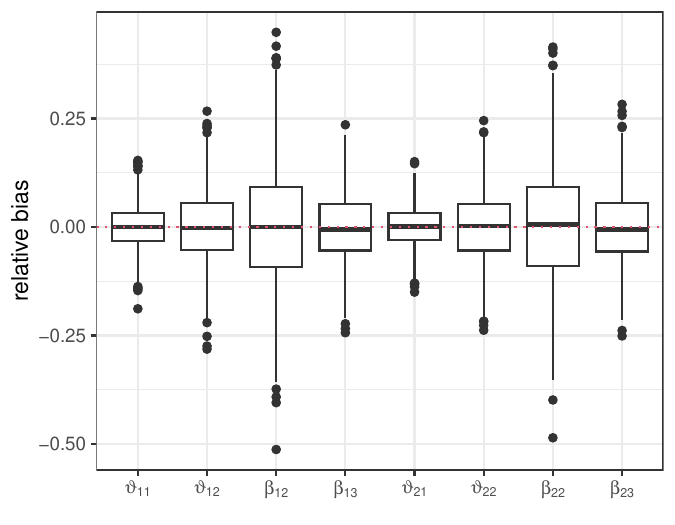}
\caption{%
Relative bias over $10^3$ repetitions
when estimating $\eparm_{ij}$, $\beta_{ij}$ when sampled
according to Example~\ref{ex:counterappendi} with $n = 10^4$.
}
\label{fig:id}
\end{figure}
\end{example}

\subsubsection{Proof of Proposition~\ref{prop:parents}}\label{proof:parents}

\begin{proof}\sloppy
We have $Y \indep \rE \mid \rX^{S_*}$ by the local Markov property because
$E_1, \dots, E_q$ are non-descendants of $Y$ by Setting~\ref{set:env} and
$\rX^{S_*}$ are the parents of $Y$, since $Y \coloneqq h_*(Z \mid X^{S_*})$.
The desired result now follows, because for all $\rx^{S_*} \in \calX^{S_*}$,
the conditional CDF of $Y$ given $\rX^{S_*} = \rx^{S_*}$ is given by
$\pZ(\h_*(\bcd \mid \rx^{S_*}))$.
\end{proof}

\subsubsection{Proof of Proposition~\ref{prop:identification}}\label{proof:identification}

\begin{proof}
Let $\calS \coloneqq \{S \subseteq [d] \mid S \not \supseteq S_*\}$. We now show
that, for all $S \in \calS$, we have $Y \not\indep_{\calG} E \mid S$, where
$\indep_\calG$ denotes $d$-separation in $\calG$. Let $S \in \calS$. By definition
of $\calS$, there exists $j \in [d]$ such that $j \in S_*$ and $j \not\in S$.
Because $S_* \subseteq \ch(\rE)$ and $S_* = \pa(Y)$, by definition, there exists a
directed path $E \rightarrow X^j \rightarrow Y$ in $\calG$. Since $j \not\in S$, we
therefore have that $Y \not\indep_{\calG} E \mid S$. We conclude, by the faithfulness
assumption, that $Y \not\indep E \mid S$ and thus $S$ is not $(\pZ, \trafosub)$-invariant.

Thus, for all $(\pZ, \trafosub)\mbox{-invariant}$ sets $S$ we have $S \supseteq S_*$
and therefore
\[
    S_I \coloneqq \bigcap_{S \subseteq [d] : S \mbox{ is }
    (\pZ, \trafosub)\mbox{-invariant}} S \supseteq
    S_*.
\]
Proposition~\ref{prop:parents} yields that $S_I \subseteq S_*$ and hence the result follows.
\end{proof}

\subsubsection{Proof of Proposition~\ref{prop:residualinv}}\label{proof:residualinv}

\begin{proof}\sloppy
By Lemma~\ref{prop:sresidmeanzero}, we have $\Ex[R(\h_0; \rY, \rX) \mid \rX]
= 0$, $\Ex[R(\h^S; \rY, \rX^S)] = 0$ and $\Ex[\rX^S R(\h^S; \rY, \rX^S)] = 0$.
By the same argument as in the proof of Proposition~\ref{prop:parents}, we
have that $\rY \indep \rE \mid \rX^S$, thus
\begin{align*}
\Ex[\rE R(\h^S; \rY, \rX^S)] &= \Ex[\Ex[\rE R(\h^S; \rY, \rX^S)
\mid \rX^S, \rE]]\\
&= \Ex[\rE \Ex[R(\h^S; \rY, \rX^S) \mid \rX^S, \rE]] =
\Ex[\rE \Ex[R(\h^S; \rY, \rX^S) \mid \rX^S]] = 0.
\end{align*}
Lastly, $\Ex[\Cov[\rE,
R(\h^S; \rY, \rX^S) \mid \rX^S]] = \Ex[\rE R(\h^S; \rY, \rX^S) \mid
\rX^S] - \Ex[\rE\mid\rX^S] \Ex[R(\h^S;\rY,\rX^S)\mid\rX^S] = 0$ follows
directly from the above.
\end{proof}

\subsubsection{Proof of Theorem~\ref{thm:icp}}\label{proof:icp}

\begin{proof}
For $i\in[n]$, we let $\xivec_i\coloneqq\rE_i-\muvec(\rX^S_i)$, $R_{P,i} \coloneqq R(\h_{P}; \rY_i, \rX_i^S)$, $\hat{R}_i \coloneqq
R(\hat\h; \rY_i, \rX_i^S)$,
$\muvec_i = (\muvec_i^1, \dots, \muvec_i^d) \coloneqq \muvec(\rX^S_i)$, $\hat{\muvec}_i = (\hat{\muvec}_i^1, \dots, \hat{\muvec}_i^d) \coloneqq
\hat\muvec(\rX_i^S)$, and define $\Sigma_P \coloneqq \Ex_P[R(h_P; Y, \rX^S)^2
\xivec \xivec^\top]$. We first argue that
\begin{equation}\label{eq:numerator}
    n^{-1/2} \sum_{i=1}^n  (\rL_i - R_{P,i} \xivec_i) = o_{\mathcal{P}}(1).
\end{equation}
To that end, we write
\begin{align*}
n^{-1/2} \sum_{i=1}^n  (\rL_i - R_{P,i} \xivec_i)
=\ &\underbrace{n^{-1/2} \sum_{i=1}^n  (\hat{R}_i - R_{P,i})(
    \muvec_i - \hat{\muvec}_i)}_{\text{Q1}} +
\underbrace{n^{-1/2} \sum_{i=1}^n  (\hat{R}_i -
    R_{P,i})\xivec_i}_{\text{Q2}} \ + \\
&\underbrace{n^{-1/2} \sum_{i=1}^n  R_{P,i}(
    \muvec_i - \hat{\muvec}_i)}_{\text{Q3}}
\end{align*}
By the Cauchy--Schwarz inequality, we have that
\begin{equation}
    \|\text{Q1}\|_{\infty} \leq \left(\frac{1}{n} \sum_{i=1}^n (\hat{R}_i -
    R_{P,i})^2 \max_{j \in [d]} \sum_{i=1}^n \|\muvec_i^j - \hat{\muvec}_i^j\|^2_2  \right)^{1/2} \leq \left(nMW  \right)^{1/2}.
\end{equation}
By Assumption~\ref{gcm:c6}, we conclude that Q1 is $o_{\mathcal{P}}(1)$.

For the Q2 term, we note first that, for all $P \in \mathcal{P}$,
$Y \indep E \mid \rX^S$ and, for all $P \in \mathcal{P}$ and $i \in [n]$,

$R(\hat{h}; Y_i, \rX_i^S)$ is a measurable function of $(Y_i,
\rX_i^S)_{i \in [n]}$. Therefore, for $i \in [n]$, we have
\[
\Ex_{P}[(\hat{R}_i - R_{P,i})\xivec_i \mid (Y_i, \rX_i^S)_{i \in [n]}] =
(\hat{R}_i - R_{P,i}) \Ex_{P}[\xivec_i \mid \rX_i^S] = 0.
\]
By the fact that for $i \neq j$, $\xivec_i$ and $\xivec_j$ are independent given $(Y_i, \rX_i^S)_{i \in [n]}$ and Assumptions~\ref{gcm:c5} and \ref{gcm:c3}, we have
\begin{align*}
\Ex_P[\| \text{Q2} \|_2^2 \mid (Y_i, \rX_i^S)_{i \in [n]}]
&= \frac{1}{n} \Ex_{P}\left[ \left(\sum_{i=1}^n
(\hat{R}_i - R_{P,i}) \xivec_i^\top \right)
\left(\sum_{i=1}^n  (\hat{R}_i - R_{P,i}) \xivec_i\right)
\mid (Y_i, \rX_i^S)_{i \in [n]} \right]\\
&=  \frac{1}{n}  \sum_{i=1}^n (\hat{R}_i -
R_{P,i})^2 \Ex_{P}\left[ \| \xivec_i \|_2^2 \mid \rX_i^S \right] = o_{\mathcal{P}}(1).
\end{align*}
By Lemma~\ref{lem:cond_conv} we conclude that $\text{Q2}= o_{\calP}(1)$.

Finally, for the Q3 term, by a similar argument as for the Q2 term
above, we have, for all $i \in [n]$,
\[
\Ex_{P}[R_{P,i}(\muvec_i - \hat{\muvec}_i) \mid
(\rE_i, \rX_i^S)_{i \in [n]}] = \Ex_{P}[R_{P,i} \mid  \rX_i^S]
(\muvec_i - \hat{\muvec}_i) = 0
\]
by Proposition~\ref{prop:sresidmeanzero}.
By this fact and Assumptions~\ref{gcm:c5} and \ref{gcm:c3}, we obtain that
\[
\Ex_{P}[\| \text{Q3} \|_2^2 \mid (\rE_i, \rX_i^S)_{i \in [n]}] = \frac{1}{n} \sum_{i=1}^n \norm{\muvec_i - \hat{\muvec}_i}_2^2
\Ex_{P}[R_{P,i}^2 \mid  \rX_i^S] = o_{\calP}(1).
\]
By Lemma~\ref{lem:cond_conv} we conclude as above that $\text{Q3} = o_{\calP}(1)$. Thus, \eqref{eq:numerator} holds.

We now argue that
\begin{equation}\label{eq:opnorm}
    \| \hat{\Sigma} - \Sigma_P \|_{\operatorname{op}} = o_{\mathcal{P}}(1).
\end{equation}
We begin by decomposing $\hat\Sigma$ into
\begin{align}
\hat\Sigma =
&\underbrace{
n^{-1}\sum_{i=1}^n R_{P,i}^2\xivec_i\xivec_i^\top
}_{\RN{1}}
+\
\underbrace{
2n^{-1}\sum_{i=1}^n R_{P,i}(\hat{R}_i - R_{P,i})\xivec_i\xivec^\top_i
}_{\RN{2}}
\\
+\
&\underbrace{
2 n^{-1}\sum_{i=1}^n R_{P,i}(\hat{R}_i - R_{P,i})\xivec_i(\muvec_i - \hat{\muvec}_i)^\top
}_{\RN{3}}
+\
\underbrace{
2n^{-1}\sum_{i=1}^n R_{P,i}(\hat{R}_i - R_{P,i})(\muvec_i - \hat{\muvec}_i)\xivec_i^\top
}_{\RN{3}^\top}
\\
+\
&\underbrace{
n^{-1}\sum_{i=1}^n R_{P,i}^2\xivec_i(\muvec_i - \hat{\muvec}_i)^\top
}_{\RN{4}}
+\
\underbrace{
n^{-1}\sum_{i=1}^n R_{P,i}^2(\muvec_i - \hat{\muvec}_i)\xivec_i^\top
}_{\RN{4}^\top}
\\
+\
&\underbrace{
n^{-1}\sum_{i=1}^n (\hat{R}_i - R_{P,i})^2(\muvec_i - \hat{\muvec}_i)\xivec_i^\top
}_{\RN{5}}
+\
\underbrace{
n^{-1}\sum_{i=1}^n (\hat{R}_i - R_{P,i})^2 \xivec_i (\muvec_i - \hat{\muvec}_i)^\top
}_{\RN{5}^\top}
\\
+\
&\underbrace{
n^{-1}\sum_{i=1}^n (\hat{R}_i - R_{P,i})^2 (\muvec_i - \hat{\muvec}_i)(\muvec_i - \hat{\muvec}_i)^\top
}_{\RN{6}}
+\
\underbrace{
2n^{-1}\sum_{i=1}^n R_{P,i}(\hat{R}_i - R_{P,i})(\muvec_i - \hat{\muvec}_i)(\muvec_i - \hat{\muvec}_i)^\top
}_{\RN{7}}
\\
+\
&\underbrace{
n^{-1}\sum_{i=1}^n R_{P,i}^2(\muvec_i - \hat{\muvec}_i)(\muvec_i - \hat{\muvec}_i)^\top
}_{\RN{8}}
+\
\underbrace{
n^{-1}\sum_{i=1}^n (\hat{R}_i - R_{P,i})^2\xivec_i\xivec_i^\top
}_{\RN{9}}
\\
-\
&\underbrace{
(n^{-1}\sum_{i=1}^n \rL_i - R_{P,i}\xivec_i)(n^{-1}\sum_{i=1}^n \rL_i - R_{P,i}\xivec_i)^\top
}_{\RN{10}}
-\
\underbrace{
(n^{-1}\sum_{i=1}^n R_{P,i}\xivec_i)(n^{-1}\sum_{i=1}^n R_{P,i}\xivec_i)^\top
}_{\RN{11}}
\\
-\
&\underbrace{
(n^{-1}\sum_{i=1}^n \rL_i - R_{P,i}\xivec_i)(n^{-1}\sum_{i=1}^n R_{P,i}\xivec_i)^\top
}_{\RN{12}}
-\
\underbrace{
(n^{-1}\sum_{i=1}^n R_{P,i}\xivec_i)(n^{-1}\sum_{i=1}^n \rL_i - R_{P,i}\xivec_i)^\top
}_{\RN{12}^\top},
\end{align}
and handle each term in the following.

We intend to argue that
\[
\RN{1} = \Sigma_P + o_{\mathcal{P}}(1).
\]
By equivalence of finite-dimensional norms, it suffices to show that for all
$j, k \in [q]$
\[
\RN{1}_{jk} - (\Sigma_P)_{jk} = o_{\calP}(1).
\]
This holds by Lemma~19 in \citet{shah2020hardness} since
\[
\sup_{P \in \calP} \Ex[\left(\lvert R_{P,i}\rvert^{2}\right)^{1+\delta/2}\lvert(\xivec_i\xivec_i^\top
)_{jk}\rvert^{1+\delta/2}]
\leq
\sup_{P \in \calP}\Ex[\lvert R_{P,i}\rvert^{2+\delta}\norm{\xivec_i
}_\infty^{2+\delta}]
\leq
\sup_{P \in \calP}\Ex[\lvert R_{P,i}\rvert^{2+\delta}\norm{\xivec_i
}_2^{2+\delta}] < \infty,
\]
by Assumption~\ref{gcm:c2}.

For the $\RN{2}$ term, by sub-multiplicativity of the operator norm and the
Cauchy--Schwarz inequality, we have
\[
\norm{\frac{1}{2} \cdot\RN{2}}_{\operatorname{op}}
\leq
n^{-1}\sum_{i=1}^n\lvert R_{P,i}\rvert \lvert(\hat{R}_i - R_{P,i})\rvert \norm{\xivec_i}_2^2
\leq (n^{-1}\sum_{i=1}^n R_{P,i}^2 \norm{\xivec_i}_2^2)^{1/2}
(n^{-1}\sum_{i=1}^n(\hat{R}_i - R_{P,i})^2\norm{\xivec_i}_2^2)^{1/2}.
\]
Further,
\[
n^{-1}\sum_{i=1}^n R_{P,i}^2 \norm{\xivec_i}_2^2 = \Ex_P[R_{P,i}^2\norm{\xivec_i}_2^2]
+ o_{\mathcal{P}}(1)
\]
again by Lemma~19 in \citet{shah2020hardness} and Assumption~\ref{gcm:c2}.
By Lemma~\ref{lem:cond_conv} and our argument for the Q2 term, we conclude that
\[
n^{-1}\sum_{i=1}^n(\hat{R}_i - R_{P,i})^2\norm{\xivec_i}_2^2 = o_{\mathcal{P}}(1),
\]
thus $\RN{2} = o_{\mathcal{P}}(1)$.

For the $\RN{3}$ term, we have
\begin{align}
\norm{\frac{1}{2} \cdot \RN{3}}_{\operatorname{op}} &\leq
n^{-1} \sum_{i=1}^n \lvert R_{P,i}\rvert \lvert \hat{R}_i - R_{P,i} \rvert \norm{\xivec_i}_2
\norm{\muvec_i - \hat{\muvec}_i}_2
\\
&\leq
(n^{-1}\sum_{i=1}^n R_{P,i}^2 \norm{\xivec_i}_2^2)^{1/2}
(n^{-1}\sum_{i=1}^n (\hat{R}_i - R_{P,i})^2 \norm{\muvec_i - \hat{\muvec}_i}_2^2)^{1/2}.
\end{align}
The first factor has already been handled in $\RN{2}$,
whereas for the second
we have
\[
n^{-1}\sum_{i=1}^n (\hat{R}_i - R_{P,i})^2 \norm{\muvec_i - \hat{\muvec}_i}_2^2
\leq n^{-1}(\sum_{i=1}^n(\hat{R}_i - R_{P,i})^2)(\sum_{i=1}^n
\norm{\muvec_i - \hat{\muvec}_i}_2^2) = o_{\calP}(1),
\]
by Assumption~\ref{gcm:c6}. The same line of argument holds for $\RN{3}^\top$.

For the $\RN{4}$ term, we have
\[
\RN{4} \leq n^{-1}\sum_{i=1}^n R_{P,i}^2 \norm{\xivec_i}_2 \norm{\muvec_i - \hat{\muvec}_i}_2 \leq (n^{-1}\sum_{i=1}^nR_{P,i}^2\norm{\xivec_i}_2^2)^{1/2}
(n^{-1}\sum_{i=1}^n R_{P,i}^2\norm{\muvec_i - \hat{\muvec}_i}_2^2)^{1/2}.
\]
The first factor was shown to be $\Ex_P[R_{P,i}^2\norm{\xivec_i}_2^2]^{1/2}+ o_{\mathcal{P}}(1)$ in the argument for $\RN{2}$,
while Lemma~\ref{lem:cond_conv} and our argument for the Q3 term let us conclude that
\[
    n^{-1}\sum_{i=1}^n R_{P,i}^2\norm{\muvec_i - \hat{\muvec}_i}_2^2 = o_{\mathcal{P}}(1),
\]
hence $\RN{4} = o_{\calP}(1)$.

For the $\RN{5}$ term, we have
\begin{align}
\norm{\RN{5}}_{\operatorname{op}}
&\leq
n^{-1}\sum_{i=1}^n (\hat{R}_i - R_{P,i})^2\norm{\muvec_i - \hat{\muvec}_i}_2
\norm{\xivec_i}_2
\\
&\leq
(n^{-1}\sum_{i=1}^n(\hat{R}_i -R_{P,i})^2\norm{\xivec_i}_2^2)^{1/2}
(n^{-1}\sum_{i=1}^n(\hat{R}_i -R_{P,i})^2\norm{\muvec_i - \hat{\muvec}_i}_2^2)^{1/2}.
\end{align}
The first factor has been handled in the argument for the (II) term and the second
in the argument for $\RN{3}$.

For the $\RN{6}$ term, we have
\[
\norm{\RN{6}}_{\operatorname{op}} \leq n^{-1} \sum_{i=1}^n
(\hat{R}_i - R_{P,i})^2 \norm{\muvec_i - \hat{\muvec}_i}_2^2,
\]
which has been dealt with in the argument for $\RN{3}$.

For the $\RN{7}$ term, we have
\begin{align}
\norm{\frac{1}{2} \cdot \RN{7}}_{\operatorname{op}}
&\leq n^{-1}\sum_{i=1}^n
\lvert \hat{R}_i - R_{P,i}\rvert \lvert R_{P,i} \rvert \norm{\muvec_i - \hat{\muvec}_i}_2^2
\\
&\leq (n^{-1}\sum_{i=1}^n ( \hat{R}_i - R_{P,i})^2 \norm{\muvec_i - \hat{\muvec}_i}_2^2)^{1/2}
(n^{-1}\sum_{i=1}^n R_{P,i}^2\norm{\muvec_i - \hat{\muvec}_i}_2^2)^{1/2}.
\end{align}
The first factor is shown to be $o_{\mathcal{P}}(1)$ as part of the argument for
$\RN{3}$ while the second factor is shown to be $o_{\mathcal{P}}(1)$ in our argument
for $\RN{4}$, thus we conclude that $\RN{7} = o_{\mathcal{P}}(1)$.

For the $\RN{8}$ term, we have
\begin{align}
\norm{\RN{8}}_{\operatorname{op}} \leq
n^{-1}\sum_{i=1}^n R_{P,i}^2\norm{\muvec_i - \hat{\muvec}_i}_2^2,
\end{align}
which was shown to be $o_{\calP}(1)$ in the argument for $\RN{4}$; thus
$\RN{8} = o_{\mathcal{P}}(1)$.

For the $\RN{9}$ term, we have
\begin{align}
\norm{\RN{9}}_{\operatorname{op}} \leq
n^{-1} \sum_{i=1}^n (\hat{R}_i - R_{P,i})^2\norm{\xivec}_2^2,
\end{align}
which was part of the argument for the (II) term.

For the $\RN{10}$ term, we have
\begin{align}
\norm{\RN{10}}_{\operatorname{op}}
= \norm{n^{-1}\sum_{i=1}^n(\rL_i - R_{P,i}\xivec_i)}_2^2
= n^{-1} \norm{n^{-1/2} \sum_{i=1}^n (\rL_i - R_{P,i}\xivec_i)}_2^2,
\end{align}
which is $o_\calP(1)$ by \eqref{eq:numerator}.

For the $\RN{11}$ term, we have
\begin{align}
\norm{\RN{11}}_{\operatorname{op}} = \norm{n^{-1} \sum_{i=1}^n R_{P,i}\xivec_i}_2^2.
\end{align}
By Markov's inequality, for all $\epsilon > 0$,
\begin{align}
\sup_{P\in\calP} \Prob_P\left(\norm{n^{-1}\sum_{i=1}^n R_{P,i}\xivec_i}_2^2
\geq \epsilon\right) \leq
\sup_{P\in\calP} \frac{\Ex_P[\norm{R_{P,i}\xivec_i}_2^2]}{n\epsilon} \to 0,
\end{align}
by Assumption~\ref{gcm:c2}.

For the $\RN{12}$ term (and similarly for the $\RN{12}^\top$
term), we have
\begin{align}
\norm{\RN{12}}_{\operatorname{op}}
=
\norm{n^{-1} \sum_{i=1}^n (\rL_i - R_{P,i}\xivec_i)}_2
\norm{n^{-1} \sum_{i=1}^n R_{P,i}\xivec_i}_2.
\end{align}
The first factor has been dealt with in $\RN{10}$ and the second
in $\RN{11}$. We have thus shown~\eqref{eq:opnorm}.

We now argue that
\begin{equation}\label{eq:sqrtopnorm}
    \| \hat{\Sigma}^{-1/2} - \Sigma_P^{-1/2} \|_{\operatorname{op}}
    = o_{\mathcal{P}}(1).
\end{equation}

To that end, let $\epsilon > 0$ be given and define $C \coloneqq
\sup_{P \in \calP} \norm{\Sigma_P}_{\operatorname{op}}$ which is finite by
Assumption~\ref{gcm:c2}. Further, define
$c \coloneqq \inf_{P \in \calP}
\lambda_{\min}(\Sigma_P)$ which is strictly positive by
Assumption~\ref{gcm:c1}. Now, by Lemma~\ref{lem:matrix}(i)
(first inequality) and (ii)(last inequality), we have
\begin{align*}
&\sup_{P \in \calP} \Prob_P( \| \hat{\Sigma}^{-1/2} - \Sigma_P^{-1/2}
\|_{\operatorname{op}} \geq \epsilon) \leq \sup_{P \in \calP}
\Prob_P( \| \hat{\Sigma}^{-1} - \Sigma_P^{-1} \|_{\operatorname{op}}
\geq C^{-1/2} \epsilon)\\
&\leq \sup_{P \in \calP} \Prob_P( \| \hat{\Sigma}^{-1} - \Sigma_P^{-1}
\|_{\operatorname{op}} \geq C^{-1/2}\epsilon, \| \hat{\Sigma} -
\Sigma_P\|_{\operatorname{op}} < c/2 ) + \sup_{P \in \calP}
\Prob_P(\| \hat{\Sigma} - \Sigma_P\|_{\operatorname{op}} \geq c/2) \\
&\leq \sup_{P \in \calP} \Prob_P( \| \hat{\Sigma} - \Sigma_P
\|_{\operatorname{op}} \geq \frac{1}{2}C^{-1/2}c^2\epsilon) + \sup_{P \in \calP}
\Prob_P(\| \hat{\Sigma} - \Sigma_P\|_{\operatorname{op}} \geq c/2 ).
\end{align*}
The results now follows from \eqref{eq:opnorm}.

We are now ready to prove the result. We write
\begin{align*}
    \rT_n &= \underbrace{\Sigma_P^{-1/2} \left(n^{-1/2}
    \sum_{i=1}^n R_{P,i} \xivec_i \right)}_{\rT_n^{(1)}} +
    \underbrace{(\hat{\Sigma}^{-1/2} - \Sigma_P^{-1/2}) \left(n^{-1/2}
    \sum_{i=1}^n  (\rL_i - R_{P,i} \xivec_i) \right)}_{\rT_n^{(2)}}\\
    &+ \underbrace{(\hat{\Sigma}^{-1/2} - \Sigma_P^{-1/2})
    \left(n^{-1/2} \sum_{i=1}^n  R_{P,i} \xivec_i \right)}_{\rT_n^{(3)}} +
    \underbrace{\Sigma_P^{-1/2} \left(n^{-1/2}
    \sum_{i=1}^n  (\rL_i - R_{P,i} \xivec_i) \right)}_{\rT_n^{(4)}}.
\end{align*}
$\rT_n^{(1)}$ converges uniformly over $\calP$ to a standard $d$-variate
Gaussian distribution by Lemma~\ref{lem:uniformlyapunov} and
Assumption~\ref{gcm:c2}. We have
\[
\norm{\rT_n^{(2)}}_2 \leq \norm{\hat{\Sigma}^{-1/2} -
\Sigma_P^{-1/2}}_{\operatorname{op}} \norm{n^{-1/2}
\sum_{i=1}^n  (\rL_i - R_{P,i} \xivec_i)}_2 = o_{\calP}(1)
\]
by \eqref{eq:numerator} and \eqref{eq:sqrtopnorm}. For any $M > 0$,
by Markov's inequality
\[
\sup_{P \in \calP} \Prob_P\left(\norm{n^{-1/2} \sum_{i=1}^n
R_{P,i} \xivec_i}_2 \geq M \right) \leq M^{-1} \sup_{P \in
\calP} \Ex_P \left[ \norm{R_{P, i} \xivec_i}_2^2 \right] \leq M^{-1} C
\]
so $n^{-1/2} \sum_{i=1}^n  R_{P,i} \xivec_i$ is bounded in probability
uniformly over $\calP$. Hence,
\[
\norm{\rT_n^{(3)}}_2 \leq \norm{\hat{\Sigma}^{-1/2} -
\Sigma_P^{-1/2}}_{\operatorname{op}} \norm{n^{-1/2}
\sum_{i=1}^n   R_{P,i} \xivec_i}_2 = o_{\calP}(1)
\]
by \eqref{eq:sqrtopnorm}. Similarly,
\[
\norm{\rT_n^{(4)}}_2 \leq \norm{\Sigma_P^{-1/2}}_{\operatorname{op}}
\norm{n^{-1/2} \sum_{i=1}^n  (\rL_i - R_{P,i} \xivec_i)}_2 =
\lambda_{\min}(\Sigma_P)^{-1/2} \norm{n^{-1/2} \sum_{i=1}^n
(\rL_i - R_{P,i} \xivec_i)}_2= o_{\calP}(1)
\]
by \eqref{eq:numerator} and Assumption~\ref{gcm:c1}. Combining the above
with a uniform version of Slutsky's theorem
\citep[][Theorem 6.3]{bengs2019uniform} we have the desired result.
\end{proof}

\subsubsection{Proof of Proposition~\ref{prop:wald}}\label{proof:wald}

\begin{proof}
First, $\gammavec^S = 0$ implies that $S$ is $(\pZ,\trafos{\calX
}^{\text{Wald}}(\basisy))$-invariant, because the canonical conditional CDF
of $Y$ conditional on $\rX^S$ and $\rE$ can, for $\Prob_{(\rX^S,\rE)}$-almost
all $(\rx^S,\evec)$, be written as
\begin{align}
    \pZ(\basisy(\bcd)^\top\thetavec + (\rx^S)^\top\shiftparm^S)
\end{align}
which implies that $(Y \mid \rX^S = \rx^S, \rE = \evec)$ and $(Y \mid
\rX^S = \rx^S)$ are identical.

For the reverse direction, $S$ is $(\pZ,\trafos{\calX}^{\text{Wald}}(
\basisy))$-invariant implies that for $\Prob_{(\rX^S,\rE)}$-almost all
$(\rx^S,\evec)$, $(\rY \mid \rX^S = \rx^S, \rE = \evec)$ is identical
to $(\rY \mid \rX^S = \rx^S)$ and the existence of an invariant
transformation function $\h^S \in\trafos{\calX}^{\text{Wald}}(\basisy)$
given by $\h^S : (\ry, \rx^S) \mapsto \basisy(\ry)^\top\tilde\thetavec -
(\rx^S)^\top\tilde\shiftparm^S$. Then, let $\h_{\parm^S} \in \trafowaldS(
\basisy)$ be such that $\parm^S = (\tilde\thetavec, \tilde\shiftparm^S, 0)$.
We then have for $\Prob_{(\rX^S,\rE)}$-almost all $(\rx^S,\evec)$,
$\h^S(\bcd\mid\rx^S) = \h_{\parm^S}(\bcd\mid \rx^S,\evec)$, which
concludes the proof.
\end{proof}

\subsubsection{Proof of Proposition~\ref{prop:empty}}\label{proof:empty}

\begin{proof}
For (i): Since $\rE \in \pa(Y)$, there exists no $S \subseteq [d]$ such that $Y$
is $d$-separated from $\rE$ given $\rX^S$ in $\calG_*$. Then, by faithfulness,
there exists no $S \subseteq [d]$ such that $Y \indep \rE \mid \rX^S$.
Therefore, there exists no $(\pZ,\trafosub)$-invariant set and
\[
S_I \coloneqq \bigcap_{S \subseteq [d] : S \text{ is }
(\pZ,\trafosub)\text{-invariant}} S =
\bigcap_{S \in \emptyset} S
= \emptyset,
\]
where the last equality follows by convention. For (ii): By the assumption of
oracle tests, we thus have
\[
\Prob(S_n \subseteq \pa_{C_*}(Y)) = \Prob(\cap_{S \subseteq [d] :
p_{S,n}(\calD_n) > \alpha} S \subseteq \pa_{C_*}(Y)) = 1,
\]
where the last equality follows because $S_n$ is almost surely empty.
\end{proof}

\subsubsection{Proof of Proposition~\ref{prop:rinvcens}}\label{proof:rinvcens}

\begin{proof}
Let $\nu \coloneqq \mu\otimes\mu$ (see Setting~\ref{set:cens}).
For all $\h \in\trafosubS$ and censored outcomes $(\delta_L, \delta_I, l, u,
\rx^S) \in \{(0,1), (1, 0), (0,0)\} \times \calY^2 \times \calX^S$, the score
residual is given by
\begin{align}
\tilde R(\h; \delta_L, \delta_I, l, u, \rx^S) =
\begin{cases}
\frac{\dZ(\h(l\mid\rx^S))}{\pZ(\h(l\mid\rx^S))} & \text{if } \delta_L = 1, \delta_I = 0, \\
\frac{\dZ(\h(u\mid\rx^S)) - \dZ(\h(l\mid\rx^S))}{
\pZ(\h(u\mid\rx^S)) - \pZ(\h(l\mid\rx^S))} & \text{if } \delta_L = 0, \delta_I = 1, \\
\frac{-\dZ(\h(u\mid\rx^S))}{1-\pZ(\h(u\mid\rx^S))} & \text{if } \delta_L = 0, \delta_I = 0.
\end{cases}
\end{align}
Then, for $\Prob_{\rX^S}$-almost all $\rx^S$, we have for the expected score
residual:
\begin{align}
&\Ex[\tilde R(\h^S; \delta_L, \delta_I, L, U, \rX^S) \mid \rX^S = \rx^S] \\
=
&\int_{\calY^2} p_{(L,U) \mid \rX^{S}}(l,u \mid \rx^{S}) \pZ(\h^S(l\mid\rx^S)) \frac{\dZ(\h^S(l\mid\rx^S))}{
\pZ(\h^S(l\mid\rx^S))} \dd\nu(l,u) \ +\\
&\int_{\calY^2} p_{(L,U) \mid \rX^{S}}(l,u \mid \rx^{S}) (\pZ(\h^S(u\mid\rx^S)) - \pZ(\h^S(l\mid\rx^S)))
\ \times \\
&\quad\quad \frac{\dZ(\h^S(u\mid\rx^S)) - \dZ(\h^S(l\mid\rx^S))}{\pZ(\h^S(u\mid\rx^S))-\pZ(\h^S(l\mid\rx^S))}
\dd\nu(l,u) \ +\\
&\int_{\calY^2} p_{(L,U) \mid \rX^{S}}(l,u \mid \rx^{S}) (1 - \pZ(\h^S(u\mid\rx^S))) \frac{-\dZ(\h^S(u\mid\rx^S))}{
1 - \pZ(\h^S(u\mid\rx^S))} \dd\nu(l,u) \\
=
&\int_{\calY^2} p_{(L,U) \mid \rX^{S}}(l,u \mid \rx^{S}) \dZ(\h^S(l\mid\rx^S)) \dd\nu(l,u) \ +\\
&\int_{\calY^2} p_{(L,U) \mid \rX^{S}}(l,u \mid \rx^{S})
(\dZ(\h^S(u\mid\rx^S)) - \dZ(\h^S(l\mid\rx^S))) \dd\nu(l,u) \ -\\
&\int_{\calY^2} p_{(L,U) \mid \rX^{S}}(l,u \mid \rx^{S}) \dZ(\h^S(u\mid\rx^S)) \dd\nu(l,u) = 0,
\end{align}
where in the first equality we have used the expression for the conditional
density computed below
Setting~\ref{set:cens}.
The result now follows analogously to the proof of Proposition~\ref{prop:residualinv}
which is given in Appendix~\ref{proof:residualinv}.
\end{proof}

\subsection{Auxiliary results}\label{app:lemmata}

Our invariance tests (Definition~\ref{def:traminv} in Section~\ref{sec:methods})
are based on score residuals and use the following property.
\begin{lemma}\label{prop:sresidmeanzero}\sloppy
Let $\calY$, $\calX$, $\pZ\in\calZ$, $\trafosub\subseteq\trafoall$ be as in
Definition~\ref{def:tram}. Impose Assumptions~\ref{asmp:densities}
and~\ref{asmp:closure}. Let $\rX \in \calX$ follow $\Prob_\rX$ and let $(\pZ,\h_0)$,
$\h_0 \in \trafosub$, be a \tram{} such that for $\Prob_\rX$-almost all
$\rx$, $(\rY \mid \rX = \rx)$ has CDF $\pZ(\h_0(\bcd\mid\rx))$. Assume that
for $\Prob_\rX$-almost all $\rx$,
\begin{equation}\label{eq:interch}
    \int_{\calY} \frac{\partial}{\partial\alpha}\dYx(\upsilon;\h_0+\alpha)
    \rvert_{\alpha=0} \dd\mu(\upsilon) =
     \frac{\partial}{\partial\alpha}\int_{\calY}\dYx(\upsilon;\h_0+\alpha)
     \dd\mu(\upsilon)\rvert_{\alpha=0},
\end{equation}
where $\alpha \in \RR$. Then,
if $\Ex[\lvert R(\h_0; Y, \rX) \rvert] < \infty$,
we have $\Ex[R(\h_0; \rY, \rX) \mid \rX] = 0$,
and hence $\Ex[R(\h_0; \rY, \rX)] = 0$ and $\Ex[\rX R(\h_0; \rY, \rX)] = 0$.
\end{lemma}
\begin{proof}\sloppy
Recall the definition of $\mu$ given beneath Assumption~\ref{asmp:densities}.
For the expected score residual, we have for $\Prob_\rX$-almost all $\rx$,
\begin{align}
    &\Ex[R(\h_0; \rY, \rX) \mid \rX = \rx] \\
    &\stackrel{\text{def.}}{=} \Ex\left[\frac{\partial}{\partial\alpha}
        \ell(\h_0 + \alpha; \rY, \rX)
        \rvert_{\alpha=0} \mid \rX = \rx\right] \\
    &\stackrel{\text{def.}}{=} \int_{\calY}
    \dYx(\upsilon; \h_0)
        \frac{\partial}{\partial\alpha}
        \log \dYx(\upsilon; \h_0 + \alpha) \big\rvert_{\alpha=0}
        \dd\mu(\upsilon)\\
    &= \int_{\calY} \dYx(\upsilon; \h_0)
        \frac{\frac{\partial}{\partial\alpha}\dYx(\upsilon;
        \h_0 + \alpha) \rvert_{\alpha=0}}{
        \dYx(\upsilon; \h_0)}
         \dd\mu(\upsilon)\\
    &= \int_{\calY}
        \frac{\partial}{\partial\alpha}\dYx(\upsilon;
        \h_0 + \alpha) \big\rvert_{\alpha=0}
         \dd\mu(\upsilon)\\
    &= \frac{\partial}{\partial\alpha} \int_{\calY}
        \dYx(\upsilon; \h_0 + \alpha)
         \dd\mu(\upsilon)\big\rvert_{\alpha=0} = 0.
\end{align}
Thus, $\Ex[R(\h_{\parm_0}; \rY, \rX)] = \Ex[\Ex[R(\h_{\parm_0}; \rY, \rX)
\mid \rX]] = 0$ and $\Ex[\rX R(\h_{\parm_0}; \rY, \rX)] =  \Ex[\rX \Ex[
R(\h_{\parm_0}; \rY, \rX) \mid \rX]] = 0$.
\end{proof}
%
One can find regularity conditions on the involved distributions and functions
that ensure the validity of interchangeability of differentiation and
integration, see~\eqref{eq:interch}, similar to
\citet[Theorems~1--3]{hothorn2018most}; see also
Propositions~\ref{prop:residualinv} and~\ref{prop:wald}.

\begin{proposition}\label{prop:cdf}
Let $\calY \subseteq \RR$, $\calX \subseteq \RR^d$. Let $\pZ : \overline{\RR}
\to [0,1]$ be the extended CDF of a continuous distribution and let $\h : \RR
\times \calX \to \overline{\RR}$ s.t. for all $\rx\in\calX$, $\h(\bcd\mid\rx)$
is ERCI on $\calY$. Then, for all $\rx\in\calX$, $\pZ(\h(\bcd\mid\rx))$ is a
CDF.
\end{proposition}
\begin{proof}
Fix $\rx\in \calX$, define $F : \calY \to [0,1]$, $F : y \mapsto \pZ(\h(\ry
\mid\rx))$. It suffices to show that $F$ is a CDF, that is,
\begin{enumerate}[label=(\roman*{})]
    \item for all $\ry\in\RR$, $\lim_{\upsilon\to\ry^+} F(\upsilon) = F(\ry)$,
    \item for all $\ry, \ry' \in\RR$ s.t. $\ry \leq \ry'$, $F(\ry) \leq F(\ry')$,
    \item $\lim_{\upsilon\to\infty} F(\upsilon) = 1$,
    \item $\lim_{\upsilon\to-\infty} F(\upsilon) = 0$.
\end{enumerate}
Property (i) follows from continuity of $\pZ\rvert_\RR$ and because $\h$ is ERCI
on $\calY$, \ie $\forall\ry\in\RR, \lim_{\upsilon\to\ry^+} F(\upsilon) =
\pZ(\lim_{\upsilon\to\ry^+}\h(\upsilon\mid\rx))=\pZ(\h(\ry\mid\rx))$.
Property (ii) follows directly because $\h$ is ERCI on $\calY$ and
$\pZ$ is strictly increasing. Properties (iii) and (iv) follow from
$\lim_{\upsilon\to-\infty} \h(\upsilon\mid\rx) = - \infty$ and
$\lim_{\upsilon\to\infty} \h(\upsilon\mid\rx) = \infty$.
\end{proof}

\begin{lemma}\label{lem:matrix}
Let $A, B \in \RR^{d \times d}$ be symmetric and positive semidefinite.
Assume further that $\lambda_{\min}(A) \geq c > 0$. Then
\begin{enumerate}[label=(\roman*{})]
    \item $\norm{A^{1/2} - B^{1/2}}_{\operatorname{op}} \leq c^{-1/2}
        \norm{A - B}_{\operatorname{op}}$.
    \item If $\norm{A - B}_{\operatorname{op}} \leq c/2$ then $B$ is
        invertible and $\norm{A^{-1} - B^{-1}}_{\operatorname{op}}
        \leq 2c^{-2} \norm{A - B}_{\operatorname{op}}$.
\end{enumerate}
\end{lemma}
\begin{proof}
\begin{enumerate}[label=(\roman*{})]
    \item Define $C\coloneqq A^{1/2} - B^{1/2}$ and let $x \in \RR^d$ denote a
    unit eigenvector corresponding to the largest absolute eigenvalue
    $\norm{C}_{\operatorname{op}}$ of $C$. Then, since $A-B = A^{1/2}C +
    CA^{1/2} - C(A^{1/2}-B^{1/2})$, we have
    \begin{align*}
    \norm{A-B}_{\operatorname{op}} &\geq |x^\top (A - B) x| =  |x^\top
    (A^{1/2}C + CA^{1/2} - C(A^{1/2}-B^{1/2})) x|\\
    &=  |x^\top A^{1/2} \norm{C}_{\operatorname{op}} x  + x^\top \norm{C}_{\operatorname{op}}  A^{1/2} x - x^\top \norm{C}_{\operatorname{op}}(A^{1/2}-B^{1/2})x|\\
    &=  \norm{C}_{\operatorname{op}} (x^\top  A^{1/2} x + x^\top B^{1/2} x)
    \geq \norm{C}_{\operatorname{op}} \lambda_{\min}(A^{1/2}) \geq
    \norm{C}_{\operatorname{op}} c^{1/2}.
    \end{align*}
    Rearranging proves the desired result.
    \item By Weyl's inequality
    we have
    \[
    \lambda_{\min}(A) - \lambda_{\min}(B) \leq \norm{A - B}_{\operatorname{op}}.
    \]
    By our assumptions, this implies
    \[
    \lambda_{\min}(B) \geq   \lambda_{\min}(A) -
    \norm{A - B}_{\operatorname{op}} \geq c/2,
    \]
    so $B$ is in fact invertible. We can now write by sub-multiplicativity
    of the operator norm and the triangle inequality
    \begin{align*}
        &\norm{A^{-1} - B^{-1}}_{\operatorname{op}} = \norm{A^{-1}(B -
        A)B^{-1}}_{\operatorname{op}} \leq \norm{A^{-1}}_{\operatorname{op}}
        \norm{A - B}_{\operatorname{op}}\norm{B^{-1}}_{\operatorname{op}}\\
        &= \norm{A^{-1}}_{\operatorname{op}}\norm{A - B}_{\operatorname{op}}
        \norm{A^{-1} - A^{-1} + B^{-1}}_{\operatorname{op}}\\
        &\leq \norm{A^{-1}}_{\operatorname{op}}\norm{A - B}_{\operatorname{op}}(
        \norm{A^{-1}}_{\operatorname{op}} + \norm{A^{-1} + B^{-1}}_{\operatorname{op}}).
    \end{align*}
    Rearranging the inequality, we have that
    \[
        (1- \norm{A^{-1}}_{\operatorname{op}}\norm{A - B}_{\operatorname{op}})
        \norm{A^{-1} - B^{-1}}_{\operatorname{op}} \leq \norm{A^{-1}}_{
        \operatorname{op}}^2 \norm{A - B}_{\operatorname{op}}.
    \]
    Since $\norm{A^{-1}}_{\operatorname{op}} = \lambda_{\min}(A)^{-1} \leq c^{-1}$
    by assumption, we conclude that $1- \norm{A^{-1}}_{\operatorname{op}}\norm{A -
    B}_{\operatorname{op}} \geq 1/2$ so we can rearrange the inequality above further
    to get
    \[
        \norm{A^{-1} - B^{-1}}_{\operatorname{op}} \leq \frac{\norm{A^{-1}}_{
        \operatorname{op}}^2 \norm{A - B}_{\operatorname{op}}}{1- \norm{A^{-1}}_{
        \operatorname{op}}\norm{A - B}_{\operatorname{op}}} \leq 2 \norm{A^{-1}}_{
        \operatorname{op}}^2 \norm{A - B}_{\operatorname{op}} \leq 2 c^{-2}
        \norm{A - B}_{\operatorname{op}}
    \]
    as desired.
\end{enumerate}
\end{proof}

\begin{lemma}[Uniform multivariate Lyapunov's Theorem]\label{lem:uniformlyapunov}
Let $\rX_1, \dots, \rX_n$ be i.i.d.\ copies of $\rX \in \RR^d$ with distribution
determined by $\calP$, such that for all $P \in \calP$, $\Ex_P[\rX] = 0$,
$\Var_P[\rX] = \Id$ and
there exists a $\delta > 0$, s.t.\
$\sup_{P\in\calP} \Ex_P[\norm{\rX}_2^{2+\delta}] < \infty$.
Let $\rS_n \coloneqq n^{-1/2} \sum_{i=1}^n \rX_i$. Then,
\begin{align}
    \lim_{n\to\infty}\sup_{P\in\calP}\sup_{\tvec\in\RR^d}
    \lvert \Prob_P(\rS_n \leq \tvec) - \Phi_d(\tvec)
    \rvert =0,
\end{align}
where $\Phi_d$ denotes the $d$-dimensional standard normal CDF.
\end{lemma}
\begin{proof}
For each $n$, let $P_n \in \mathcal{P}$ be such that
\[
\sup_{P\in\calP}\sup_{\tvec\in\RR^d}
\lvert \Prob_P(\rS_n \leq \tvec)  - \Phi_d(\tvec)
\rvert\leq \sup_{\tvec\in\RR^d}
\lvert \Prob_{P_n}(\rS_n \leq \tvec) - \Phi_d(\tvec)
\rvert + n^{-1}.
\]
For $n \in \mathbb{N}$, let $(\mX_{n, i})_{i \in [n]}$ be independent
random variables with the same distribution as $n^{-1/2} \mX$ under
$P_n$. Note that $\sum_{i=1}^n \mX_{n, i}$ has the same distribution
as $\rS_n$ under $P_n$ so if we can show that $\sum_{i=1}^n \mX_{n, i}$
converges in distribution to a standard $d$-variate Gaussian, we will
have shown the desired result. This follows from
\citet[][2.27]{vandervaart2000asymptotic} since for any $\epsilon >0$
\begin{align*}
    &\sum_{i=1}^n \Ex_{P_n}\left[\| \rX_{n, i}\|_2^2 \1_{\{\|\rX_{n, i}\|_{2}
    >  \epsilon\}} \right] =  \Ex_{P_n}\left[\|\rX\|_2^2
    \1_{\{\|\rX\|_{2} > n^{1/2} \epsilon\}}\right]\\
    &\leq \left(\Ex_{P_n}\left[\|\rX\|_2^{2+\delta} \right] \right)^{2/(2+\delta)}
    \Prob_{P_n}\left(\|\rX\|_{2} > n^{1/2} \epsilon\right)^{\delta/(2+\delta)}\\
    &\leq \epsilon^{-\delta} n^{-\delta/2} \Ex_{P_n}
    \left[\|\rX\|_2^{2+\delta} \right] \leq \epsilon^{-\delta}
    n^{-\delta/2} \sup_{P \in \mathcal{P}}
    \Ex_{P}\left[\|\rX\|_2^{2+\delta} \right] \to 0
\end{align*}
as $n \to \infty$ by Hölder's inequality, Markov's inequality and our
assumption that $\sup_{P\in\calP} \Ex_P[\norm{\rX}_2^{2+\delta}] < \infty$.
\end{proof}

The following result is essentially \citet[][Lemma~10]{lundborg2022conditional} that we
include here for completeness.
\begin{lemma}\label{lem:cond_conv}
    Let $X_1, \dots, X_n$ be a sequence of real-valued random variables and $W_1, \dots, W_n$ be a sequence of random variables taking values in a measurable space $\calW$, both sequences with distributions determined by $\calP$. If $\Ex_P[|X_n| \mid W_n] = o_{\calP}(1)$ then $X_n = o_{\calP}(1)$.
\end{lemma}
\begin{proof}
    Let $\epsilon > 0$ be given and for $x, y \in \RR$ we let
    $x \wedge y \coloneqq \min(x, y)$.
    By Markov's inequality,
    \[
        \sup_{P \in \calP} \Prob_P(|X_n| \geq \epsilon) = \sup_{P \in \calP} \Prob_P(|X_n| \wedge \epsilon \geq \epsilon) \leq \sup_{P \in \calP} \frac{\Ex_P[|X_n| \wedge \epsilon]}{\epsilon}.
    \]
    By monotonicity of conditional expectations, we have for each $P \in \calP$ that
    \[
        \Ex_P[|X_n| \wedge \epsilon \mid W_n] \leq \Ex_P[\epsilon \mid W_n] = \epsilon
    \]
    and
    \[
        \Ex_P[|X_n| \wedge \epsilon \mid W_n] \leq \Ex_P[|X_n| \mid W_n].
    \]
    Combining these inequalities, we have
    \[
        \Ex_P[|X_n| \wedge \epsilon \mid W_n] \leq \Ex_P[|X_n| \mid W_n] \wedge \epsilon,
    \]
    thus by the tower property, we obtain that
    \[
        \sup_{P \in \calP} \Ex_P[|X_n| \wedge \epsilon] = \sup_{P \in \calP} \Ex_P[ \Ex_P[ |X_n| \wedge \epsilon \mid W_n]] \leq
        \sup_{P \in \calP}
        \Ex_P[ \Ex_P[|X_n| \mid W_n] \wedge \epsilon].
    \]
    Define $Y_n \coloneqq \Ex_P[|X_n| \mid W_n] \wedge \epsilon$ and let $\delta > 0$ be given. We have
    \[
        \sup_{P \in \calP} \Ex_P[Y_n] \leq \sup_{P \in \calP} \Ex_P[\1_{\{Y_n < \delta/2\}} Y_n] + \sup_{P \in \calP} \Ex_P[\1_{\{Y_n \geq \delta/2\}}Y_n] \leq \frac{\delta}{2} + \epsilon \sup_{P \in \calP} \Prob_P(Y_n \geq \delta/2).
    \]
    By assumption, for any $\eta > 0$, we can choose $N \in \mathbb{N}$ so that for all $n \geq N$, we can make $\sup_{P \in \calP} \Prob_P(\Ex_P[|X_n| \mid W_n] \geq \delta/2) < \eta$. Hence, letting $\eta \coloneqq \frac{\delta}{2\epsilon}$, we have that
    \[
        \sup_{P \in \calP} \Ex_P[Y_n] < \delta
    \]
    for $n$ sufficiently large, which proves the desired result.
\end{proof}

\begin{lemma}\label{lem:unique}
Assume Setting~\ref{set:env}. If $S \subseteq [d]$ is
$(\pZ,\trafosub)$-invariant, then
the invariant transformation function
$\h^S \in \trafosubS$
from Definition~\ref{def:traminv}
is $\Prob_{\rX^S}$-almost surely unique.
\end{lemma}
\begin{proof}
Assume there exists another invariant transformation function $\h \in
\trafosubS$. Then for $\Prob_{\rX^S}$-almost all $\rx^S$,
the CDF of $\rY \mid \rX^S = \rx^S$ can be written as
$\pZ(\h^S(\bcd\mid\rx^S))$ or $\pZ(\h(\bcd\mid\rx^S))$ by
Definition~\ref{def:traminv}. Now, we have
for $\Prob_{\rX^S}$-almost all $\rx^S$
\begin{align}
    \pZ(\h^S(\bcd\mid\rx^S)) = \pZ(\h(\bcd\mid\rx^S)).
\end{align}
By strict monotonicity of $\pZ\rvert_\RR$ and $\pZ(-\infty)
= 0$ and $\pZ(+\infty) = 1$, we have
for $\Prob_{\rX^S}$-almost all $\rx^S$
\begin{align}
    \h^S(\bcd\mid\rx^S) = \h(\bcd\mid\rx^S),
\end{align}
which concludes the proof.
\end{proof}

\bibliographystyle{abbrvnat}
\bibliography{bibliography}

\end{document}